\documentclass[12pt]{article}

\usepackage{iftex}
\usepackage{amsmath,amsfonts,amssymb,amsthm}
\ifPDFTeX
  \usepackage[T1]{fontenc}
  \usepackage{kpfonts}
\else
  \usepackage{kpfonts-otf}
\fi
\usepackage{eurosym}
\usepackage{graphicx}
\usepackage{enumerate}
\usepackage{amsmath}
\usepackage{amsfonts}
\usepackage{amssymb}
\usepackage{amsfonts}
\usepackage{amssymb,amsmath,amsthm,epsfig,graphicx,natbib,subfigure,hyperref}
\usepackage{geometry}
\usepackage{caption}
\usepackage{color}
\usepackage{enumerate}
\usepackage{setspace}
\usepackage[textwidth=30mm]{todonotes}
\usepackage{sgame, tikz}
\usetikzlibrary{arrows.meta}
\usepackage{multirow,array}
\usepackage{rotating}
\usepackage{graphicx}
\usepackage{subfigure}
\usepackage{float}
\usepackage{booktabs}
\usepackage{makecell}
\usepackage{xurl}
\usepackage{threeparttable}
\usepackage{subcaption}
\usepackage{hyperref}
\usepackage{pgfplots}
\usepackage{tabularx}
\usepackage{needspace}

\hypersetup{colorlinks=true,linkcolor=blue,citecolor=blue,urlcolor=blue}
\allowdisplaybreaks
\bibpunct[: ]{(}{)}{;}{a}{}{,}

\newtheorem{lemma}{Lemma}

\newtheorem{corollary}{Corollary}
\newtheorem{assumption}{Assumption}

\newtheorem{proposition}{Proposition}

\newcommand {\be}{\begin{equation}}
\newcommand {\ee}{\end{equation}}

\graphicspath{{figure/}}

\title{Product Returns and the Value of Search Prominence}

\author{
Sanxi Li\thanks{School of Economics, Renmin University
of China, Beijing, 100872, China. \url{sanxi@ruc.edu.cn}}
\and Jun Yu\thanks{School of Economics, Shanghai University
of Finance and Economics; Key Laboratory of Mathematical Economics (SUFE),
Ministry of Education, Shanghai 200433, China. \url{yu.jun@mail.shufe.edu.cn}}
\and Mingsheng Zhang\textsuperscript{*}\thanks{School of Economics, Renmin University of China, Beijing, 100872, China. Corresponding author. E-mail: \url{mingsheng_zhang@yeah.net}}}

\date{}
\hypersetup{hypertexnames=false}

\begin{document}

\setcounter{footnote}{1}
\maketitle

\begin{abstract}

Ordered-search models typically make prominence valuable because the first-listed seller receives more traffic.
We show that product returns can overturn this logic by changing not only the volume but also the composition of demand.
We develop a duopoly model in which consumers inspect products sequentially, learn a coarse match signal through search, discover realized taste only after purchase, and may return the product.
With positive consumer hassle costs for product returns and a general taste distribution, optimal consumer behavior is characterized by pre-search inspection, post-search purchase, and retention cutoffs.
In a low-hassle, full-refund benchmark with uniformly distributed taste values, the non-prominent firm charges a higher price because consumers who reach it are screened by their experience with the prominent product.
Although prominence generates more visits and gross orders, it also exposes the first seller to less-screened, more return-prone demand.
When seller-side return losses are sufficiently large, the prominent firm can earn less than its rival while both firms earn positive profits.
The reversal also persists on local active pricing branches under small positive hassle costs, small departures from full refunds, and finite oligopoly.
Sellers' willingness to pay for prominence therefore depends on how a higher ranking changes both the profits from retained sales and the losses on returned units.

\textbf{Keywords: }Consumer Search; Search Cost; Prominence; Product Returns; Return Costs

\textbf{JEL Classification}: D83; L13; L81

\end{abstract}

\thispagestyle{empty}

\clearpage

\setcounter{page}{1}


\section{Introduction}\label{sec:intro}
Consumers rarely inspect all available products at once.
Online platforms make the order of inspection explicit through ranked search results, recommendation lists, and sponsored placements.
Because moving from one listing to the next is costly, the first product inspected usually receives a traffic advantage.
Standard ordered-search models therefore predict that prominence raises the first seller's profit \citep{arbatskaya2007ordered,armstrong2009prominence,zhou2011ordered}.
This logic also underlies the economic value of sponsored positions in search markets \citep{athey2011position,armstrong2011paying,chen2011paid}.

E-commerce complicates this traffic logic.
Search can reveal prices, descriptions, reviews, and broad suitability, but it often does not reveal realized fit before purchase.
A consumer may know that a product is broadly suitable and still learn only after delivery whether a shirt fits, a cosmetic shade works, or a device suits her use.
Returns therefore become part of the information process: a consumer buys to learn the realized match value and may then reverse an unsatisfactory purchase.
This margin is economically large.
Industry estimates put annual e-commerce returns near \$550 billion and associated firm-side costs near \$50 billion.\footnote{See Reverse Logistics Association, ``The Impact of E-commerce Returns'': \url{https://rla.org/media/article/view?id=1275}.}
At the same time, consumer-side return frictions are low in many online retail markets, either because consumer-protection rules grant withdrawal rights or because platforms and retailers compete on return convenience.\footnote{For China, see the \textit{Consumer Rights Protection Law} and the \textit{Provisional Measures for the Seven-Day No-Reason Return of Goods Purchased Online}: \url{https://www.samr.gov.cn/zw/zfxxgk/fdzdgknr/fgs/art/2023/art_26ca8fe29e184edd899fa0a7a060d935.html} and \url{https://www.samr.gov.cn/zw/zfxxgk/fdzdgknr/fgs/art/2023/art_0a323e046fba43f0b6e9e977f1e8d5fc.html}. For the European Union, see the Consumer Rights Directive (2011/83/EU): \url{https://eur-lex.europa.eu/eli/dir/2011/83/oj/eng}. For examples of platform return convenience in the United States, see Amazon's no-box or no-label return options, \url{https://www.aboutamazon.com/news/operations/free-amazon-returns}; Walmart's free return channels, \url{https://www.walmart.com/cp/returns/1231920}; and Zappos' return policy, \url{https://www.zappos.com/c/shipping-and-returns}.}

Once purchases can be reversed, traffic is not sufficient for profit.
A prominent firm is seen by more consumers and sells to more of them initially.
But when consumers can cheaply return products and sellers bear return losses, the same traffic also creates more units that are bought, tried, and returned.
By contrast, the non-prominent firm receives fewer visitors, but its visitors are screened.
Consumers reach it only after the prominent product has failed to match or has delivered low realized surplus.
Prominence therefore changes both the volume and the composition of demand.
The central question of the paper is whether this screened composition of later traffic can overturn the usual traffic advantage of the prominent position.

Existing research has not answered this question because the two relevant literatures have largely isolated the two sides of the problem.
Models of ordered search and prominence emphasize inspection order, traffic, and price competition, and typically treat a sale as the final allocation.
Models of product returns emphasize refund policies, consumer learning after purchase, and the allocation of return costs, but they usually do not study how the return margin interacts with an exogenously ordered search path.
The missing interaction is important for online markets: the same ranking that increases exposure also determines which consumers are screened out before reaching later sellers.

We study this interaction in a duopoly model of ordered consumer search with product returns.
Consumers first inspect the prominent firm at zero cost and may then inspect the non-prominent firm by paying a search cost.
Search reveals the price and whether the product broadly matches the consumer.
Purchase reveals the realized taste value conditional on a match.
A purchased product can be returned for a refund; in the benchmark pricing analysis, refunds are full, and consumer-side return hassle is taken to the limit as the hassle cost approaches zero from above.
On the firm side, each sale entails a cost, while a returned unit generates a net salvage value. 
Hence, the seller's net loss on a returned unit equals the sale-side cost minus the salvage value.

Our first contribution is to solve the consumer's search, purchase, and return problem with a positive consumer return hassle cost and a general taste distribution.
The optimal rule is characterized by two reservation cutoffs and a retention cutoff.
The pre-search cutoff determines whether another product should be inspected before its match status is known; the post-search cutoff determines whether a matched product should be tried after inspection.
The retention cutoff determines whether a purchased product is worth keeping after the realized value is learned.
This result separates search frictions from return frictions.
The search cost is paid before the consumer observes the non-prominent product's match status, whereas the return hassle remains payoff-relevant at the retention stage.

The second contribution is to characterize the benchmark pricing game and the value of prominence.
For uniformly distributed taste values, the low-hassle full-refund benchmark has a unique interior active-branch pricing solution.
Under additional primitive restrictions that rule out global deviations, the same price pair is a full-game Perfect Bayesian equilibrium.
The non-prominent firm charges a higher price than the prominent firm.
This price ranking reflects screening: consumers who reach the non-prominent firm either have no matched prominent product or have a low realized surplus from it.
The same screening force drives the main profit result.
The profit gap can be decomposed into the prominent firm's profit from retained demand, the non-prominent firm's screened-demand contribution weighted by its search probability, and the prominent firm's extra exposure to seller-side return losses.
When this exposure is sufficiently costly, prominence is detrimental: the prominent firm earns less than the non-prominent firm.

The third contribution is to show how the sale-side cost, salvage value, search cost, and match probability affect firms, consumers, and social welfare.
A lower sale-side cost reduces the resource cost of initial sales.
In the benchmark equilibrium, it raises both firms' profits and social welfare without changing prices, search, return volume, or consumer surplus.
A higher salvage value improves recovery from returned units.
It raises firms' profits, but it also raises equilibrium prices and total return volume, lowers consumer surplus, and raises social welfare under the maintained equilibrium conditions.
Within the maintained full-game equilibrium region, a higher search cost reduces follow-on search and total return volume, and lowers consumer surplus and social welfare. 
A higher match probability also reduces follow-on search, but raises total return volume, consumer surplus, and social welfare.

The model has implications for how sellers value prominence, for the economics of position auctions, and for welfare assessment in markets with product returns.
For sellers, the value of a prominent position depends on the profit from retained sales and the losses on returned units. 
A higher ranking can generate more orders and still reduce profit if enough of the additional orders are returned at a loss. 
For platforms that auction prominent positions, lowering sellers' sale-side costs or raising salvage values can increase sellers' gains from prominence. 
This can raise auction revenue when bids reflect those gains.
For welfare assessment, platforms and policymakers need to look beyond the number of returns. 
Tools that help consumers search or find suitable products can lead them to keep products that better match their needs. 
The value of better matches and savings in search costs can raise social welfare even when return volume increases. 

The central insight extends beyond the low-hassle benchmark.
Small positive consumer hassle costs, small departures from full refunds, and finite oligopoly preserve detrimental prominence locally on the corresponding active pricing branches.
In the stationary infinite-seller limit, consumers no longer make a final comparison after searching all sellers, and non-prominent firms no longer have a retention-rate advantage.

The remainder of the paper is organized as follows.
Section~\ref{sec:literature} discusses the related literature.
Section~\ref{sec:a search model with product returns} sets up the model and characterizes the consumer's search, purchase, and return rule.
Section~\ref{sec:limiting_equilibrium_propositions} studies the benchmark equilibrium, the value of prominence, and welfare.
Section~\ref{sec:extensions} considers small positive consumer return hassle costs, partial refunds, and oligopoly.
Section~\ref{sec:concluding} concludes.

\section{Related Literature}\label{sec:literature}

The paper contributes first to the literature on ordered search and prominence.
In standard ordered-search models, earlier firms benefit from being inspected first.
\cite{arbatskaya2007ordered} shows that ordered search allows earlier firms to charge higher prices and earn more.
\cite{armstrong2009prominence} and \cite{zhou2011ordered} extend the differentiated-product search framework of \cite{wolinsky1986true} and show how prominence can raise profits.
\cite{rhodes2011can} finds that even with low search costs, prominent firms still outperform others. 
Our contribution is to introduce a product return channel into this logic.
Prominence still raises traffic, but traffic is valuable only when the induced purchases are retained rather than returned at a loss.

Empirical evidence also shows that list position affects attention and search behavior.
\cite{baye2009clicks} study a price-comparison site and find that screen location is an important determinant of retailer clicks.
\cite{feenberg2017s} exploit quasi-random ordering in NBER weekly announcement emails and show that papers listed first receive more views, downloads, and citations.
Most closely, \cite{ursu2018power} uses randomized hotel rankings at Expedia and finds that rankings affect which products consumers search, but not purchases conditional on search.
These findings suggest that ranking primarily shapes the set of products consumers inspect.
Our model adds that, when purchases can be returned, this inspection advantage may generate broader but less screened demand for the prominent firm.

A smaller body of theoretical work identifies conditions under which being inspected first can be detrimental to a seller.
\cite{fishman2018search} show that prominence can be detrimental when consumers must incur a cost to return to a previously visited seller and accept its offer.
In our model, ``return'' refers to sending a purchased product back to the seller. 
Consumers can continue searching while holding the product and can keep it without revisiting the seller.
In a model of blind buying, \cite{chen2021consumer} show that prominence can be detrimental when match-value distributions are asymmetric. Our firms are symmetric except for their positions in the search order.
Our setup reveals how search order shapes sellers' exposure to return losses.
Consumers reach the later seller only if the first product fails to match or yields too little realized surplus to stop search. This screening lowers the later seller's return rate and can reverse the profit ranking when losses on returned units are sufficiently high.

The paper also relates to work on product returns.
The broader return-policy literature studies return-policy leniency and the option value of returns \citep{wood2001remote,anderson2009option,bower2012return}, refund and return-policy design \citep{hsiao2012returns,inderst2015refunds,zhang2022signaling}, and reverse-channel or competitive management of consumer returns \citep{shulman2010optimal,shulman2011managing}.
The closest papers connect returns to search.
\cite{PETRIKAITE2018} shows that consumer-borne return costs can make purchase behavior resemble sequential costly search.
\cite{jerath2025consumer} analyze a monopolist's joint pricing and return-policy choice when consumers can conduct costly prepurchase search.
\cite{janssen2024consumer} allow consumers to acquire information either ex ante through search or ex post through purchase and return.
We combine returns with ordered prominence and focus on firm-borne return losses.
This combination makes the private value of prominent placement depend on the screened composition of demand, not only on traffic volume.

Finally, the results speak to the economics of sponsored positions.
Following \cite{varian2007position,armstrong2011paying,athey2011position,chen2011paid,haan2011advertising,anderson2026search}, a seller's willingness to pay for prominence can be interpreted as the incremental profit generated by the position.
In our environment that incremental profit is influenced by product returns.
Product return management therefore affects the value of ranking and advertising slots.

\section{Model and Consumer Behavior}\label{sec:a search model with product returns}
We adapt the ordered-search framework of \cite{armstrong2009prominence} to a market with product returns.
Two firms sell horizontally differentiated products, and consumers inspect them sequentially in an exogenously fixed order.

\subsection{Model Setup}

There are two firms, each selling a horizontally differentiated product.
They are symmetric except for their positions in the search order.
Without loss of generality, firm~1 is prominent and is inspected first.
Each sale requires a unit sale-side cost $c\geq0$, which includes production or procurement costs as well as outbound logistics, packaging, and fulfillment.

The demand side consists of a unit mass of risk-neutral consumers, each of whom wants at most one unit.
Following \cite{wolinsky1986true}, consumer $l$'s utility from product $i=1,2$ is
\begin{equation*}
\tilde{u}_{il}=m_{il}\varepsilon_{il}-p_{i},
\end{equation*}%
where $p_{i}$ is product $i$'s price, $m_{il}$ indicates whether product $i$ matches consumer $l$ ($m_{il}=1$ if it matches and $m_{il}=0$ otherwise), and $\varepsilon_{il}$ is the taste value conditional on a match.
\footnote{Similar multiplicative match specifications appear in \cite{chen2011paid,chen2022experience,anderson2026search}: utility is generated by a match indicator and a conditional taste value.
The closest paper is \cite{chen2022experience}.
However, they suppose that consumers can ascertain $\varepsilon_{il}$ and the price through search and learn $m_{il}$ after purchasing.}
The match indicators $m_{il}$ are independently and identically distributed across consumers and products, with success probability $\alpha \in (0,1)$.
The taste values $\varepsilon_{il}$ are independently and identically distributed across consumers and products according to a common distribution $F$ on $[0,1]$ with positive density.
The outside option for consumers is normalized to zero.

Consumers know the distributions of $m_{il}$ and $\varepsilon_{il}$ ex ante.
Searching product $i$ reveals the price $p_i$ and the match indicator $m_{il}$, but not the taste value $\varepsilon_{il}$.
If the consumer purchases product $i$ after search, she then observes $\varepsilon_{il}$.
Searching the prominent product is free, while searching the non-prominent product costs $s>0$.

This information structure captures a central feature of online retail: search can reveal prices, product descriptions, reviews, and broad suitability, but it need not reveal realized fit before purchase.
The variable $m_{il}$ can be interpreted as a coarse signal that product $i$ broadly suits consumer $l$'s needs, while $\varepsilon_{il}$ measures the realized quality of that fit.
For example, a consumer shopping online for a suit may observe a seller's listing, price, and coarse product information before buying, but learn the fit and comfort only after delivery.
The return option matters because the realized taste value is learned only after purchase.

After purchasing product $i$, a consumer who returns it receives a refund $\tau p_i$ and bears a hassle cost $h\geq 0$.
The benchmark analysis focuses on full refunds, $\tau=1$.\footnote{As discussed in Section \ref{sec:intro}, full refunds are consistent with many e-commerce platform practices and provide a natural benchmark for modern online retail. 
Section \ref{subsec:partial_refund} analyzes partial refunds.}
The firm must accept the return and receives a net salvage value $\eta\in[0,c)$.
The salvage value is the value recovered from a returned unit net of return-side costs, such as reverse logistics, return processing, restocking, and refurbishment.
Thus, when a sold unit is returned, the firm's per-unit net loss is the sale-side cost net of the salvage value, $c-\eta$.
We write
$$
    \rho\equiv c-\eta
$$
when isolating this return-loss channel.

The timing is as follows.
First, firms simultaneously choose nonnegative prices $p_i\geq0$.
Second, each consumer searches the prominent firm~1, observes $(m_{1l},p_1)$, and decides whether to purchase product~1.
If she purchases product~1, she observes $\varepsilon_{1l}$ and decides whether to keep or return it and whether to search the non-prominent firm~2.
If she does not purchase product~1, she decides whether to search firm~2.
After searching firm~2, she observes $(m_{2l},p_2)$ and decides whether to purchase product~2. 
If she purchases it, she observes $\varepsilon_{2l}$ and then makes the final retention and return decisions, taking into account any product~1 she may already hold.

The equilibrium concept is Perfect Bayesian equilibrium with passive beliefs.\footnote{Passive beliefs are standard in the consumer-search literature, including \cite{wolinsky1986true}. 
They require consumers' beliefs about uninspected prices to remain at their equilibrium levels, regardless of the information observed at an inspected firm.}
We first characterize consumers' optimal behavior and then analyze the firms' pricing game.

\subsection{Consumer Behavior}
We suppress consumer-specific subscripts whenever there is no confusion.
For a generic taste draw $\varepsilon\sim F$, define the expected surplus above threshold $z$ by
$$
    \phi_F(z)\equiv E_\varepsilon[\max\{\varepsilon-z,0\}],\qquad z\in\mathbb R,
$$
where $\phi_F$ is strictly decreasing on $(-\infty,1)$.
For $h>0$, define $A(h)$ and $B(h)$ as the unique cutoffs satisfying
$$
    \phi_F(A(h))=h+\frac{s}{\alpha},
    \qquad
    \phi_F(B(h))=h .
$$
The cutoffs have intuitive roles.
$B(h)$ is the post-search purchase cutoff: once a product is known to match, purchasing it is worthwhile only if its opportunity cost is below $B(h)$.
\footnote{The opportunity cost equals the price of the newly purchased product plus the payoff from the alternative the consumer gives up by keeping it. With no product at hand, the alternative is returning product~2 and exiting, which yields payoff $-h$; hence the opportunity cost is $p_2-h$. If the consumer already holds product~1 and has fallback payoff $v\geq -h$, the opportunity cost is $p_2+v$. Before product~2 is searched, the same opportunity cost is evaluated at the equilibrium belief $p_2^*$.}
$A(h)$ is the pre-search cutoff for product~2: searching product~2 is worthwhile only if the analogous opportunity cost is below $A(h)$.
Since $\phi_F(A(h))=h+s/\alpha>h=\phi_F(B(h))$ and $\phi_F$ is strictly decreasing, $A(h)<B(h)$.
Economically, search has a lower cutoff because it is costly and is chosen before product~2's match status is known.

For a consumer with no product at hand, define the value of having the option to search product $2$ as
$$
    Q_0(p_2^*;h)
    \equiv
    \max\left\{
    \alpha\max\{\phi_F(p_2^*-h)-h,0\}-s,\ 0
    \right\}.
$$
If a consumer has purchased product $1$ and observed $x\equiv\varepsilon_1-p_1$, her fallback payoff before searching product $2$ is
$$
    v(x; h)\equiv \max\{x,-h\}.
$$
Her continuation value after purchasing product $1$ is therefore
$$
    Q(x;p_2^*,h)
    \equiv
    v(x;h)+
    \max\left\{
    \alpha\max\{\phi_F(p_2^*+v(x;h))-h,0\}-s,\ 0
    \right\}.
$$
Thus the ex ante value of purchasing a matched product $1$ is
$$
    \mathcal B_1(p_1;p_2^*,h)
    \equiv
    \int Q(\varepsilon_1-p_1;p_2^*,h)\,dF(\varepsilon_1).
$$

We adopt a strict-action convention: at cutoff points where payoffs are equal and the consumer is indifferent, the corresponding no-action choice is taken.

\begin{proposition}{Consumer Search, Purchase, and Return with Positive Hassle} \label{prop: general consumer rule}

For $h>0$ and nonnegative prices, the consumer's optimal rule is characterized by the following cutoffs and continuation values.

\begin{enumerate}[(i)]
    \item \textit{Initial product-$1$ decision.}
    After observing $(m_1,p_1)$, the consumer never purchases product~$1$ if $m_1=0$.
    If $m_1=1$, she purchases product~$1$ if and only if $\mathcal B_1(p_1;p_2^*,h)>Q_0(p_2^*;h)$.
    If she does not purchase product~$1$, she enters the no-product state; if she purchases it, she observes $x=\varepsilon_1-p_1$ and enters the post-product-$1$ state.

    \item \textit{No-product state.}
    The consumer searches product~$2$ if and only if $p_2^*-h<A(h)$.
    If she does not search, she exits.
    After searching product~$2$, she does not purchase it if $m_2=0$.
    If $m_2=1$, she purchases product~$2$ if and only if $p_2-h<B(h)$.
    Conditional on purchase, she keeps product~$2$ if and only if $\varepsilon_2-p_2\geq-h$; otherwise she returns it.

    \item \textit{Post-product-$1$ state.}
    After purchasing product~$1$ and observing $x=\varepsilon_1-p_1$, the consumer returns product~$1$ and enters the no-product state if $x<-h$.
    If $x\geq -h$, she keeps product~$1$ and searches product~$2$ if and only if $x<A(h)-p_2^*$.
    After searching product~$2$ while holding product~$1$, she does not purchase product~$2$ if $m_2=0$.
    If $m_2=1$, she purchases product~$2$ if and only if $x<B(h)-p_2$.
    Conditional on purchase, she keeps product~$2$ and returns product~$1$ if and only if $\varepsilon_2-p_2>x$; otherwise she keeps product~$1$ and returns product~$2$.
\end{enumerate}
\end{proposition}

Proposition~\ref{prop: general consumer rule} is the first step of the analysis.
It shows that product returns do not simply add an outside option after purchase; they reshape the entire dynamic search problem.
The pre-search cutoff $A(h)$ governs whether the consumer is willing to incur the search cost before learning whether product~$2$ matches.
The post-search cutoff $B(h)$ governs whether she is willing to buy a matched product after inspection.
The retention threshold $-h$ governs whether a purchased product is worth keeping once its realized taste value is known.

\cite{PETRIKAITE2018} shows that costly returns can make sequential online shopping resemble costly search: the consumer buys one product first, and the return cost acts like the cost of trying the second alternative.
Our consumer rule preserves this common inspection channel.
Both $s$ and $h$ reduce the expected value of continuing to product~$2$ and can therefore reduce inspection.
At the same time, the two frictions enter the decision problem differently.
The search cost $s$ is paid before product~$2$'s match status is known and is sunk after inspection, whereas the hassle cost $h$ remains relevant for post-inspection purchase and retention because it is borne when an unsuitable product is returned.
This separation is useful for our prominence question: rankings affect which products consumers inspect, while costly returns affect which inspected and purchased products are ultimately retained.

\section{Equilibrium and the Value of Prominence}\label{sec:limiting_equilibrium_propositions}

We now specialize to the uniform taste-value distribution on $[0,1]$ and study a small-hassle full-refund benchmark, obtained by taking $h\to0^+$ in the positive-hassle consumer problem.\footnote{The consumer problem is solved for arbitrary positive $h$ in Proposition~\ref{prop: general consumer rule}. Section~\ref{subsec:positive_hassle_costs} studies the corresponding local pricing branch for small positive hassle costs.}
This benchmark captures online retail environments with small consumer-side return frictions and yields a tractable equilibrium analysis.
This section asks whether the prominent firm earns more than the non-prominent firm when returned products are costly.

\subsection{Benchmark Consumer Rule and Active-Branch Pricing}\label{sec:equilibrium_analysis}
For the uniform distribution, define
$$
    \phi(z)\equiv E_{\varepsilon}\left[\max\{\varepsilon-z, 0\}\right] .
 $$
For any $r>0$ in the range of $\phi$, let $\phi^{-1}(r)$ denote the unique $z<1$ such that $\phi(z)=r$.

Let $a$ be defined by
$$
    \phi(a)=\frac{s}{\alpha},
$$
and write $v_0(x)\equiv\max\{x,0\}$ for the outside-option payoff generated by product~1.
Set $v_0(\varepsilon_1-p_1)=0$ when product~1 is not held.
Under the uniform distribution and in the limit as $h\to0^+$, Proposition~\ref{prop: general consumer rule} yields the following consumer rule.
\begin{corollary}{Consumer Rule in Hassle-less Environment}\label{prop:limiting_consumer_rule}

    In the limit as $h\to0^+$:

    1. product $1$ is bought iff $m_1=1$ and $p_1<1$;

    2. product $2$ is searched iff $v_0(\varepsilon_1-p_1)<a-p_2^*$;

    3. after search, product $2$ is bought iff $m_2=1$ and $p_2+v_0(\varepsilon_1-p_1)<1$.

    Kept products are those with the highest nonnegative realized net surplus.
\end{corollary}

Corollary~\ref{prop:limiting_consumer_rule} is the behavioral foundation for the pricing analysis.
In the small-hassle benchmark, the return option limits the downside from trying a matched product.
The central margin is therefore not whether product~1 can be tried, but whether its realized surplus is high enough to stop further search.
This margin screens the non-prominent firm's demand before that firm is reached: consumers with high realized surplus from product~1 stop, while consumers with low realized surplus continue to product~2.

Following \cite{armstrong2009prominence}, we maintain the standing restriction
\begin{assumption}
    \begin{equation}
        0<s<\frac{\alpha}{8}.
        \label{eq:armstrong_style_search_cost}
    \end{equation}
    Thus $a\in(1/2,1)$.
\end{assumption}

We first solve the interior active-branch pricing problem.
\footnote{Although the first-order system below is obtained from local deviations on the active branch, Appendix \ref{appendix:global_deviation_condition} gives sufficient global-deviation restrictions under which the same price pair is immune to all unilateral deviations, including boundary and inactive deviations. }
For local deviation prices $(p_1,p_2)$ and equilibrium belief $p_2^*$, the interior pricing region is
$$
    c<p_1<H_1(p_1)<1,
    \qquad
    c<p_2<G_2(p_2)<1,
    \qquad
    p_i<a,\quad i=1,2,
$$
where $H_1(p_1)\equiv a+p_1-p_2^*$ and $G_2(p_2)\equiv a+p_2-p_2^*$.
The first chain requires firm~1's local price to exceed cost and keeps the retention cutoff $p_1$ and the search cutoff $H_1(p_1)$ interior and ordered.
The second chain imposes the analogous requirements for firm~2: its local price exceeds cost, and its retention and post-search cutoffs are interior and ordered.
The final inequalities keep both prices below the upper active-branch price boundary.

These inequalities keep the analysis on the active branch.
Matched consumers buy product~1 when $p_1<1$.
Consumers with no product in hand search firm~2, while consumers holding product~1 search firm~2 exactly when $\varepsilon_1-p_1<a-p_2^*$.
In either state, a matched product~2 is purchased, after which the consumer keeps the product with the highest nonnegative realized net surplus.

We compute demand by separating the four match states.
If $m_1=m_2=0$, no product is purchased and neither firm has retained demand.
If $m_1=1$ and $m_2=0$, the consumer buys product $1$ and keeps it if and only if $\varepsilon_1-p_1\geq0$.
Thus firm $1$'s conditional retained demand in this exclusive branch is $1-p_1$.
If $m_1=0$ and $m_2=1$, the consumer searches product $2$ from the no-product state and keeps it if and only if $\varepsilon_2-p_2\geq0$.
Thus firm $2$'s conditional retained demand in this exclusive branch is $1-p_2$.

It remains to compute retained demand in the competing branch, $m_1=m_2=1$.
The consumer first buys the matched prominent product.
After observing $\varepsilon_1$, she searches product $2$ if and only if $\varepsilon_1<H_1(p_1)$.
If $\varepsilon_1>H_1(p_1)$, she keeps product $1$ and never buys product $2$, giving the no-search demand term
$$
    d_1^{n}(p_1)=\Pr\{\varepsilon_1>H_1(p_1)\}=1-H_1(p_1).
$$
If $\varepsilon_1<H_1(p_1)$, she searches and buys the matched product~2.
In this search subregion, product~1 is finally kept if and only if
$$
    \varepsilon_1-p_1>\max\{\varepsilon_2-p_2,0\}.
$$
Interior cutoffs imply that this event has area
$$
    k_1(p_2)
    =
    \int_{p_1}^{H_1(p_1)}
    \Pr\{\varepsilon_2<\varepsilon_1-p_1+p_2\}\,d\varepsilon_1
    =
    \frac{1}{2}\left[(G_2(p_2))^2-p_2^2\right],
$$

Hence firm $1$'s conditional retained demand in the competing branch is
$$
    d_1^{n}(p_1)+k_1(p_2).
$$

Firm~2 has two retained-demand regions inside the competing branch.
First, when $\varepsilon_2>G_2(p_2)$ and $\varepsilon_1<H_1(p_1)$, product $2$'s net payoff exceeds product $1$'s payoff throughout the search subregion.
This rectangle has area
$$
    d_2^{r}(p_1,p_2)=\Pr\{\varepsilon_1<H_1(p_1),\ \varepsilon_2>G_2(p_2)\}=H_1(p_1)\left[1-G_2(p_2)\right].
$$
Second, when $p_2\leq\varepsilon_2<G_2(p_2)$, product $2$ is kept exactly when
$\varepsilon_1<\varepsilon_2+p_1-p_2$.
This lower trapezoid region has area
$$
    k_2(p_1)=
    \int_{p_2}^{G_2(p_2)}
    \Pr\{\varepsilon_1<\varepsilon_2+p_1-p_2\}\,d\varepsilon_2
    =
    \frac{1}{2}\left[(H_1(p_1))^2-p_1^2\right].
$$
Hence firm $2$'s conditional retained demand in the competing branch is
$$
    d_2^{r}(p_1,p_2)+k_2(p_1).
$$
Figure \ref{fig:limiting_retained_demand_regions} summarizes the retained-demand regions in the exclusive and competing branches.
The vertical cutoff $H_1(p_1)$ separates consumers who stop after product $1$ from those who search product $2$.
Within the search region, the diagonal $\varepsilon_2=\varepsilon_1-p_1+p_2$ separates the final kept product when both net taste values are positive.

\begin{figure}[h]
    \centering

    \begin{minipage}[c]{0.39\textwidth}
        \centering

        \subfigure[Exclusive: firm~1 matched.]{
            \resizebox{0.90\linewidth}{!}{%
            \begin{tikzpicture}[xscale=6]
                \coordinate (O) at (0,0);
                \coordinate (Pone) at (0.35,0);
                \coordinate (One) at (1,0);

                \draw[->, >=Stealth, thick] (0,0) -- (1.1,0) node[right] {$\varepsilon_1$};
                \draw[thick] (O) -- ++(0,-0.08) node[below] {$0$};
                \draw[thick] (Pone) -- ++(0,0.10) node[above] {$p_1$};
                \draw[thick] (One) -- ++(0,-0.08) node[below] {$1$};

                \draw[decorate, decoration={brace, amplitude=5pt, mirror}, thick, dashed]
                    (O) -- (Pone) node[midway, below=10pt, align=center]
                    {\small return \\ \small product $1$};
                \draw[decorate, decoration={brace, amplitude=5pt, mirror}, thick]
                    (Pone) -- (One) node[midway, below=10pt, align=center]
                    {\small retained demand \\ \small for firm $1$};
            \end{tikzpicture}}
        }

        \vspace{4mm}

        \subfigure[Exclusive: firm~2 matched.]{
            \resizebox{0.90\linewidth}{!}{%
            \begin{tikzpicture}[xscale=6]
                \coordinate (O) at (0,0);
                \coordinate (Ptwo) at (0.40,0);
                \coordinate (One) at (1,0);

                \draw[->, >=Stealth, thick] (0,0) -- (1.1,0) node[right] {$\varepsilon_2$};
                \draw[thick] (O) -- ++(0,-0.08) node[below] {$0$};
                \draw[thick] (Ptwo) -- ++(0,0.10) node[above] {$p_2$};
                \draw[thick] (One) -- ++(0,-0.08) node[below] {$1$};

                \draw[decorate, decoration={brace, amplitude=5pt, mirror}, thick, dashed]
                    (O) -- (Ptwo) node[midway, below=10pt, align=center]
                    {\small return \\ \small product $2$};
                \draw[decorate, decoration={brace, amplitude=5pt, mirror}, thick]
                    (Ptwo) -- (One) node[midway, below=10pt, align=center]
                    {\small retained demand \\ \small for firm $2$};
            \end{tikzpicture}}
        }

    \end{minipage}
    \hfill
    \begin{minipage}[c]{0.54\textwidth}
        \centering

        \subfigure[Competing branch.]{
            \resizebox{0.92\linewidth}{!}{%
            \begin{tikzpicture}[scale=5]
                \def\Pone{0.30}
                \def\Ptwo{0.28}
                \def\Hone{0.66}
                \def\Gtwo{0.64}

                \draw[->, thick] (0,0) -- (1.15,0) node[right] {$\varepsilon_1$};
                \draw[->, thick] (0,0) -- (0,1.15) node[above] {$\varepsilon_2$};
                \draw[thick] (1,0) node[below] {$1$} -- (1,1);
                \draw[thick] (0,1) node[left] {$1$} -- (1,1);

                \draw[thick] (\Pone,0) -- (\Pone,\Ptwo);
                \draw[thick] (0,\Ptwo) -- (\Pone,\Ptwo);
                \draw[thick] (\Pone,\Ptwo) -- (\Hone,\Gtwo);
                \draw[thick] (\Hone,0) -- (\Hone,1);
                \draw[dashed, thick] (0,\Gtwo) -- (\Hone,\Gtwo);

                \node[below] at (0,0) {$0$};
                \node[below] at (\Pone,0) {$p_1$};
                \node[below] at (\Hone,0) {$H_1(p_1)$};
                \node[left] at (0,\Ptwo) {$p_2$};
                \node[left] at (0,\Gtwo) {$G_2(p_2)$};

                \node[font=\scriptsize, align=center] at (0.14,0.13)
                    {return\\1 \& 2};
                \node[font=\scriptsize, align=center] at (0.22,0.45)
                    {$k_2(p_1)$\\hold 2\\return 1};
                \node[font=\scriptsize, align=center] at (0.35,0.80)
                    {$d_2^{r}$\\hold 2\\return 1};
                \node[font=\scriptsize, align=center] at (0.50,0.27)
                    {$k_1(p_2)$\\hold 1\\return 2};
                \node[font=\scriptsize, align=center] at (0.83,0.50)
                    {$d_1^{n}$\\hold 1\\no search};
            \end{tikzpicture}}
        }

    \end{minipage}

    \caption{Retained-demand Regions by Match State}
    \label{fig:limiting_retained_demand_regions}

    \vspace{0.3em}

    \begin{minipage}{0.95\textwidth}
        \footnotesize
        \textit{Notes:}
        Panels (a) and (b) show the exclusive-match branches: the cutoff $p_i$ separates returned units from retained demand for firm~$i$.
        Panel (c) shows the competing branch when both products match.
        \textbf{$d_1^n$}: the consumer stops after product~1 and keeps it.
        \textbf{$k_1$}: the consumer searches product~2, buys it, but keeps product~1 and returns product~2.
        \textbf{$d_2^r$} and \textbf{$k_2$}: the consumer keeps product~2 and returns product~1.
        The bottom-left  region is the event in which both purchased products are returned.
    \end{minipage}

\end{figure}

Collecting these terms, active-branch demand depends on the belief $p_2^*$ through $H_1$ and $G_2$.
Writing this dependence explicitly, the retained-demand and initial-sales functions are
$$
\begin{aligned}
    D_1(p_1,p_2;p_2^*)
    &=
    \alpha(1-\alpha)(1-p_1)
    +\alpha^2\left[d_1^{n}(p_1)+k_1(p_2)\right],\\
    S_1&=\alpha,
\end{aligned}
$$
and
$$
\begin{aligned}
    D_2(p_1,p_2;p_2^*)
    &=
    \alpha(1-\alpha)(1-p_2)
    +\alpha^2\left[d_2^{r}(p_1,p_2)+k_2(p_1)\right],\\
    S_2(p_1)&=\alpha(1-\alpha)+\alpha^2H_1(p_1).
\end{aligned}
$$
Local profits are
$$
    \pi_i(p_1,p_2;p_2^*)=(p_i-c)D_i(p_1,p_2;p_2^*)-(c-\eta)\left[S_i-D_i(p_1,p_2;p_2^*)\right],
    \qquad i=1,2.
$$

The interior active-branch solution is characterized by the following first-order system.
For firm~1,
\begin{equation}
    p_1^*=
    \frac{1}{2}\left[
    1+\eta-\alpha(a-p_2^*)+\frac{\alpha}{2}\left(a^2-(p_2^*)^2\right)
    \right].
    \label{eq:limiting_FOC1}
\end{equation}
For firm~2, let $h_2\equiv H_1(p_1^*)=a+p_1^*-p_2^*$.
Firm~2's first-order condition is
\begin{equation}
    p_2^*=
    \frac{1-\alpha}{\alpha}\frac{1-2p_2^*+\eta}{h_2}
    +1-a+\eta+
    \frac{h_2^2-(p_1^*)^2}{2h_2}.
    \label{eq:limiting_FOC2}
\end{equation}

The following assumption defines the parameter region in which the first-order system has an interior pricing solution.

\begin{assumption}{Interior Pricing Region.}
    \begin{equation}
        0\leq\eta<\min\{1-a,2a-1\}.
        \label{eq:armstrong_style_salvage_conditions}
    \end{equation}
\end{assumption}

\begin{proposition}{Unique Active-Branch Pricing Solution}
    \label{lem:limiting_armstrong_style}

    Under Assumptions 1 and 2, the first-order system \eqref{eq:limiting_FOC1}--\eqref{eq:limiting_FOC2} has a unique solution in the interior price interval
    \begin{equation}
        1-a<p_i^*<\frac{1+\eta}{2}<a,
        \qquad i=1,2.
        \label{eq:armstrong_style_price_interval}
    \end{equation}
    The price pair is locally optimal for each firm on the active branch and is independent of $c$.
    If, in addition, the sufficient global-deviation restrictions in Appendix \ref{appendix:global_deviation_condition} are maintained, this solution is a full-game Perfect Bayesian pricing equilibrium.
\end{proposition}

Proposition \ref{lem:limiting_armstrong_style} is the pricing-equilibrium object used in the analysis below.
Under Assumptions 1 and 2, it delivers a unique interior active-branch price pair.
When the sufficient global-deviation restrictions in Appendix \ref{appendix:global_deviation_condition} are also imposed, the same price pair is a full-game equilibrium.
Thus, throughout the subsequent equilibrium analysis, equilibrium statements are understood under these maintained sufficient restrictions.\footnote{Technically, Proposition \ref{lem:limiting_armstrong_style} first establishes a local active-branch result: the first-order system has a unique solution and local deviations within the maintained active branch are unprofitable.
Appendix \ref{appendix:global_deviation_condition} supplies sufficient restrictions under which the same price pair is also immune to boundary and inactive deviations.
We do not classify inactive-branch or corner equilibria outside the maintained region.
The excluded cases are ones in which at least one core active-search margin collapses.
First, the believed price of product~2 may be high enough that consumers never inspect the non-prominent firm.
Second, the product~2 purchase boundary may bind, so a matched product~2 is not bought even after inspection.
Third, the product~1 purchase condition may bind or fail, so matched consumers do not buy the prominent product.
Fourth, terminal retention cutoffs may bind, so retained-demand regions are pushed to the boundary.
In all these cases, search, trial, or retention is shut down on a margin that is positive in the active branch.}

The price pair is independent of $c$, so the sale-side cost affects profit levels but not equilibrium prices in this benchmark.
The reason is the discrete pre-purchase match signal, $m_i\in\{0,1\}$, together with the timing of search.
On the maintained branch, firm~1's initial sales are locally fixed because all consumers inspect firm~1 and every matched consumer buys product~1.
Firm~2's initial sales are also locally fixed with respect to its own deviation price: consumers decide whether to search firm~2 using the equilibrium belief $p_2^*$, and every consumer who searches and finds a match buys product~2.
Thus $c$ drops out of the first-order conditions.
\footnote{This price independence is a feature of the benchmark information structure.
If search revealed a continuous match signal rather than $m_i\in\{0,1\}$, a local price deviation could change the post-search purchase probability among consumers who have reached the firm.
Initial sales would then vary with the deviating price, so the sale-side cost $c$ would generally enter the pricing first-order conditions.}
The broader distinction is that $c$ and $\eta$ operate through different margins: $c$ is tied to initial sales, whereas $\eta$ changes the payoff from retained demand and returned units.

\begin{proposition}{Price Ranking}
    \label{prop:limiting_price_comparison}

    The equilibrium prices satisfy
    $$
        p_1^*<p_2^*.
    $$
\end{proposition}

The ranking $p_1^*<p_2^*$ reflects demand screening in ordered search \citep{armstrong2009prominence}.
The prominent firm serves the broad initial pool, while the non-prominent firm is reached only by consumers for whom product~1 has not generated a sufficiently attractive option.
Conditional on reaching firm~2, consumers either have no matched product~1 or have a low realized surplus from product~1.
Firm~2 therefore faces a smaller but more favorably screened continuation pool, which supports a higher price.
The timing of price discovery reinforces this asymmetry.
Consumers observe firm~1's price before deciding whether to continue search, so a change in $p_1$ affects both the product-retention margin and the continuation margin.
Firm~2's actual price is observed only after search, so the mass of consumers who reach firm~2 is determined by the equilibrium belief $p_2^*$ rather than by firm~2's local deviation price.
Firm~2 therefore faces a weaker marginal search response.
Prominence thus has two countervailing effects: it gives firm~1 access to the broad initial pool, but it also exposes firm~1 to the continuation margin, while firm~2 serves a screened pool of consumers for whom its product is relatively attractive.

Propositions \ref{lem:limiting_armstrong_style} and \ref{prop:limiting_price_comparison} verify the non-cost interior cutoff inequalities.
Evaluated at the equilibrium belief $p_2^*$, we have $G_2(p_2^*)=a$ and $H_1(p_1^*)=a+p_1^*-p_2^*$.
Since $p_1^*<p_2^*<a$, it follows that $p_1^*<H_1(p_1^*)<1$ and $p_2^*<G_2(p_2^*)<1$.
The lower price-cost inequalities $c<p_i^*$ are separate from the $c$-independent first-order construction.
For a full-game Perfect Bayesian equilibrium, these inequalities are imposed by the restrictions in Appendix \ref{appendix:global_deviation_condition}, which also rule out global deviations.

For later use, define the product-$2$ search probability at the pricing equilibrium by $q_2\equiv 1-\alpha+\alpha h_2$.

\begin{proposition}{Price Comparative Statics}
    \label{prop:limiting_comparative_statics}

    1. With respect to the salvage value, both posted prices increase, the price gap narrows, and the probability that consumers search the non-prominent firm rises.

    2. The equilibrium prices increase with the search cost.

    3. The equilibrium prices decrease with the match probability.
\end{proposition}

A higher salvage value $\eta$ lowers the firm's net loss on returned units, so both equilibrium prices rise.
It increases the value the seller obtains from each returned unit.
This return-recovery channel matters more for the prominent firm because firm~1 serves the broad consumer pool, whereas firm~2 serves a screened continuation pool.
Thus $p_1^*$ rises more than $p_2^*$, the price gap $p_2^*-p_1^*$ narrows, and the narrower gap raises the product-$2$ search probability $q_2$.
A higher search cost $s$ weakens consumers' willingness to inspect the non-prominent product.
It therefore softens competition and raises both equilibrium prices, consistent with the standard logic of ordered search models \citep{armstrong2009prominence}.
A higher match probability $\alpha$ works in the opposite direction.
It makes products more likely to be viable alternatives, increases the frequency with which consumers compare matched products, and lowers both equilibrium prices.
The profit advantage of prominence depends on these price responses and on how many units each firm sells and takes back.

\subsection{Detrimental Prominence}\label{sec:detrimental_prominence}
We next compare the two firms' equilibrium profits.
Let
$$
    \Gamma\equiv \pi_1^*-\pi_2^*
$$
denote the prominent firm's extra profit.
The first-order conditions imply
$$
    \pi_1^*=\alpha\left[(p_1^*-\eta)^2-\rho\right],
    \qquad
    \pi_2^*=\alpha q_2\left[(p_2^*-\eta)^2-\rho\right],
$$
where $q_2$ is the product-$2$ search probability, with $q_2=1-\alpha+\alpha h_2$ and $h_2=a+p_1^*-p_2^*$.
Hence
\begin{equation}
    \frac{\Gamma}{\alpha}
    =
    (p_1^*-\eta)^2
    -q_2(p_2^*-\eta)^2
    -(1-q_2)(c-\eta).
    \label{eq:profit_gap_decomposition}
\end{equation}
This decomposition is the sufficient statistic for the value of prominence.
The first term is the prominent firm's retained-margin contribution.
The second term is the non-prominent firm's screened-demand contribution, scaled by the probability that consumers search it.
The last term is the prominent firm's extra exposure to the seller-side return loss $\rho=c-\eta$.

For each fixed salvage value $\eta$, define
$$
    \rho^\dagger(\eta)
    \equiv
    \frac{(p_1^*-\eta)^2-q_2(p_2^*-\eta)^2}{1-q_2},
    \qquad
    \overline\rho(\eta)\equiv (p_1^*-\eta)^2 .
$$
The first object is the detrimental-prominence threshold.
The second is the return-loss upper bound used in Appendix \ref{appendix:global_deviation_condition} to obtain a full-game equilibrium; it is also the level at which the prominent firm's equilibrium profit is zero.

\begin{proposition}{Detrimental Prominence}
\label{prop:limiting_detrimental_prominence}

The prominent firm earns less than the non-prominent firm if and only if
$$
    \rho\in(\rho^\dagger(\eta),\overline\rho(\eta)].
$$
\end{proposition}

All consumers inspect the prominent firm first and buy its product if it matches. 
Only consumers who get little or no surplus from the first product continue to the later firm. 
The later firm therefore makes fewer initial sales, but its product has a weak alternative to beat. 
In equilibrium, the later firm keeps a larger share of its initial sales even though it charges a higher price. 
The prominent firm has both more initial sales and a higher return rate, so it loses $c-\eta$ on more units. 
When the loss per return is large enough, those extra losses can outweigh any profit advantage from retained sales.

These results matter for sellers buying visibility and for platforms selling prominent positions. 
When a later position is the relevant alternative, a seller's willingness to pay for the first position depends on the incremental operating profit $\Gamma$ \citep{varian2007position,chen2011paid,anderson2026search}. 
Proposition~\ref{prop:limiting_detrimental_prominence} allows $\Gamma<0$ even when both positions yield positive profits. 
A seller can therefore receive more orders in the first position and still prefer the later one.
A similar concern arises in promotional events. During China's 2024 618 Shopping Festival, some sellers were reluctant to participate because they expected more traffic but worried that returns and other costs would reduce profits.
\footnote{For more details, see \url{https://m.21jingji.com/article/20240606/37106f444fc26f7d0d2e00a1cdd9c685_zaker.html}.} 
While promotions also involve price discounts and other costs, Proposition~\ref{prop:limiting_detrimental_prominence} suggests that return losses may help explain why some sellers are reluctant to participate.

The threshold in Proposition~\ref{prop:limiting_detrimental_prominence} is defined for a fixed $\eta$.
Changing $\eta$ changes both the return loss $\rho=c-\eta$ and the equilibrium objects entering $\rho^\dagger(\eta)$.
The comparative statics below therefore analyze the profit gap directly rather than treating $\rho^\dagger$ as fixed.

\begin{proposition}{Profit-Gap Comparative Statics}\label{prop:limiting_profit_gap_comparative_statics}

    1. The extra profit from prominence, $\Gamma$, is strictly decreasing in $c$:
    $
        \frac{\partial\Gamma}{\partial c}<0.
    $
    However, this negative $c$-effect is weaker when the salvage value $\eta$ is higher:
    $
        \frac{\partial}{\partial\eta}
        \left(\frac{\partial\Gamma}{\partial c}\right)>0.
    $

    2. The profit gap increases with the salvage value:
    $
        \frac{d\Gamma}{d\eta}>0 .
    $
\end{proposition}

The first part of Proposition~\ref{prop:limiting_profit_gap_comparative_statics} follows from the difference in initial sales.
In the benchmark, equilibrium prices and the probability $q_2$ of searching firm~2 are independent of $c$, so an increase in $c$ leaves each firm's initial sales unchanged.
Each firm pays $c$ on every initial sale, including units that are later returned. Firm~1 makes more initial sales, so its profit falls faster as $c$ rises. 
The profit gap therefore falls.

A higher salvage value $\eta$ raises $\Gamma$ through both a direct recovery channel and induced equilibrium responses.
Directly, a higher $\eta$ reduces the net loss on returned units, which benefits the more return-exposed prominent firm.
Indirectly, the equilibrium price gap narrows and the product-$2$ search probability $q_2$ rises, which shrinks the prominent firm's excess exposure $1-q_2$.
The proof combines these channels with the equilibrium price responses, using \eqref{eq:profit_gap_decomposition} and Appendix Lemma~\ref{lem:nonbinding_cgamma}, to obtain $d\Gamma/d\eta>0$.

The positive cross-partial comes from this exposure term.
Because $\partial\Gamma/\partial c=-\alpha(1-q_2)$ and $dq_2/d\eta>0$, a higher salvage value reduces the prominent firm's excess exposure and therefore weakens the negative effect of $c$ on the profit gap.
Viewed from the other direction, a higher sale-side cost makes the $\eta$-induced reduction in excess exposure more valuable for the prominent firm.
Thus the cross-partial captures a local marginal interaction: a higher salvage value lowers the cost sensitivity of prominence, while a higher sale-side cost raises the marginal value of reducing the prominent firm's excess exposure.

The same decomposition explains why the search cost $s$ and the match probability $\alpha$ do not have globally signed effects on $\Gamma$ under the maintained assumptions.
A higher $s$ raises both prices but tends to lower $q_2$.
The lower search probability weakens the non-prominent firm's pricing contribution, which favors the prominent firm, but it also increases the prominent firm's relative exposure to return losses.
A higher $\alpha$ has several offsetting effects.
It expands the set of consumers who can buy from either firm, which scales up the potential value of the prominent position.
It also intensifies price competition and lowers both equilibrium prices, changing the retained-margin terms in \eqref{eq:profit_gap_decomposition}.
Finally, changes in $\alpha$ affect the search probability $q_2$ and hence the relative exposure of the two firms to return losses.
Because these scale, price, and exposure effects need not move $\Gamma$ in the same direction, the maintained assumptions do not deliver a global sign.

When deciding how much to pay for a prominent position, a seller must account for delivery and other costs on every order, including those later returned. 
The salvage value of a returned unit also matters. 
For example, a custom-made garment may be difficult to resell after a consumer sends it back.
\footnote{A June 2024 report on China's 618 Shopping Festival describes an apparel seller whose returned custom-made garment was difficult to resell, as well as sellers facing substantial return-related costs. See \url{https://finance.sina.com.cn/wm/2024-06-27/doc-incaefsi5368934.shtml}.}
To evaluate a placement, the seller can subtract the cost of serving orders from the revenue on units kept and the total salvage value of returned units.

\subsection{Welfare Analysis}\label{sec:welfare_analysis}
Let
$$
    S\equiv S_1+S_2=\alpha(1+q_2),
    \qquad
    R\equiv S-D_1-D_2
$$
denote total initial sales and total return volume, respectively.
Let $CS^*$ denote equilibrium consumer surplus, which is equivalent to ex ante expected utility. 
Let $\Pi^*$ denote total industry profit, and $W^*=CS^*+\Pi^*$ denote social welfare.
On the active branch, consumer surplus is
$$
\begin{aligned}
    CS^*
    &=
    (1-\alpha)\left[\alpha\phi(p_2^*)-s\right] \\
    &\quad
    +\alpha\Bigg\{
    p_1^*\left[\alpha\phi(p_2^*)-s\right]
    +\int_{p_1^*}^{h_2}
    \left[
    \varepsilon_1-p_1^*
    +\alpha\phi(p_2^*+\varepsilon_1-p_1^*)-s
    \right]\,d\varepsilon_1
    +\int_{h_2}^{1}(\varepsilon_1-p_1^*)\,d\varepsilon_1
    \Bigg\}.
\end{aligned}
$$
The profit-accounting identity is
\[
    \Pi^*=\pi_1^*+\pi_2^*
    =
    p_1^*D_1+p_2^*D_2-cS+\eta R.
\]

Table~\ref{tab:comparative_statics} summarizes the signed effects established in Proposition~\ref{prop:limiting_welfare_analysis} and Corollaries~\ref{cor: welfare_search_cost}--\ref{cor: welfare_match_probability}.

\begin{table}[H]
\centering
\begin{threeparttable}

\caption{Equilibrium Effects of the Main Primitive Margins}
\label{tab:comparative_statics}

\footnotesize
\renewcommand{\arraystretch}{1.15}

\begin{tabularx}{\linewidth}{
    @{} l l *{3}{>{\raggedright\arraybackslash}X} l @{}
}
\toprule
Margin & Prices & Search / returns & Firms & Consumers & Welfare \\
\midrule
Lower $c$
    & No change
    & No change
    & Profits increase
    & No change
    & Increases \\
Higher $\eta$
    & Increase
    & Search and returns increase
    & Profits increase
    & Surplus decreases
    & Increases \\
Higher $s$
    & Increase
    & Returns decrease
    & Firm~1 profit increases
    & Surplus decreases
    & Decreases \\
Higher $\alpha$
    & Decrease
    & Search decreases but Returns increase
    & Non-monotone
    & Surplus increases
    & Increases \\
\bottomrule
\end{tabularx}

\begin{tablenotes}[flushleft]
\footnotesize
\item \textit{Notes:} ``Non-monotone'' means that, with other primitives fixed, each firm's profit can rise over one range and fall over another.
The comparisons apply to parameter changes that remain within the full-game equilibrium region in Proposition~\ref{prop:limiting_global_deviation}.
\end{tablenotes}

\end{threeparttable}
\end{table}

\begin{proposition}{Welfare Analysis}
    \label{prop:limiting_welfare_analysis}

    1. A higher sale-side cost $c$ leaves consumer surplus and total return volume unchanged, but lowers both firms' profits and social welfare.

    2. A higher salvage value $\eta$ lowers consumer surplus, raises total return volume, and raises both firms' profits and social welfare.
\end{proposition}

The welfare results distinguish the sale-side cost of initial sales from the salvage value of returned units.
Reducing $c$ lowers the real resource cost of serving initial demand.
Because $c$ does not enter the benchmark pricing first-order conditions, a reduction in $c$ raises profits and welfare without changing prices, search, return volume, or consumer surplus.

By Proposition~\ref{prop:limiting_comparative_statics}, a higher salvage value $\eta$ raises both prices, which lowers the surplus consumers obtain from the products they keep. 
Consumers adjust their search decisions, but those at the search cutoff are indifferent between continuing and stopping. 
Moving this cutoff therefore has no first-order effect on expected consumer surplus. 
The price increases make the overall effect negative.

Total returns rise for two reasons. 
First, each firm has a higher return rate. 
The pricing first-order conditions imply that the share of initial sales consumers keep is $D_i/S_i=p_i^*-\eta$. 
The proof of Proposition~\ref{prop:limiting_comparative_statics} establishes $0<dp_i^*/d\eta<1$, so this share falls for both firms. 
Second, the same proposition shows that the search probability $q_2$ rises. 
This increases firm~2's initial sales, while firm~1 continues to sell to the same mass of matched consumers. 
Total initial sales therefore increase. 
Together with the higher return rates, this raises total return volume.

At the equilibrium prices, firm $i$'s expected profit per initial sale is
$
\frac{\pi_i^*}{S_i}=(p_i^*-\eta)^2-(c-\eta).
$
A higher $\eta$ reduces the loss $c-\eta$ one for one. 
As shown above, $p_i^*-\eta$ also decreases. The resulting fall in $(p_i^*-\eta)^2$ is smaller than the fall in $c-\eta$.
Expected profit per initial sale therefore rises for both firms.
Firm~1 makes the same number of initial sales, so its total profit increases. 
Firm~2 also gains from more initial sales, since its expected profit per sale is positive.

A higher $\eta$ also raises social welfare under the maintained equilibrium conditions. 
The resource-accounting identity helps identify the effects at work,
$$
 W^*=\mathcal M^*-cS+\eta R-sq_2,
$$
where $\mathcal M^*$ is the expected gross value of the product finally kept. Holding $c$, $s$, and $\alpha$ fixed and using $S=\alpha(1+q_2)$ gives
$$
\frac{dW^*}{d\eta}
 =R+\eta\frac{dR}{d\eta}
   +\frac{d\mathcal M^*}{d\eta}
   -(\alpha c+s)\frac{dq_2}{d\eta}.
$$
The term $R$ is the direct gain from recovering more value on each returned unit. The term $\eta\,dR/d\eta$ is the additional recovery from the larger return volume established in Proposition~\ref{prop:limiting_welfare_analysis}. Proposition~\ref{prop:limiting_comparative_statics} also gives $dq_2/d\eta>0$. More searches lead to more initial sales, so fulfillment and search costs rise, as captured by the last term.

The change in $\mathcal M^*$ reflects what consumers ultimately keep. More search creates opportunities to find a better product. The smaller price gap makes consumers' choices between purchased products more closely reflect realized values. Higher prices also raise the private retention cutoffs, which can lead consumers to return products they would otherwise keep.\footnote{Once a product has been purchased, $c$ is sunk. Keeping product~$i$ rather than returning it adds $\varepsilon_i-\eta$ to welfare. The efficient choice is to keep the purchased product with the highest realized value if that value is at least $\eta$, and return all other units.} Although these effects include both gains and costs, the equilibrium conditions imply that their sum is positive. The proof of Proposition~\ref{prop:limiting_welfare_analysis} establishes this by combining the profit and consumer-surplus derivatives. The gain in firms' profits exceeds the loss in consumer surplus.

In 2024, Taobao and Tmall announced several measures to reduce sellers' operating costs. 
These included fee waivers for some sales and the introduction of ``Return Treasure,'' a service intended to lower return-related costs.
\footnote{See \url{https://www.guancha.cn/economy/2024_09_25_749677.shtml?s=zwyxgtjbt}.} 
A reduction in a fee paid on each initial sale is closer to a reduction in $c$, while a saving that occurs only when a product is returned is closer to an increase in $\eta$. 
Both can benefit sellers, but their effects on prices and consumers can differ. 
Besides, fee reductions can transfer income between sellers and the platform, so their welfare effects also depend on changes in real resource costs.

\begin{corollary}{Welfare Effect of Search Costs}\label{cor: welfare_search_cost}

    A higher search cost $s$ lowers consumer surplus and total return volume, but raises the prominent firm's profit.
    Under the full-game equilibrium conditions in Proposition~\ref{prop:limiting_global_deviation}, it also lowers social welfare.
\end{corollary}

A higher search cost lowers consumer surplus directly and through prices. 
Each inspection of firm~2 costs more, and Proposition~\ref{prop:limiting_comparative_statics} shows that both equilibrium prices rise. 
These two effects reduce expected consumer surplus even after consumers adjust their search decisions.

The proof of this corollary also shows that $q_2$ falls with $s$. 
Firm~2 therefore makes fewer initial sales, while firm~1's initial sales remain fixed at $\alpha$. 
The pricing first-order conditions give an equilibrium retention share $D_i/S_i=p_i^*-\eta$. 
Since both equilibrium prices rise while $\eta$ stays fixed, this share increases for both firms. 
Both firms thus have lower return rates. 
Fewer initial sales at firm~2 and lower return rates at both firms reduce total returns.

The higher price also raises firm~1's expected profit per initial sale, $(p_1^*-\eta)^2-(c-\eta)$. 
Since firm~1 makes the same number of initial sales, its total profit increases. 
Firm~2 also earns more per initial sale, but it makes fewer initial sales as $q_2$ falls. 
These two effects work in opposite directions, so the preceding results do not determine the change in firm~2's total profit.

The resource expression $W^*=\mathcal M^*-cS+\eta R-sq_2$ also explains the welfare effect of search costs. 
A higher $s$ reduces visits to firm~2 and, by Proposition~\ref{prop:limiting_comparative_statics}, raises both prices. 
The proof of this corollary shows that fewer consumers keep a product and that the expected gross value $\mathcal M^*$ falls. 
Fewer initial sales reduce the cost $cS$ of supplying products. 
Fewer returns also reduce the total salvage value $\eta R$ when $\eta>0$. 
Total search expenditure $sq_2$ rises despite the fall in the number of searches, because the higher cost per search more than offsets the saving from fewer searches. 
Under the maintained equilibrium conditions, the loss in consumption value and recovery, together with the increase in search expenditure, exceeds the saving on initial-sale costs.

Online platforms offer tools that help consumers find alternative products. 
For example, Amazon Lens lets shoppers use a photo to find similar items.
\footnote{See \url{https://www.aboutamazon.com/news/retail/how-to-use-amazon-lens}} 
Our results suggest that an evaluation of such tools should consider which products consumers ultimately keep as well as the number returned. 
When finding alternatives becomes easier, consumer surplus and social welfare can increase even if more products are returned. 
At the same time, the prominent seller may earn less, giving it a reason to prefer that consumers find alternatives less easily.

\begin{corollary}{Effects of Match Probability}\label{cor: welfare_match_probability}

    A higher match probability raises consumer surplus.
    Under the full-game equilibrium conditions in Proposition~\ref{prop:limiting_global_deviation}, it reduces follow-on search and raises total return volume and social welfare.
\end{corollary}

A higher match probability benefits consumers directly and through lower prices. 
At given prices, a higher $\alpha$ makes each inspected product more likely to match, increasing the expected value of the consumer's options. 
Proposition~\ref{prop:limiting_comparative_statics} also shows that both equilibrium prices fall, so consumers pay less for the product they keep. 
Both effects raise consumer surplus.

The proof of this corollary also shows that both firms make more initial sales. 
For firm~1, this follows from $S_1=\alpha$. 
For firm~2, the proof establishes $d(\alpha q_2)/d\alpha>0$.
The same proof shows that $q_2$ falls. 
Firm~2 makes more initial sales because the higher match probability more than offsets the lower number of visits.
Total returns increase as well. 
The pricing first-order conditions give an equilibrium return rate of $1-(p_i^*-\eta)$ for each firm. 
This rate rises because both prices fall while $\eta$ stays fixed. 
Each firm therefore makes more initial sales and receives returns on a larger share of them.

The resource expression $W^*=\mathcal M^*-cS+\eta R-sq_2$ shows the gains and costs from better matching. 
The proof of this corollary establishes that a higher $\alpha$ increases $\mathcal M^*$, the expected gross value of the product finally kept. 
It also increases total initial sales $S$, raising the cost $cS$ of supplying these units. 
More returns increase the total salvage value $\eta R$ when $\eta>0$. 
At the same time, more consumers stop after the first product, reducing the product-$2$ search probability $q_2$ and total search expenditure $sq_2$. 
Under the maintained full-game equilibrium conditions, the gains in consumption value and recovery, together with the saving in search costs, exceed the extra cost of serving more orders.

Platforms use product recommendations to help consumers find products suited to their needs. 
For example, Taobao's Wenwen shopping assistant asks consumers questions and suggests products based on their answers.
\footnote{See \url{https://www.alibabagroup.com/en-US/document-1738398759789789184}.} 
If recommendations increase the chance that an inspected product matches, Corollary~\ref{cor: welfare_match_probability} predicts a higher total return volume, even though consumers search less. 
The additional returns can accompany higher social welfare because consumers are more likely to find a product worth keeping. 

\section{Extensions}
\label{sec:extensions}

This section considers three extensions: small positive consumer return hassle costs, partial refunds close to full refunds, and oligopoly with random search among non-prominent sellers.
\footnote{The positive-hassle, partial-refund, and finite-oligopoly results are established on local active pricing branches. }

 \subsection{Positive Consumer Hassle Costs}
\label{subsec:positive_hassle_costs}

The consumer problem in Proposition~\ref{prop: general consumer rule} allows arbitrary positive return hassle.
The benchmark pricing game, however, was solved in the right-limit economy $h\to0^+$ in order to obtain a tractable benchmark analysis.
This subsection asks whether the pricing result and the detrimental-prominence region survive when the pricing game itself is perturbed to small positive $h$.
The answer is yes: once the reservation and retention cutoffs are adjusted, the positive-hassle active branch is a continuous perturbation of the benchmark branch.

By Proposition \ref{prop: general consumer rule}, the positive-hassle consumer problem is fully characterized by the reservation cutoffs
$$
    A(h)\equiv \phi^{-1}\left(h+\frac{s}{\alpha}\right),
    \qquad
    B(h)\equiv \phi^{-1}(h).
$$
Throughout this subsection we impose $h+s/\alpha<1/8$, which is equivalent to $A(h)>1/2$. Under this condition the cutoffs are
$$
    A(h)=1-\sqrt{2\left(h+\frac{s}{\alpha}\right)},
    \qquad
    B(h)=1-\sqrt{2h},
$$
with $A(h)<B(h)$.
Relative to the benchmark case, positive hassle costs change the active-branch cutoffs in three ways: the product $2$ search boundary is governed by $A(h)$, the product $2$ purchase boundary is governed by $B(h)$, and the lower retention cutoffs are $p_i-h$ rather than $p_i$.
Consumers search based on the belief $p_2^*$ but decide whether to purchase after observing the actual $p_2$.
Thus, for large deviations, product~2's purchase boundary can become relevant.
The local analysis below keeps attention on a neighborhood of the limiting price pair, where this purchase boundary is slack.
The proof formalizes the positive-$h$ active branch, derives the corresponding demand system, and verifies the local first-order system.

The first-order conditions of the resulting local profit-maximization problem lead to the following local characterization.

\begin{corollary}{Small Positive Hassle Costs}
    \label{prop:positive_h_extension}

    Fix $(s,\eta,\alpha)$ satisfying \eqref{eq:armstrong_style_search_cost} and \eqref{eq:armstrong_style_salvage_conditions}, and recall that $(p_1^*,p_2^*)$ is the benchmark pricing equilibrium from Proposition \ref{lem:limiting_armstrong_style}.
    Then there exist a neighborhood $U$ of $(p_1^*,p_2^*)$ and a number $\bar h>0$ such that, for every $h\in(0,\bar h)$, the positive-$h$ active-branch first-order system has a unique interior solution $(p_1^*(h),p_2^*(h))\in U$, with
    $$
        \lim_{h\rightarrow0^+}p_i^*(h)=p_i^*,\qquad i=1,2.
    $$
    Along this local branch:
    \begin{enumerate}
        \item The benchmark price ranking persists: $p_1^*(h)<p_2^*(h)$.
        The strict price comparative-static signs in Proposition~\ref{prop:limiting_comparative_statics} remain valid for local changes in $s$, $\eta$, and $\alpha$.

        \item The detrimental-prominence region with nonnegative profits is nonempty.
    \end{enumerate}
\end{corollary}

Corollary~\ref{prop:positive_h_extension} establishes that the $h\to0^+$ limiting economy is not a knife-edge case.
The detrimental-prominence result, the price ranking $p_1^*<p_2^*$, and the qualitative comparative statics all survive small positive consumer hassle costs.
The economic logic is unchanged: the prominent firm sells to more consumers and therefore incurs more losses on returned units, regardless of whether the consumer's return decision is governed by a small positive hassle cost or by the zero-hassle limit.
The main role of a positive $h$ is to shift the reservation cutoffs and retention thresholds, but the screening asymmetry persists: consumers who continue search are those with relatively low realized surplus from the prominent product.

\begin{proposition}{Small Positive Hassle Costs: Local Hassle-Cost Comparative Statics}
    \label{prop:positive_h_hassle_comparative_statics}

    For every sufficiently small $h>0$, a marginal increase in the consumer hassle cost raises both firms' prices and the prominent firm's profit, while lowering the product-$2$ search cutoff, the probability that consumers search the non-prominent firm, total returns, and consumer surplus.
    \end{proposition}

The intuition is that consumer hassle costs combine a search-friction effect with a retention-threshold effect.
They make consumers less willing to continue to product~2, which raises both firms' local market power and reduces total return volume.
They also make consumers worse off directly, because returning an unsuitable product becomes more costly.

The positive hassle cost analysis clarifies how the baseline result should be interpreted.
The limit $h\to0^+$ is used for tractability and clean geometry, but the disadvantage of prominence survives small positive consumer return frictions.
Positive hassle costs mainly add boundary cutoffs and shift the active-branch reservation surplus from $a$ to $A(h)$; they do not remove the screening mechanism that gives the non-prominent firm a higher-retention demand pool while the prominent firm continues to serve broader initial demand.

\subsection[Partial Refund (tau<1)]{Partial Refund}
\label{subsec:partial_refund}

The baseline model assumes full refund ($\tau=1$): upon return the consumer receives back the entire price $p_i$.
We now relax this to a partial refund $\tau\in(0,1)$, where a returning consumer receives $\tau p_i$ and bears net loss $(1-\tau)p_i$.
All other elements are as in the benchmark case ($h\to0^+$).
Partial refunds change the active branch in two ways.
First, the consumer's return threshold becomes $\varepsilon_i\geq\tau p_i$, and the search cutoff after product~1 depends on the believed price of firm~2.
Second, the firm keeps the unrefunded share $(1-\tau)p_i$ on returned units, so the profit accounting differs from the full-refund benchmark.
For $\tau$ close to $1$, these changes are a continuous perturbation of the full-refund active branch.

\begin{corollary}{Partial Refund: Continuity for $\tau$ near $1$}
    \label{prop:partial_refund_continuity}

    Suppose the full-refund benchmark case has the pricing equilibrium $(p_1^*,p_2^*)$ characterized in Proposition~\ref{lem:limiting_armstrong_style}.
    Then there exist $\bar\tau\in(0,1)$ and a neighborhood $U$ of $(p_1^*,p_2^*)$ such that, for every $\tau\in(\bar\tau,1]$, the partial-refund active-branch first-order system has a unique local solution $(p_1^\tau,p_2^\tau)\in U$.
    This solution is continuous in $\tau$, satisfies
    \[
        \lim_{\tau\to1^-}(p_1^\tau,p_2^\tau)=(p_1^*,p_2^*),
    \]
    and preserves the benchmark price ranking $p_1^\tau<p_2^\tau$.
    If $c-\eta\in(\rho^\dagger(\eta),(p_1^*-\eta)^2)$ at $\tau=1$, then $\bar\tau$ can be chosen so that this local branch has positive profits and $\Pi_1^\tau<\Pi_2^\tau$ for all $\tau\in(\bar\tau,1]$.
\end{corollary}

Corollary~\ref{prop:partial_refund_continuity} is a local robustness result for an exogenous refund fraction.
A partial refund changes both consumers' search and retention incentives and the seller's receipts from returned units.
For $\tau$ sufficiently close to one, these changes preserve the strict profit reversal on the active branch.

\subsection{Oligopoly with Random Search Order}
\label{sec:oligopoly}

The duopoly model isolates the mechanism most transparently.
Consider an oligopoly with one prominent firm and $n-1$ symmetric non-prominent firms, where $n\geq3$.
Firm~$1$ is searched first at zero cost.
The remaining firms are searched in a uniformly random order, and each additional search costs $s>0$.
All other primitives are as in the benchmark model.
The following proposition records the finite- and infinite-seller implications.
For finite $n$, the result is local to the symmetric interior pricing branch.

\begin{proposition}{Oligopoly with Random Search Order}
    \label{prop:oligopoly_random_search}

    \begin{enumerate}
        \item For any finite number of sellers, the detrimental-prominence mechanism survives locally: there is a nonempty interval of return losses on which the non-prominent firm earns more than the prominent firm while both firms earn nonnegative profits.

        \item In the infinite-seller limit, detrimental prominence cannot arise: whenever the prominent firm has nonnegative stationary profit, it earns at least as much as any individual non-prominent firm.
    \end{enumerate}
\end{proposition}

With finitely many sellers, a consumer who finds no product worth stopping for eventually reaches the last seller. 
She can then keep the purchased product with the highest positive net surplus, if there is one, and return the others. 
A non-prominent firm is inspected only after the products searched earlier have failed to stop search. 
On the local pricing branch, this screening gives the non-prominent firm a higher retention rate than the prominent firm. 
When the loss per return is sufficiently high, the retention advantage can outweigh the prominent firm's greater traffic, even when both firms earn nonnegative profits.

In the stationary infinite-seller limit, consumers can continue searching until a matched product reaches the reservation cutoff. 
There is no last seller, so the final comparison after all sellers have been searched disappears. 
At the stationary pricing solution, the prominent and non-prominent firms charge the same price and earn the same expected profit per inspection. 
The prominent firm is inspected by every consumer, while an individual non-prominent firm is reached with lower probability. 
The prominent firm therefore earns at least as much when profit per inspection is nonnegative.

\section{Conclusion}\label{sec:concluding}

We study how product returns affect the value of search prominence. 
Consumers learn whether a product broadly matches their needs through search, but learn its realized fit only after purchase. 
The prominent firm receives more initial orders. 
The non-prominent firm is reached only by consumers for whom the first product did not match or did not generate enough surplus to stop search. 
This screening gives the non-prominent firm a higher retention rate.

For a positive consumer return hassle cost and a general taste distribution, we characterize the consumer's search, purchase, and return decisions. 
We then study a full-refund benchmark with uniformly distributed taste values as the consumer return hassle cost approaches zero. 
The non-prominent firm charges a higher price, but retains a larger share of its initial orders. 
When the loss on each returned unit is sufficiently high, the prominent firm earns less despite making more initial sales. 

The results have implications for the value of prominent positions and for welfare. 
A seller's willingness to pay for prominence depends on the profit from retained sales and the losses on returned units. 
A higher sale-side cost reduces the value of prominence, while a higher salvage value raises it. 
In the benchmark, a lower sale-side cost raises both firms' profits and social welfare without changing consumer surplus. A higher salvage value also raises firms' profits and social welfare, but lowers consumer surplus.
A higher match probability also raises consumer surplus and social welfare, even though total returns increase.

Two questions remain for further research. 
First, search could reveal a richer signal of product fit before purchase, allowing consumers to learn about fit at both stages. 
Second, a platform could choose its return policy. 
Analyzing that choice would show how the platform's incentives interact with consumer search, sellers' pricing, and the value of prominence.

\section*{Acknowledgments}

We are grateful for valuable comments from Daniel Z. Li, Wei Zhao, Jidong Zhou, and participants at 2024 Conference on Behavioral Experimental and Theoretical Economics, the 7th World Congress of the Game Theory Society, the 2024 Shanghai Microeconomics Workshop, and the 2026 China Meeting on Game Theory and Its Applications.
Daniel Z. Li was a collaborator on this research, but he voluntarily withdrew from the paper for personal reasons. 
Li thanks the financial support by National Natural Science Foundation of China (grant no. 72192801), Beijing Natural Science Foundation (Z220001), and fund for building world-class universities (disciplines) of Renmin University of China; 
Yu thanks the financial support by National Natural Science Foundation of China (grant no. 72273078), and support from the SUFE Theoretical Economics Gaofeng II Discipline Innovation Project (2018110721).

\bibliographystyle{chicagoa}
\IfFileExists{ref.bib}{\bibliography{ref}}{\bibliography{article/ref}}

\appendix
\setcounter{lemma}{0}
\renewcommand{\thelemma}{A.\arabic{lemma}}
\renewcommand{\theHlemma}{A.\arabic{lemma}}
\setcounter{proposition}{0}
\renewcommand{\theproposition}{A.\arabic{proposition}}
\renewcommand{\theHproposition}{A.\arabic{proposition}}

\section{Omitted Proofs}
\subsection{Proof of Proposition \ref{prop: general consumer rule}}
\begin{proof}
    The statement is organized in chronological order; the proof proceeds by backward induction.
    First fix the terminal decision after product $2$ has been searched.
    If the consumer has no product and observes $m_2=0$, buying product $2$ cannot generate consumption value.
    Its best possible payoff is $\max\{-p_2,-h\}\leq0,$ while not buying gives $0$.
    Thus, an unmatched product $2$ is not bought.
    If instead $m_2=1$, the expected payoff from buying product $2$ is
    $$
        E_{\varepsilon_2}\left[\max\{-h,\varepsilon_2-p_2\}\right]
        =\phi_F(p_2-h)-h,
    $$
    and not buying gives $0$.
    Hence product $2$ is bought if and only if $\phi_F(p_2-h)>h$, equivalently $-h<B(h)-p_2$.
    Conditional on purchase, the consumer keeps product $2$ if and only if $\varepsilon_2-p_2\geq-h$.

    Now suppose the consumer is holding product $1$ with payoff $x\geq-h$.
    If $m_2=0$, buying product $2$ again has no consumption value.
    The final payoffs from keeping product $1$, keeping product $2$, and returning both products would be
    $$
        x-h,\qquad -p_2-h,\qquad -2h.
    $$
    Since $p_2\geq0$, $h>0$, and $x\geq-h$, none of these payoffs strictly exceeds $x$.
    By the strict-action convention, a weak tie does not induce purchase.
    Thus an unmatched product $2$ is not bought.
    If $m_2=1$, the final payoffs from keeping product $1$, keeping product $2$, and returning both products are
    $$
        x-h,\qquad \varepsilon_2-p_2-h,\qquad -2h.
    $$
    Because $x\geq-h$, returning both is (weakly) dominated by keeping product $1$ and returning product $2$.
    Buying product $2$ has a strictly positive incremental payoff if and only if
    $$
        E_{\varepsilon_2}\left[\max\{x-h,\varepsilon_2-p_2-h\}\right]> x,
    $$
    or equivalently $\phi_F(p_2+x)>h$, hence $x<B(h)-p_2$.
    Conditional on buying product $2$, comparing the two remaining payoffs gives the final product return decision.

    Next consider the search decision for product $2$.
    The no-product state can arise from two histories: either the consumer skipped product $1$, or she purchased product $1$, observed $x< -h$, and returned product $1$ before deciding whether to search product $2$.
    When the consumer does not hold product $1$ at hand, without search she obtains $0$.
    With search, the branch $m_2=0$ has zero incremental value, while the matched branch has value
    $$
        \max\{\phi_F(p_2^*-h)-h,0\}.
    $$
    Thus searching has a strictly positive incremental payoff if and only if
    $$
        \alpha\max\{\phi_F(p_2^*-h)-h,0\}> s.
    $$
    Since $s>0$, this inequality can hold only when the inner term is positive, so it is equivalent to
    $$
        \phi_F(p_2^*-h)> h+\frac{s}{\alpha},
    $$
    which is equivalent to $p_2^*-h<A(h)$.

    Next show how to obtain $Q_0(p_2^*;h)$.
    When the consumer reaches the no-product state without any product-$1$ history, the outside option is $0$.
    Using the search payoff derived above, the value of this option is
    $$
        Q_0(p_2^*;h)
        =
        \max\left\{
        \alpha\max\{\phi_F(p_2^*-h)-h,0\}-s,\ 0
        \right\}.
    $$
    The outer maximum reflects the option not to search product $2$.
    If the no-product state is reached after returning product $1$, that is $x< -h$, the total continuation payoff is instead $-h+Q_0(p_2^*;h)$; the additive term $-h$ does not change the search rule.

    It remains to consider a consumer who has purchased product $1$ and still holds it, where $x\geq-h$.
    Without search she obtains $x$.
    With search, the branch $m_2=0$ has zero incremental value relative to this fallback payoff, and the additional expected value of search is
    $$
        \alpha\max\{\phi_F(p_2^*+x)-h,0\}-s.
    $$
    Hence search is taken if and only if this term is strictly positive.
    Again, because $s>0$, this is equivalent to
    $$
        \phi_F(p_2^*+x)> h+\frac{s}{\alpha},
    $$
    or $x<A(h)-p_2^*$.

    Next show how to obtain $Q(x;p_2^*,h)$.
    Combining this holding-product rule for $x\geq-h$ with the no-product rule for $x< -h$ gives the compact condition $v(x;h)<A(h)-p_2^*$.
    After purchasing product $1$ and observing $x$, the consumer can guarantee
    $$
        v(x;h)=\max\{x,-h\},
    $$
    by either keeping product $1$ when $x\geq-h$ or returning it when $x< -h$.
    If she additionally searches product $2$, the branch $m_2=0$ adds no value relative to this fallback payoff.
    Conditional on $m_2=1$, the expected incremental value of the product-$2$ purchase option is
    $$
        \max\{\phi_F(p_2^*+v(x;h))-h,0\}.
    $$
    Thus the continuation value after purchasing product $1$ and observing $x$ is
    $$
        Q(x;p_2^*,h)
        =
        v(x;h)+
        \max\left\{
        \alpha\max\{\phi_F(p_2^*+v(x;h))-h,0\}-s,\ 0
        \right\}.
    $$
    The first term is the fallback payoff from product $1$, the inner maximum is the option not to buy product $2$ after a match, and the outer maximum is the option not to search product $2$.

    Finally consider the initial purchase decision after the free search of product $1$.
    If $m_1=0$, product $1$ has no consumption value.
    Buying it cannot give a strict gain over skipping it.
    To see this, fix any continuation strategy after buying the unmatched product.
    If the consumer eventually returns product $1$, she pays the return hassle cost $h$, whereas a consumer who skipped product $1$ can follow the same search, purchase, and return decisions for product $2$ without paying this cost.
    If the consumer never returns product $1$, then any final payoff involving product $1$ includes the term $-p_1\leq0$ and has no offsetting consumption benefit.
    Hence buying a nonmatched product $1$ cannot generate a strictly positive incremental payoff over skipping it.
    The consumer therefore does not buy product $1$ when $m_1=0$.
    When $m_1=1$, skipping product $1$ gives the continuation value $Q_0(p_2^*;h)$.
    Purchasing product $1$ gives the expected continuation value $\mathcal B_1(p_1;p_2^*,h)$.
    Thus product $1$ is purchased if and only if $\mathcal B_1(p_1;p_2^*,h)> Q_0(p_2^*;h)$.
\end{proof}

\subsection{Proof of Corollary \ref{prop:limiting_consumer_rule}}
\begin{proof}
    We take the right limit of Proposition \ref{prop: general consumer rule}.
    Specializing the proposition to the benchmark uniform distribution, $\phi_F$ coincides with the function $\phi$ defined in the main text.
    For $h>0$, the relevant cutoffs and fallback payoff are
    $$
        A(h)=\phi^{-1}\left(h+\frac{s}{\alpha}\right),
        \qquad
        B(h)=\phi^{-1}(h),
        \qquad
        v(x;h)=\max\{x,-h\}.
    $$
    As $h\to0^+$,
    $$
        A(h)\to a,\qquad B(h)\to1,\qquad v(x;h)\to v_0(x)=\max\{x,0\}.
    $$
    The formula for $a$ follows from $\phi(a)=s/\alpha$ and
    $\phi(z)=(1-z)^2/2$ on $[0,1]$ when $s<\alpha/2$.

    The search and retention rules follow by substitution.
    Since $v(x;h)\to v_0(x)$, the consumer returns product $1$ when $x<0$ and keeps it when $x\geq0$.
    The no-product search condition
    $$
        \alpha\max\{\phi(p_2^*-h)-h,0\}>s
    $$
    converges to $\alpha\phi(p_2^*)>s$, equivalently $p_2^*<a$.
    The post-product-$1$ search condition converges to
    $$
        v_0(x)<a-p_2^*.
    $$
    The product-$2$ purchase cutoffs $-h<B(h)-p_2$ in the no-product state and $x<B(h)-p_2$ when product $1$ is held converge to $v_0(x)<1-p_2$.
    The nonpurchase rule after $m_2=0$ and the final retention comparisons are the corresponding limits of the payoff comparisons in Proposition \ref{prop: general consumer rule}.

    It remains to simplify the product-$1$ purchase condition.
    Proposition \ref{prop: general consumer rule} already implies that an unmatched product $1$ is not bought.
    For a matched product $1$, define the limiting continuation value from entering the product-$2$ search stage with fallback payoff $v_0(x)\geq0$ by
    $$
        Q(v_0(x);p_2^*)\equiv v_0(x)+\max\{\alpha\phi(p_2^*+v_0(x))-s,0\}.
    $$
    This is the $h\to0^+$ limit of the continuation value embedded in $Q(x;p_2^*,h)$.

    The function $Q(\cdot;p_2^*)$ is strictly increasing.
    When search is inactive, $Q(z;p_2^*)=z$; when search is active,
    $$
        \frac{\partial Q(z;p_2^*)}{\partial z}
        =1-\alpha(1-p_2^*-z)
        \geq 1-\alpha>0.
    $$
    Continuity at the search cutoff gives strict monotonicity globally.

    Skipping product $1$ gives $Q(0;p_2^*)$.
    Buying product $1$ gives fallback payoff $\max\{\varepsilon_1-p_1,0\}$ after $\varepsilon_1$ is observed, so its ex ante value is
    $$
        \int_0^1 Q\big(\max\{\varepsilon_1-p_1,0\};p_2^*\big)\,d\varepsilon_1 .
    $$
    If $p_1<1$, the set $\{\varepsilon_1>p_1\}$ has positive measure, so strict monotonicity of $Q$ makes buying strictly better than skipping.
    If $p_1\geq1$, then $\max\{\varepsilon_1-p_1,0\}=0$ almost surely, so buying gives no strict gain.
    Under the strict-action convention, a matched product $1$ is purchased if and only if $p_1<1$.
\end{proof}

\subsection{Proof of Proposition \ref{lem:limiting_armstrong_style}}
\begin{proof}
    We first derive the local first-order conditions and then solve the resulting algebraic system on the interior price interval; the remaining active-branch cutoff inequalities are verified after the price-ranking result.
    The proof follows the logic of \cite{armstrong2009prominence}.

    In local deviations, consumers' belief $p_2^*$ is held fixed while the deviating firm's actual price varies.
    Fix $p_2^*$ and vary firm~1's price locally inside the interior pricing region.
    Firm~1's retained demand is
    $$
        D_1(p_1,p_2^*)=
        \alpha(1-\alpha)(1-p_1)
        +\alpha^2\left[1-(a+p_1-p_2^*)+\frac{a^2-(p_2^*)^2}{2}\right],
    $$
    and its initial sales are $S_1^*=\alpha$.
    Hence
    $$
        \frac{\partial D_1(p_1,p_2^*)}{\partial p_1}=-\alpha,
        \qquad
        \frac{\partial^2\pi_1(p_1,p_2^*)}{\partial p_1^2}=-2\alpha<0.
    $$
    Because $p_1^*$ is interior, firm~1's necessary and sufficient local optimality condition is
    $$
        D_1(p_1^*,p_2^*)-\alpha(p_1^*-\eta)=0,
    $$
    which is equivalent to \eqref{eq:limiting_FOC1}.

    Now fix $p_1^*$ and vary firm~2's price locally, holding consumers' equilibrium belief about the unobserved price at $p_2^*$.
    Denote $h_2=H_1(p_1^*)=a+p_1^*-p_2^*$.
    Firm~2's initial sales are
    $$
        S_2=\alpha(1-\alpha)+\alpha^2h_2,
    $$
    and its local retained demand is
    $$
        D_2(p_1^*,p_2)=
        \alpha(1-\alpha)(1-p_2)
        +\alpha^2\left[(1-a-p_2+p_2^*)h_2+
        \frac{h_2^2-(p_1^*)^2}{2}\right].
    $$
    Thus
    $$
        \frac{\partial D_2(p_1^*,p_2)}{\partial p_2}=-S_2,
        \qquad
        \frac{\partial^2\pi_2(p_1^*,p_2)}{\partial p_2^2}=-2S_2<0.
    $$
    Firm~2's necessary and sufficient local optimality condition is
    $$
        D_2(p_1^*,p_2^*)-S_2(p_2^*-\eta)=0,
    $$
    which is equivalent to \eqref{eq:limiting_FOC2}.

    Since $a>1/2$ and $\eta\geq0$, we have $1-a<\frac{1}{2}\leq\frac{1+\eta}{2}$.
    By \eqref{eq:armstrong_style_salvage_conditions}, $\frac{1+\eta}{2}<\frac{1+2a-1}{2}=a$ and $\eta<1-a$.
    Thus $\left(1-a,\frac{1+\eta}{2}\right)$ is nonempty, lies below the search cutoff $a$, and lies above $\eta$.

    First, reduce firm $1$'s first-order condition.
    Define the price mapping at a generic argument $p_2$ by
    $$
        \chi(p_2)\equiv\frac{1}{2}\left[
        1+\eta-\alpha a+\alpha p_2+\frac{\alpha}{2}(a^2-p_2^2)
        \right].
    $$

    Thus, given $p_2^*$, firm $1$'s first-order condition has the unique solution $p_1^*=\chi(p_2^*)$.
    We show that $\chi$ maps the interior price interval $\left(1-a,\frac{1+\eta}{2}\right)$ into itself.

    $\chi(p_2^*)$ is increasing in $p_2^*$ because
    $$
        \chi'(p_2^*)=\frac{\alpha(1-p_2^*)}{2}>0.
    $$

    Because $a>1/2$, $\eta\geq0$, and $0<\alpha<1$,
    $$
    \begin{aligned}
    \chi(1-a)
    &=\frac{1+\eta-\alpha(a-\tfrac12)}{2}
      \geq\frac{1-\alpha(a-\tfrac12)}{2}\\
    &>\frac{1-(a-\tfrac12)}{2}
      =\frac{1-a+\tfrac12}{2}>1-a.
    \end{aligned}
    $$
    Also, $(1+\eta)/2<a<1$ implies $(1+\eta)/2+a<2$, so
    $$
    \begin{aligned}
    \chi\left(\frac{1+\eta}{2}\right)
    &=\frac{1}{2}\left[1+\eta
       +\alpha\left(\frac{1+\eta}{2}-a\right)
       \left(1-\frac{1}{2}\left(\frac{1+\eta}{2}+a\right)\right)\right]\\
    &<\frac{1+\eta}{2}.
    \end{aligned}
    $$

    Therefore $\chi(p_2^*)\in(1-a,\frac{1+\eta}{2})$ whenever $p_2^*\in(1-a,\frac{1+\eta}{2})$.

    Next substitute \eqref{eq:limiting_FOC1} into \eqref{eq:limiting_FOC2}.
    Along the equation $p_1^*=\chi(p_2^*)$, let $h_2(p_2^*)\equiv a+\chi(p_2^*)-p_2^*$ so that $S_2=\alpha [1-\alpha+\alpha h_2(p_2^*)]$.
    Therefore we have
    $$
        \frac{D_2}{S_2}=\frac{
        (1-\alpha)(1-p_2^*)
        +\alpha\left[(1-a)h_2(p_2^*)+\frac{h_2(p_2^*)^2-\chi(p_2^*)^2}{2}\right]
        }{1-\alpha+\alpha h_2(p_2^*)}.
    $$ 

    Since $h_2(p_2^*)-\chi(p_2^*)=a-p_2^*$,
    $$
    \begin{aligned}
        (1-a)h_2(p_2^*)+\frac{h_2(p_2^*)^2-\chi(p_2^*)^2}{2}
        &=
        (1-a)h_2(p_2^*)+\frac{(a-p_2^*)\left[h_2(p_2^*)+\chi(p_2^*)\right]}{2}\\
        &=
        (1-p_2^*)h_2(p_2^*)-\frac{(a-p_2^*)^2}{2}.
    \end{aligned}
    $$

    Therefore the numerator in the preceding display is
    $$
    \begin{aligned}
        &(1-\alpha)(1-p_2^*)
        +\alpha\left[(1-a)h_2(p_2^*)+\frac{h_2(p_2^*)^2-\chi(p_2^*)^2}{2}\right]\\
        &\qquad =
        [1-\alpha+\alpha h_2(p_2^*)](1-p_2^*)-\frac{\alpha(a-p_2^*)^2}{2}.
    \end{aligned}
    $$

    Hence \eqref{eq:limiting_FOC2} can be written as
    $$
        p_2^*-\eta=\frac{D_2}{S_2}
        =
        1-p_2^*-\frac{\alpha(a-p_2^*)^2}{2[1-\alpha+\alpha h_2(p_2^*)]}.
    $$

    Thus firm $2$'s first-order condition is equivalent to 
    \begin{equation}\label{eq:foc1_2}
    p_2^*=\mathcal{T}(p_2^*),
    \end{equation}

    Here the fixed-point mapping is defined at a generic argument by
    $$
      \mathcal T(p_2)\equiv\frac{1+\eta}{2}
      -\frac{\alpha(a-p_2)^2}{4[1-\alpha+\alpha h_2(p_2)]}.
    $$
    The function argument $p_2$ can be distinct from the equilibrium value $p_2^*$ satisfying \eqref{eq:foc1_2}.

    At $p_2=\frac{1+\eta}{2}$, we have $\mathcal{T}(\frac{1+\eta}{2})=\frac{1+\eta}{2}-\frac{\alpha(a-\frac{1+\eta}{2})^2}{4[1-\alpha+\alpha h_2(\frac{1+\eta}{2})]}<\frac{1+\eta}{2}$.

    At $p_2=1-a$, we have
    $$
    \mathcal{T}(1-a)=\frac{1+\eta}{2}-\frac{\alpha(a-(1-a))^2}{4[1-\alpha+\alpha h_2(1-a)]}=\frac{1+\eta}{2}-\frac{\alpha(2a-1)^2}{4[1-\alpha+\alpha h_2(1-a)]}.
    $$
    Because
    $$
        1-\alpha+\alpha h_2(1-a)-\frac{\alpha(2a-1)}{2}=1-\alpha+\alpha(2a+\chi(1-a)-1)-\frac{\alpha(2a-1)}{2}=1-\alpha+\frac{\alpha}{2}(1+\eta)+\alpha(\frac{1}{2}-\frac{\alpha}{4})(2a-1)>0
    $$
    and $\frac{2a-1}{2}>0$, we have $\mathcal{T}(1-a)>\frac{1+\eta}{2}-\frac{2a-1}{2}\geq 1-a$.

    Hence, by the intermediate value theorem, there exists at least one $p_2^*\in\left(1-a,\frac{1+\eta}{2}\right)$ such that $p_2^*=\mathcal{T}(p_2^*)$.

    Finally, show that the root is unique.
    Differentiating $\mathcal{T}(p_2)$ gives
    $$
    \begin{aligned}
        \mathcal{T}'(p_2)&=\frac{\alpha(a-p_2)\left[2[1-\alpha+\alpha h_2(p_2)]-\alpha(a-p_2)(1-\chi'(p_2))\right]}{4[1-\alpha+\alpha h_2(p_2)]^2}\\
        &=\frac{1}{4}\cdot \frac{\alpha(a-p_2)}{1-\alpha+\alpha h_2(p_2)}\left[ 2- \frac{\alpha(a-p_2)}{1-\alpha+\alpha h_2(p_2)}(1-\chi'(p_2))  \right].
    \end{aligned}
    $$

    Because $1-\alpha+\alpha h_2(p_2)-\alpha(a-p_2)=1-\alpha+\alpha \chi(p_2)>0$ and $\alpha(a-p_2)>0$, we have $0<\frac{\alpha(a-p_2)}{1-\alpha+\alpha h_2(p_2)}<1$.
    Recall that $\chi'(p_2)=\frac{\alpha(1-p_2)}{2}\in (0,\frac{1}{2})$.
    Therefore, we have
    $$
        \mathcal{T}'(p_2)>\frac{1}{4}\cdot \frac{\alpha(a-p_2)}{1-\alpha+\alpha h_2(p_2)}\left[ 2- \frac{\alpha(a-p_2)}{1-\alpha+\alpha h_2(p_2)} \right]>0,
    $$
    and
    $$
    \begin{aligned}
        \mathcal{T}'(p_2)&<\frac{1}{4}\cdot \frac{\alpha(a-p_2)}{1-\alpha+\alpha h_2(p_2)}\left[ 2- \frac{\alpha(a-p_2)}{1-\alpha+\alpha h_2(p_2)}(1-1/2)  \right]\\
        &<\frac{1}{2}\cdot \frac{\alpha(a-p_2)}{1-\alpha+\alpha h_2(p_2)}<\frac{1}{2}.
    \end{aligned}
    $$

    Therefore, $p_2-\mathcal{T}(p_2)$ is strictly increasing in $p_2$, and $p_2-\mathcal{T}(p_2)=0$ has a unique solution given $p_2\in(1-a,\frac{1+\eta}{2})$.

    Finally, setting $p_2^*$ equal to this root and $p_1^*=\chi(p_2^*)$ gives the unique solution to the first-order system \eqref{eq:limiting_FOC1}--\eqref{eq:limiting_FOC2} in the interior price interval.
    The proof above shows that $1-a<p_i^*<\frac{1+\eta}{2}<a$, for $i=1,2$.
\end{proof}

\subsection{Proof of Proposition \ref{prop:limiting_price_comparison}}
\begin{proof}
    Write
    $$
        S_2=\alpha\left[(1-\alpha)+\alpha h_2\right].
    $$
    At a solution of the first-order system, the two first-order conditions can be written as
    $$
        p_1^*-\eta=\frac{D_1}{\alpha},
        \qquad
        p_2^*-\eta=\frac{D_2}{S_2}.
    $$
    Hence
    $$
        p_2^*-p_1^*
        =
        \frac{\alpha D_2-S_2D_1}{\alpha S_2}.
    $$
    Suppose, to the contrary, that $p_1^*\geq p_2^*$.
    Put $\mu\equiv p_1^*-p_2^*\geq0$.
    By \eqref{eq:armstrong_style_price_interval}, $0<p_1^*<a<1$, and
    $$
        p_2^*=p_1^*-\mu,
        \qquad
        h_2=a+p_1^*-p_2^*=a+\mu .
    $$
    Substitution into the demand formulas gives
    $$
    \begin{aligned}
        \frac{D_1}{\alpha}
        =&\ (1-\alpha)(1-p_1^*)
        +\alpha\left[
        1-a-\mu+\frac{a^2-(p_1^*-\mu)^2}{2}
        \right]\\
        =&\ (1-\alpha)(1-p_1^*)
        +\alpha\left[
        1-a+\frac{a^2-(p_1^*)^2}{2}
        +(p_1^*-1)\mu-\frac{\mu^2}{2}
        \right],
    \end{aligned}
    $$
    and
    $$
    \begin{aligned}
        \frac{D_2}{\alpha}
        =&\ (1-\alpha)(1-p_1^*+\mu)
        +\alpha\left[
        (1-a)(a+\mu)+\frac{(a+\mu)^2-(p_1^*)^2}{2}
        \right]\\
        =&\ (1-\alpha)(1-p_1^*+\mu)
        +\alpha\left[
        a-\frac{a^2}{2}-\frac{(p_1^*)^2}{2}
        +\mu+\frac{\mu^2}{2}
        \right].
    \end{aligned}
    $$
    Moreover,
    $$
        \frac{S_2}{\alpha}=1-\alpha+\alpha(a+\mu).
    $$
    Hence
    $$
    \begin{aligned}
        \frac{2(\alpha D_2-S_2D_1)}{\alpha^2}
        =&\ 2\left[
        \frac{D_2}{\alpha}
        -\frac{S_2}{\alpha}\frac{D_1}{\alpha}
        \right]\\
        =&\ 2\Bigg\{
        (1-\alpha)(1-p_1^*+\mu)
        +\alpha\left[
        a-\frac{a^2}{2}-\frac{(p_1^*)^2}{2}
        +\mu+\frac{\mu^2}{2}
        \right]\\
        &\quad
        -\left[1-\alpha+\alpha(a+\mu)\right]\\
        &\quad\quad \times
        \left[
        (1-\alpha)(1-p_1^*)
        +\alpha\left(
        1-a+\frac{a^2-(p_1^*)^2}{2}
        +(p_1^*-1)\mu-\frac{\mu^2}{2}
        \right)
        \right]
        \Bigg\}.
    \end{aligned}
    $$
    Expanding and collecting powers of $\mu$ gives
    $$
    \begin{aligned}
        \frac{2(\alpha D_2-S_2D_1)}{\alpha^2}
        =&\ \alpha(a-p_1^*)(1-a)\left[2-\alpha(2-a-p_1^*)\right]\\
        &+\mu\left[2+\alpha^2\left((p_1^*)^2-2ap_1^*-a^2+4a-2\right)\right]\\
        &+\alpha\mu^2\left[2+\alpha(1+a-2p_1^*)\right]+\alpha^2\mu^3 .
    \end{aligned}
    $$
    The constant term is strictly positive.
    Indeed, $0<\alpha<1$ and $0<p_1^*<a<1$ imply $a-p_1^*>0$, and $1-a>0$.
    Since $2-a-p_1^*>0$,
    $$
        2-\alpha(2-a-p_1^*)>2-(2-a-p_1^*)=a+p_1^*>0.
    $$
    The coefficient of $\mu$ is positive because
    $$
        (p_1^*)^2-2ap_1^*-a^2+4a-2=(p_1^*-a)^2-2a^2+4a-2>-2.
    $$
    The coefficient of $\mu^2$ is positive because $1+a-2p_1^*>1-a>0$.
    Therefore $\alpha D_2-S_2D_1>0$.
    It follows that $p_2^*>p_1^*$, a contradiction.

    This contradiction proves that $p_1^*<p_2^*$.
\end{proof}

\subsection{Proof of Proposition \ref{prop:limiting_comparative_statics}}
\begin{proof}
    Define $\Lambda(p_2)=p_2-\mathcal{T}(p_2)$, so that $p_2^*$ is pinned down by the reduced equation $\Lambda(p_2)=0$.
    From the proof of Proposition \ref{lem:limiting_armstrong_style}, we have $\frac{\partial \Lambda(p_2;a,\eta,\alpha)}{\partial p_2}=1-\mathcal{T}'(p_2)>0$.
    Unless an argument is displayed, derivatives below are evaluated at the pricing equilibrium.

    First consider $\eta$.
    We first derive and bound the response of $p_2^*$.
    Direct differentiation gives
    $$
        \frac{\partial \Lambda(p_2;a,\eta,\alpha)}{\partial \eta}=-\frac{1}{2}-\frac{\alpha^2(a-p_2)^2}{8[1-\alpha+\alpha h_2(p_2)]^2}<0.
    $$
    Hence, by the implicit function theorem,
    $$
        \frac{dp_2^*}{d\eta}
        =
        -\frac{\partial\Lambda/\partial\eta}{\partial\Lambda/\partial p_2}>0.
    $$
    It remains to obtain the upper bound on $dp_2^*/d\eta$ and compare it with $dp_1^*/d\eta$.
    For this purpose, set
    $$
        d\equiv a-p_2^*,
        \qquad
        z\equiv\frac{\alpha d}{q_2},
        \qquad
        \kappa\equiv1-\chi'(p_2^*).
    $$
    Since $q_2-\alpha d=1-\alpha+\alpha p_1^*>0$, we have $0<z<1$.
    From the proof of Proposition \ref{lem:limiting_armstrong_style},
    $$
        \mathcal{T}'(p_2^*)=\frac{z(2-\kappa z)}{4}.
    $$
    Moreover, the partial derivative of $\mathcal{T}$ with respect to $\eta$, holding $p_2$ fixed, is
    $$
        \mathcal{T}_\eta(p_2^*)
        \equiv
        \left.
        \frac{\partial \mathcal{T}(p_2;a,\eta,\alpha)}{\partial\eta}
        \right|_{p_2=p_2^*}
        =
        \frac{1}{2}+\frac{z^2}{8}.
    $$
    Differentiating the reduced equation $p_2^*=\mathcal{T}(p_2^*)$ with respect to $\eta$ therefore gives
    $$
        \frac{dp_2^*}{d\eta}
        =
        \frac{\frac{1}2+\frac{z^2}{8}}{1-\frac{z(2-\kappa z)}{4}}.
    $$
    Since $\chi'(p_2^*)<1/2$, we have $\kappa>1/2$.
    Hence the denominator exceeds the numerator:
    $$
        1-\frac{z(2-\kappa z)}{4}
        -\left(\frac{1}{2}+\frac{z^2}{8}\right)
        =
        \frac{1-z}{2}+\frac{(2\kappa-1)z^2}{8}>0,
    $$
    and therefore $dp_2^*/d\eta<1$.

    We need to use the following key auxiliary inequality to compare $dp_2^*/d\eta$ with $dp_1^*/d\eta$:
    $$
        z<\alpha(1-p_2^*).
    $$

    First show the auxiliary inequality.
    The fixed-point equation gives
    $$
        \eta=2p_2^*-1+\frac{zd}{2}.
    $$
    Let $r=1-a$.
    Substituting this expression for $\eta$ into $q_2=1-\alpha+\alpha[a+\chi(p_2^*)-p_2^*]$ and using $q_2=\alpha d/z$ gives
    $$
        \Psi(z)\equiv
        1-\alpha r+\frac{\alpha dz}{4}
        -\frac{\alpha^2d(2r+d)}{4}
        -\frac{\alpha d}{z}=0,
    $$
    with $\Psi'(z)=\alpha d/4+\alpha d/z^2>0$.
    Thus $\Psi$ is strictly increasing.
    We next show that the comparison point $\alpha(r+d)=\alpha(1-p_2^*)$ lies to the right of the root $z$.
    The same fixed-point equation implies $r+d<2/3$.
    Since $p_2^*=a-d=1-r-d$,
    $$
        \eta
        =
        1-2r-2d+\frac{zd}{2}
        =
        1-2r-\left(2-\frac{z}{2}\right)d .
    $$
    Because $\eta\geq0$ and $0<z<1$,
    $$
        1-2r
        \geq
        \left(2-\frac{z}{2}\right)d
        >
        \frac{3d}{2}.
    $$
    Hence $d<2(1-2r)/3$, and therefore
    $$
        r+d
        <
        r+\frac{2(1-2r)}{3}
        =
        \frac{2}{3}-\frac{r}{3}
        <
        \frac{2}{3} .
    $$
    Evaluating at $\alpha(r+d)=\alpha(1-p_2^*)$ gives
    $$
        \Psi(\alpha(r+d))
        =
        r\left[
        \frac{1}{r+d}-\alpha-\frac{\alpha^2d}{4}
        \right]>0.
    $$
    The strict inequality follows because $r+d<2/3$ implies
    $1/(r+d)>3/2$, while $d<2/3$ and $\alpha<1$ imply
    $\alpha+\alpha^2d/4<1+d/4<7/6$.
    Since $\Psi$ is strictly increasing and $\Psi(z)=0$, it follows that
    $z<\alpha(r+d)=\alpha(1-p_2^*)$.
    This proves the auxiliary inequality.

    We now use the auxiliary inequality to compare the two price derivatives.
    Since
    $$
       \chi'(p_2^*)=\alpha(1-p_2^*)/2,
    $$
    the auxiliary inequality gives
    $$
        \kappa=1-\frac{\alpha(1-p_2^*)}{2}<1-\frac{z}{2}.
    $$
    This implies
    $$
        1-\frac{z(2-\kappa z)}{4}
        -2\kappa\left(\frac{1}{2}+\frac{z^2}{8}\right)
        =
        1-\kappa-\frac{z}{2}>0,
    $$
    so $dp_2^*/d\eta<1/(2\kappa)$.

    Since $p_1^*=\chi(p_2^*)$,
    $$
        \frac{dp_1^*}{d\eta}
        =
        \frac{1}{2}+\chi'(p_2^*)\frac{dp_2^*}{d\eta}
        =
        \frac{1}{2}+(1-\kappa)\frac{dp_2^*}{d\eta}.
    $$
    The inequality $dp_2^*/d\eta<1/(2\kappa)$ gives $dp_1^*/d\eta>dp_2^*/d\eta$, while
    $\chi'(p_2^*)<1/2$ and $dp_2^*/d\eta<1$ give $dp_1^*/d\eta<1$.
    Hence
    $$
        0<\frac{dp_2^*}{d\eta}<\frac{dp_1^*}{d\eta}<1.
    $$
    Therefore the price gap decreases in $\eta$:
    $$
        \frac{d(p_2^*-p_1^*)}{d\eta}
        =
        \frac{dp_2^*}{d\eta}-\frac{dp_1^*}{d\eta}<0.
    $$
    Moreover,
    $$
        q_2=1-\alpha+\alpha(a+p_1^*-p_2^*),
    $$
    so
    $$
        \frac{dq_2}{d\eta}
        =
        \alpha\left(\frac{dp_1^*}{d\eta}-\frac{dp_2^*}{d\eta}\right)>0.
    $$

    Next consider $s$.
    Since $a=1-\sqrt{2s/\alpha}$ and $da/ds<0$, it is convenient to first sign the effect of $a$.
    Direct differentiation gives
    $$
    \begin{aligned}
        \frac{\partial \Lambda(p_2;a,\eta,\alpha)}{\partial a}&=\frac{\alpha (a-p_2)}{4[1-\alpha+\alpha h_2(p_2)]^2}\left[ 2(1-\alpha+\alpha h_2(p_2))-\alpha(a-p_2)(1-\frac{\alpha}{2}(1-a)) \right]\\
        &=\frac{\alpha (a-p_2)}{4[1-\alpha+\alpha h_2(p_2)]^2}\left[ 2(1-\alpha)+\alpha(a-p_2)\frac{\alpha}{2}(1-a)+\alpha(2h_2(p_2)-(a-p_2))   \right]\\
        &=\frac{\alpha (a-p_2)}{4[1-\alpha+\alpha h_2(p_2)]^2}\left[ 2(1-\alpha)+\frac{\alpha^2}{2}(a-p_2)(1-a)+\alpha(a-p_2+2\chi(p_2))   \right]>0.
    \end{aligned}
    $$ 

    The final expression is strictly positive in the interior price interval.
    Therefore,
    $$
        \frac{dp_2^*}{da}
        =
        -\frac{\partial\Lambda/\partial a}{\partial\Lambda/\partial p_2}<0.
    $$
    Since $da/ds<0$, we have $dp_2^*/ds>0$.
    Along $p_1^*=\chi(p_2^*)$,
    $$
        \frac{dp_1^*}{da}
        =
        \frac{\partial \chi}{\partial p_2}\frac{dp_2^*}{da}
        +\frac{\partial \chi}{\partial a}
        =
        \frac{\partial \chi}{\partial p_2}\frac{dp_2^*}{da}
        -\frac{\alpha}{2}(1-a)<0.
    $$
    The two terms in the last line are both negative.
    Hence $dp_1^*/ds>0$.

    Finally consider $\alpha$.
    Here $\alpha$ has both a direct effect and an indirect effect through $a$, with $da/d\alpha=(1-a)/(2\alpha)>0$.
    Direct differentiation gives
    $$
    \begin{aligned}
        &\frac{\partial \Lambda(p_2;a,\eta,\alpha)}{\partial \alpha}\\
        &=\frac{(a-p_2)^2}{4[1-\alpha+\alpha h_2(p_2)]^2}\left[ 1-\alpha +\alpha h_2(p_2)+(1-h_2(p_2))\alpha+\alpha^2(a-p_2)\frac{2-a-p_2}{4} \right]>0.
    \end{aligned}
    $$

    Combining the direct $\alpha$ effect with the induced change in $a$ gives
    $$
    \begin{aligned}
    \frac{d p_2^*}{d \alpha}&=-\frac{\frac{\partial \Lambda(p_2;a,\eta,\alpha)}{\partial \alpha}  + \frac{\partial \Lambda(p_2;a,\eta,\alpha)}{\partial a}  \cdot \frac{d a}{d \alpha}}{\frac{\partial \Lambda(p_2;a,\eta,\alpha)}{\partial p_2}}\\
    &=-\frac{\frac{\partial \Lambda(p_2;a,\eta,\alpha)}{\partial \alpha}  + \frac{\partial \Lambda(p_2;a,\eta,\alpha)}{\partial a}  \cdot \frac{1-a}{2\alpha}}{\frac{\partial \Lambda(p_2;a,\eta,\alpha)}{\partial p_2} }<0.
    \end{aligned}
    $$

    Since $p_1^*=\chi(p_2^*)$, we have
    $$
        \frac{d p_1^*}{d \alpha}
        =
        \frac{\partial \chi}{\partial p_2}\frac{d p_2^*}{d \alpha}
        +\frac{\partial \chi}{\partial \alpha}
        +\frac{\partial \chi}{\partial a}\frac{d a}{d \alpha}
        =
        \frac{\partial \chi}{\partial p_2}\frac{d p_2^*}{d \alpha}
        -\frac{a-p_2^*}{4}(2-a-p_2^*)
        -\frac{\alpha}{2}(1-a)\frac{1-a}{2\alpha}<0.
    $$
    The first term is negative because $\partial \chi/\partial p_2>0$ and $dp_2^*/d\alpha<0$, and the two direct terms are strictly negative in the interior price interval.
\end{proof}

\subsection{Proof of Proposition \ref{prop:limiting_detrimental_prominence}}
\begin{proof}
    The first-order conditions imply
    $$
        p_1^*-\eta=\frac{D_1}{\alpha},
        \qquad
        p_2^*-\eta=\frac{D_2}{S_2}.
    $$
    Hence equilibrium profits can be written as
    $$
        \pi_1^*
        =
        \alpha\left[(p_1^*-\eta)^2-\rho\right],
        \qquad
        \pi_2^*
        =
        \alpha q_2 \left[(p_2^*-\eta)^2-\rho\right].
    $$

    Rearranging the profit gap gives
    $$
    \begin{aligned}
        \Gamma
        &=
        \alpha\left[(p_1^*-\eta)^2-\rho\right]
        -\alpha q_2\left[(p_2^*-\eta)^2-\rho\right]\\
        &=
        \alpha\left[
        (p_1^*-\eta)^2-q_2(p_2^*-\eta)^2-(1-q_2)\rho
        \right]\\
        &=
        \alpha(1-q_2)\left[\rho^\dagger(\eta)-\rho\right].
    \end{aligned}
    $$

    Since $0<h_2<1$, we have $1-q_2=\alpha(1-h_2)>0$.
    Hence $\pi_1^*<\pi_2^*$ if and only if $\rho>\rho^\dagger(\eta)$.

    Let $\overline\rho(\eta)=(p_1^*-\eta)^2$.
    Proposition~\ref{prop:limiting_price_comparison} gives $p_1^*<p_2^*$, and Proposition~\ref{lem:limiting_armstrong_style} gives $p_1^*>\eta$.
    Therefore $(p_1^*-\eta)^2<(p_2^*-\eta)^2$.
    Since $0<q_2<1$,
    $$
        \rho^\dagger(\eta)-\overline\rho(\eta)
        =
        \frac{q_2\left[(p_1^*-\eta)^2-(p_2^*-\eta)^2\right]}{1-q_2}<0.
    $$
    Thus the interval $(\rho^\dagger(\eta),\overline\rho(\eta)]$ is nonempty.
    If $\rho\leq\overline\rho(\eta)$, then $\pi_1^*\geq0$; since $(p_2^*-\eta)^2>\overline\rho(\eta)$, the same inequality implies $\pi_2^*>0$.
    Conversely, if both firms' equilibrium profits are nonnegative, then in particular $\pi_1^*\geq0$, so $\rho\leq\overline\rho(\eta)$.
    Combining this nonnegative-profit bound with $\pi_1^*<\pi_2^*$ if and only if $\rho>\rho^\dagger(\eta)$ proves the characterization.
\end{proof}

\subsection{Proof of Proposition \ref{prop:limiting_profit_gap_comparative_statics}}
\begin{proof}
    From the proof of Proposition \ref{prop:limiting_detrimental_prominence},
    $$
        \Gamma=\alpha\left[(p_1^*-\eta)^2-q_2(p_2^*-\eta)^2-(1-q_2)(c-\eta)\right].
    $$
    Holding $(s,\eta,\alpha)$ fixed, the pricing equilibrium and $q_2$ are fixed when $c$ varies.
    Since $0<h_2<1$, $1-q_2=\alpha(1-h_2)>0$.
    Hence $\partial\Gamma/\partial c=-\alpha(1-q_2)<0$.

    Proposition \ref{prop:limiting_comparative_statics} gives $dq_2/d\eta>0$.
    Therefore
    $$
        \frac{\partial}{\partial\eta}
        \left(\frac{\partial\Gamma}{\partial c}\right)
        =
        \frac{\partial}{\partial\eta}
        \left[-\alpha(1-q_2)\right]
        =
        \alpha\frac{dq_2}{d\eta}>0.
    $$

    Now consider the effect of $\eta$ on $\Gamma$ itself, holding $c$ fixed.

    Denote the corresponding lower bound on the unit cost by
    $$
    \begin{aligned}
        c_\Gamma(s,\alpha,\eta)
        \equiv\ 
        \eta+(p_2^*-\eta)^2
        -
        \frac{
        (1-q_2)
        +2(p_1^*-\eta)\frac{d(p_1^*-\eta)}{d\eta}
        -2q_2(p_2^*-\eta)\frac{d(p_2^*-\eta)}{d\eta}
        }{
        \frac{dq_2}{d\eta}
        } .
    \end{aligned}
    $$

    Differentiating the profit gap with respect to $\eta$ gives
    $$
    \begin{aligned}
        \frac{1}{\alpha}\frac{d\Gamma}{d\eta}
        &=
        (1-q_2)
        +2(p_1^*-\eta)\frac{d(p_1^*-\eta)}{d\eta}
        -2q_2(p_2^*-\eta)\frac{d(p_2^*-\eta)}{d\eta}\\
        &\qquad
        +\frac{dq_2}{d\eta}\left[c-\eta-(p_2^*-\eta)^2\right].
    \end{aligned}
    $$
    Since $dq_2/d\eta>0$, this expression is strictly increasing in $c$.

    Solving the right-hand side for zero gives $c=c_\Gamma(s,\alpha,\eta)$.
    Appendix Lemma \ref{lem:nonbinding_cgamma} shows that $c_\Gamma(s,\alpha,\eta)<\eta$.
    Under the maintained condition $c>\eta$, we therefore have
    $c>c_\Gamma(s,\alpha,\eta)$, and hence $d\Gamma/d\eta>0$.

\end{proof}

\subsection{Proof of Proposition \ref{prop:limiting_welfare_analysis}}
\begin{proof}
    For part (i), $c$ does not enter the limiting first-order system or the consumers' search problem. 
    Hence $p_i^*$, $D_i^*$, $S_i^*$, $q_2$, and $R^*$ are fixed when only $c$ varies. 
    Since
    $$
        \pi_1^*
        =
        \alpha\left[(p_1^*-\eta)^2-(c-\eta)\right],
        \qquad
        \pi_2^*
        =
        \alpha q_2\left[(p_2^*-\eta)^2-(c-\eta)\right],
    $$
    differentiation with respect to $c$ gives
    $$
        \frac{d\pi_1^*}{dc}=-\alpha,
        \qquad
        \frac{d\pi_2^*}{dc}=-\alpha q_2.
    $$
    The consumer problem does not contain $c$, so $dCS^*/dc=0$. 
    Similarly, $dR^*/dc=0$. 
    Finally, because $S^*=\alpha(1+q_2)$,
    $$
        \frac{d\Pi^*}{dc}
        =
        -\alpha-\alpha q_2
        =
        -S^*,
        \qquad
        \frac{dW^*}{dc}
        =
        -S^*,
    $$
    where the last equality follows from $W^*=CS^*+\Pi^*$ and $dCS^*/dc=0$.

    For part (ii), differentiate the displayed expression for $CS^*$ directly.
    Since $a$ is fixed when $\eta$ changes and $h_2=a+p_1^*-p_2^*$, the chain rule gives
    $$
        \frac{dCS^*}{d\eta}
        =
        \frac{\partial CS^*}{\partial p_1^*}\frac{dp_1^*}{d\eta}
        +\frac{\partial CS^*}{\partial p_2^*}\frac{dp_2^*}{d\eta}
        +\frac{\partial CS^*}{\partial h_2}\frac{dh_2}{d\eta},
        \qquad
        \frac{dh_2}{d\eta}=\frac{dp_1^*}{d\eta}-\frac{dp_2^*}{d\eta}.
    $$
    We evaluate the three partial derivatives from the displayed expression for $CS^*$, holding the other two arguments fixed.

    First, the derivative with respect to the integration boundary $h_2$ is zero:
    $$
    \begin{aligned}
        \frac{\partial CS^*}{\partial h_2}
        &=
        \alpha\left[
        h_2-p_1^*+\alpha\phi(p_2^*+h_2-p_1^*)-s
        -(h_2-p_1^*)
        \right]\\
        &=
        \alpha[\alpha\phi(a)-s]=0,
    \end{aligned}
    $$
    where the last equality uses $p_2^*+h_2-p_1^*=a$ and $\alpha\phi(a)=s$.

    Second, differentiating with respect to $p_1^*$ gives
    $$
    \begin{aligned}
        \frac{\partial CS^*}{\partial p_1^*}
        &=
        \alpha\Bigg\{
        \int_{p_1^*}^{h_2}
        \left[-1-\alpha\phi'(p_2^*+\varepsilon_1-p_1^*)\right]\,d\varepsilon_1
        -\int_{h_2}^{1}1\,d\varepsilon_1
        \Bigg\}\\
        &=
        -\alpha\left[1-p_1^*+\alpha\{\phi(a)-\phi(p_2^*)\}\right]\\
        &=
        -\left\{
        \alpha(1-\alpha)(1-p_1^*)
        +\alpha^2\left[1-h_2+\frac{a^2-(p_2^*)^2}{2}\right]
        \right\}
        =
        -D_1.
    \end{aligned}
    $$
    The first line uses the cancellation between the derivative of $p_1^*[\alpha\phi(p_2^*)-s]$ and the lower-limit term of the first integral.

    Third, using $\phi'(p_2^*)=-(1-p_2^*)$,
    $$
    \begin{aligned}
        \frac{\partial CS^*}{\partial p_2^*}
        &=
        \alpha(1-\alpha)\phi'(p_2^*)
        +\alpha^2 p_1^*\phi'(p_2^*)
        +\alpha^2\int_{p_1^*}^{h_2}
        \phi'(p_2^*+\varepsilon_1-p_1^*)\,d\varepsilon_1\\
        &=
        -\alpha(1-\alpha)(1-p_2^*)
        -\alpha^2\left[
        p_1^*(1-p_2^*)-\{\phi(a)-\phi(p_2^*)\}
        \right]\\
        &=
        -\left\{
        \alpha(1-\alpha)(1-p_2^*)
        +\alpha^2\left[(1-a)h_2+\frac{h_2^2-(p_1^*)^2}{2}\right]
        \right\}
        =
        -D_2.
    \end{aligned}
    $$
    Substituting these partial derivatives into the chain rule yields
    $$
        \frac{dCS^*}{d\eta}
        =
        -D_1\frac{dp_1^*}{d\eta}
        -D_2\frac{dp_2^*}{d\eta}<0.
    $$
    The strict inequality follows from $D_i>0$ and Proposition \ref{prop:limiting_comparative_statics}, which gives $dp_i^*/d\eta>0$ for $i=1,2$.
    For firm $1$, write
    $$
        \pi_1^*=(p_1^*-\eta)D_1^*-(c-\eta)S_1^*.
    $$
    Taking the total derivative with respect to $\eta$ gives
    $$
        \frac{d\pi_1^*}{d\eta}
        =
        \underbrace{\frac{\partial \pi_1^*}{\partial p_1}}_{=0}
        \frac{dp_1^*}{d\eta}
        +\frac{\partial \pi_1^*}{\partial p_2}\frac{dp_2^*}{d\eta}
        +\left.\frac{\partial \pi_1^*}{\partial\eta}\right|_{p_1^*,p_2^*}.
    $$
    The first term is zero by firm $1$'s first-order condition. 
    Holding prices fixed,
    $$
        \left.\frac{\partial \pi_1^*}{\partial\eta}\right|_{p_1^*,p_2^*}
        =
        S_1^*-D_1^*
        =
        \alpha-\alpha(p_1^*-\eta)
        =
        \alpha[1-(p_1^*-\eta)]>0.
    $$
    The strategic effect through firm $2$'s price is also positive:
    $$
        \frac{\partial \pi_1^*}{\partial p_2}
        =
        (p_1^*-\eta)\frac{\partial D_1^*}{\partial p_2}
        =
        (p_1^*-\eta)\alpha^2(1-p_2^*)>0,
    $$
    and Proposition \ref{prop:limiting_comparative_statics} gives $dp_2^*/d\eta>0$.
    Hence $d\pi_1^*/d\eta>0$.

    The total return volume can be signed using the accounting expression already derived above:
    $$
        R^*
        =
        \alpha[1-(p_1^*-\eta)]
        +\alpha q_2[1-(p_2^*-\eta)].
    $$
    Proposition \ref{prop:limiting_comparative_statics} gives
    $$
        \frac{d(p_i^*-\eta)}{d\eta}<0\quad(i=1,2),
        \qquad
        \frac{dq_2}{d\eta}>0.
    $$
    Since $1-(p_2^*-\eta)>0$, differentiating $R^*$ yields
    $$
    \begin{aligned}
        \frac{dR^*}{d\eta}
        &=
        -\alpha\frac{d(p_1^*-\eta)}{d\eta}
        +\alpha\frac{dq_2}{d\eta}[1-(p_2^*-\eta)]
        -\alpha q_2\frac{d(p_2^*-\eta)}{d\eta}
        >0.
    \end{aligned}
    $$

    For firm $2$,
    $$
        \pi_2^*
        =
        \alpha q_2\left[(p_2^*-\eta)^2-(c-\eta)\right].
    $$
    Because $\pi_2^*\geq0$ when $(p_1^*,p_2^*)$ is the equilibrium price pair, $(p_2^*-\eta)^2-(c-\eta)\geq0$.
    Therefore
    $$
    \begin{aligned}
        \frac{d\pi_2^*}{d\eta}
        =
        \alpha\frac{dq_2}{d\eta}\left[(p_2^*-\eta)^2-(c-\eta)\right]
        +\alpha q_2\left[
        2(p_2^*-\eta)\frac{d(p_2^*-\eta)}{d\eta}+1
        \right].
    \end{aligned}
    $$
    The first term is nonnegative because $dq_2/d\eta>0$.
    The second term is strictly positive: by Proposition \ref{prop:limiting_comparative_statics}, $-1<d(p_2^*-\eta)/d\eta<0$, while Proposition \ref{lem:limiting_armstrong_style} gives $p_2^*-\eta<(1-\eta)/2\leq1/2$, so
    $$
        2(p_2^*-\eta)\frac{d(p_2^*-\eta)}{d\eta}+1
        >
        1-2(p_2^*-\eta)>0.
    $$
    Hence $d\pi_2^*/d\eta>0$.
    Since $d\pi_1^*/d\eta>0$, this also implies $d\Pi^*/d\eta>0$.

    Finally, hold $c$, $s$, and $\alpha$ fixed. 
    Recall that $\pi_1^*=\alpha\left[(p_1^*-\eta)^2-(c-\eta)\right],\pi_2^*=\alpha q_2\left[(p_2^*-\eta)^2-(c-\eta)\right]$.
    Combine the consumer-surplus and profit derivatives obtained above, using the pricing first-order conditions $D_1=\alpha(p_1^*-\eta)$ and $D_2=\alpha q_2(p_2^*-\eta)$. 
    For firm~1, the sum of its profit derivative and the consumer-surplus term due to its price is
    $$
    \begin{aligned}
    &\alpha\left[1+2(p_1^*-\eta)
       \left(\frac{dp_1^*}{d\eta}-1\right)\right]
       -D_1\frac{dp_1^*}{d\eta}\\
    &\qquad=\alpha\left[1-2(p_1^*-\eta)
       +(p_1^*-\eta)\frac{dp_1^*}{d\eta}\right].
    \end{aligned}
    $$
    Applying the same calculation to firm~2 and adding the term due to the change in $q_2$ gives
    $$
    \begin{aligned}
    \frac{dW^*}{d\eta}
    &=\alpha\left[1-2(p_1^*-\eta)
       +(p_1^*-\eta)\frac{dp_1^*}{d\eta}\right]\\
    &\quad+\alpha q_2\left[1-2(p_2^*-\eta)
       +(p_2^*-\eta)\frac{dp_2^*}{d\eta}\right]\\
    &\quad+\alpha\frac{dq_2}{d\eta}
       \left[(p_2^*-\eta)^2-(c-\eta)\right]>0.
    \end{aligned}
    $$ 

    The price bound in Proposition~\ref{lem:limiting_armstrong_style}, together with $\eta<1-a$, implies $0<p_i^*-\eta<1/2$. 
    Proposition~\ref{prop:limiting_comparative_statics} gives $dp_i^*/d\eta>0$. 
    Since $\alpha>0$ and $q_2>0$, the first two terms are strictly positive. 
    The same proposition gives $dq_2/d\eta>0$. 
    Moreover, the maintained full-game equilibrium conditions imply $\pi_2^*\geq0$, so the bracket in the last term is nonnegative. 
    Hence $dW^*/d\eta>0$ on the maintained equilibrium branch.
\end{proof}

\subsection{Proof of Corollary \ref{cor: welfare_search_cost}}
\begin{proof}
    Differentiate the displayed expression for $CS^*$ with respect to $s$. 
    Since $h_2=a+p_1^*-p_2^*$ and $a$ is defined by $\alpha\phi(a)=s$, the chain rule gives
    $$
        \frac{dCS^*}{ds}
        =
        \frac{\partial CS^*}{\partial s}
        +\frac{\partial CS^*}{\partial p_1^*}\frac{dp_1^*}{ds}
        +\frac{\partial CS^*}{\partial p_2^*}\frac{dp_2^*}{ds}
        +\frac{\partial CS^*}{\partial h_2}\frac{dh_2}{ds}.
    $$

    As shown in the proof of Proposition~\ref{prop:limiting_welfare_analysis}, $\partial CS^*/\partial h_2=0$ at the equilibrium cutoff.
    Holding $p_1^*$, $p_2^*$, and $h_2$ fixed, the direct partial derivative with respect to $s$ is
    $$
    \begin{aligned}
        \frac{\partial CS^*}{\partial s}
        &=
        -(1-\alpha)-\alpha p_1^*
        -\alpha\int_{p_1^*}^{h_2}1\,d\varepsilon_1\\
        &=
        -(1-\alpha)-\alpha h_2
        =
        -q_2.
    \end{aligned}
    $$
    The price partial derivatives are the same partial derivatives computed in part (ii) from the displayed expression for $CS^*$:
    $$
        \frac{\partial CS^*}{\partial p_1^*}=-D_1,
        \qquad
        \frac{\partial CS^*}{\partial p_2^*}=-D_2,
    $$
    where $D_i$ denotes firm $i$'s retained demand evaluated at the candidate.
    Substituting these terms into the chain rule gives
    $$
        \frac{dCS^*}{ds}
        =
        -q_2-D_1\frac{dp_1^*}{ds}-D_2\frac{dp_2^*}{ds}<0,
    $$
    because $q_2>0$, $D_i>0$, and Proposition \ref{prop:limiting_comparative_statics} gives $dp_i^*/ds>0$.

    For total return volume, the first-order conditions imply $D_1^*=\alpha(p_1^*-\eta)$ and $D_2^*=\alpha q_2(p_2^*-\eta)$. 
    Since $S_1^*=\alpha$ and $S_2^*=\alpha q_2$,
    $$
        R^*
        =
        S_1^*+S_2^*-D_1^*-D_2^*
        =
        \alpha[1-(p_1^*-\eta)]
        +\alpha q_2[1-(p_2^*-\eta)].
    $$
    Along  $p_1^*=\chi(p_2^*)$, $q_2=1-\alpha+\alpha[a+\chi(p_2^*)-p_2^*]$.
    Differentiating this expression gives
    $$
        \frac{dq_2}{ds}
        =
        \alpha\left[
        \left(1+\frac{\partial \chi}{\partial a}\right)\frac{da}{ds}
        +\left(\chi'(p_2^*)-1\right)\frac{dp_2^*}{ds}
        \right].
    $$
    Here $da/ds<0$, $dp_2^*/ds>0$,
    $$
        1+\frac{\partial \chi}{\partial a}
        =
        1-\frac{\alpha(1-a)}{2}>0,
        \qquad
        \chi'(p_2^*)-1
        =
        \frac{\alpha(1-p_2^*)}{2}-1<0.
    $$
    Therefore both terms in brackets are negative, and so $dq_2/ds<0$.  Differentiating the expression for $R^*$ yields
    $$
        \frac{dR^*}{ds}
        =
        -\alpha\frac{dp_1^*}{ds}
        +\alpha\frac{dq_2}{ds}[1-(p_2^*-\eta)]
        -\alpha q_2\frac{dp_2^*}{ds}<0.
    $$
    The inequality follows from $dp_i^*/ds>0$, $dq_2/ds<0$, $q_2>0$, and $1-(p_2^*-\eta)>0$.

    It remains to sign the prominent firm's profit.
    $$
        \frac{d\pi_1^*}{ds}
        =
        2\alpha(p_1^*-\eta)\frac{dp_1^*}{ds}>0.
    $$
    This follows by differentiating
    $\pi_1^*=\alpha[(p_1^*-\eta)^2-(c-\eta)]$ with respect to $s$.
    The sign is positive because $p_1^*>\eta$ and $dp_1^*/ds>0$.

    To establish the welfare effect, hold $\alpha$, $\eta$, and $c$ fixed.
    Adding the derivatives of consumer surplus and the two firms' profits gives
    \[
    \begin{aligned}
    \frac{dW^*}{ds}
    ={}&-q_2+\alpha(p_1^*-\eta)\frac{dp_1^*}{ds}
       +\alpha q_2(p_2^*-\eta)\frac{dp_2^*}{ds}\\
       &+\alpha\frac{dq_2}{ds}\bigl[(p_2^*-\eta)^2-(c-\eta)\bigr].
    \end{aligned}
    \]
    Proposition~\ref{lem:limiting_armstrong_style} gives $0<p_i^*-\eta<1/2$.
    Both price derivatives are positive, while $dq_2/ds<0$, as shown above.
    The maintained equilibrium conditions and $p_2^*>p_1^*>\eta$ imply $(p_2^*-\eta)^2-(c-\eta)\geq0$.
    Hence
    \begin{equation}
    \frac{dW^*}{ds}
    <-q_2+\frac\alpha2\left(\frac{dp_1^*}{ds}
                              +q_2\frac{dp_2^*}{ds}\right).
    \label{eq:welfare_search_initial_bound}
    \end{equation}

    We next bound the price responses.
    Differentiating firm~1's pricing equation \eqref{eq:limiting_FOC1} and using $da/ds=-1/[\alpha(1-a)]$ gives
    \begin{equation}
    \frac{dp_1^*}{ds}=\frac12+\frac{\alpha(1-p_2^*)}{2}\frac{dp_2^*}{ds}.
    \label{eq:welfare_search_price1_derivative}
    \end{equation}
    To bound $dp_2^*/ds$, differentiate firm~2's pricing equation \eqref{eq:welfare_eta_identity} with respect to $a$.
    Since $dq_2/ds<0$ and $da/ds<0$, we have $dq_2/da>0$.
    Therefore
    \[
    \begin{aligned}
    -\frac{dp_2^*}{da}
    &=\frac{\alpha(a-p_2^*)}{2q_2}
       \left(1-\frac{dp_2^*}{da}\right)
       -\frac{\alpha(a-p_2^*)^2}{4q_2^2}\frac{dq_2}{da}\\
    &<\frac{\alpha(a-p_2^*)}{2q_2}
       \left(1-\frac{dp_2^*}{da}\right).
    \end{aligned}
    \]
    Since $-dp_2^*/da=\alpha(1-a)dp_2^*/ds$ and $2q_2>\alpha(a-p_2^*)$ by Lemma~\ref{lem:welfare_parameter_bounds}, rearranging gives
    \begin{equation}
    0<\alpha(1-a)\frac{dp_2^*}{ds}
    <\frac{\alpha(a-p_2^*)}{2q_2-\alpha(a-p_2^*)}
    <\frac14.
    \label{eq:welfare_search_price2_bound}
    \end{equation}
    The final inequality follows from $\alpha(a-p_2^*)/q_2<2/5$ in Lemma~\ref{lem:welfare_parameter_bounds}.

    Substituting \eqref{eq:welfare_search_price1_derivative} and the upper bound on $dp_2^*/ds$ in \eqref{eq:welfare_search_initial_bound} gives
    \[
    \begin{aligned}
    \frac{dW^*}{ds}
    &<-q_2+\frac\alpha4
      +\frac{[\alpha(1-p_2^*)/2+q_2]\alpha(a-p_2^*)}
             {2(1-a)[2q_2-\alpha(a-p_2^*)]}\\
    &=-\frac{
       4(1-a)q_2^2-\alpha q_2[1-p_2^*+2(1-a)(a-p_2^*)]
       -\alpha^2(a-p_2^*)^2/2}
       {2(1-a)[2q_2-\alpha(a-p_2^*)]}.
    \end{aligned}
    \]
    The denominator is positive.
    To sign the numerator, use $1-p_2^*=(1-a)+(a-p_2^*)$.
    Holding $a-p_2^*$ fixed, the numerator increases in $1-a$ and $q_2$ and decreases in $\alpha$.
    To see the first two signs, Lemma~\ref{lem:welfare_parameter_bounds} gives $1-a>1/4$ and $0<a-p_2^*<5/18$.
    The lower bound for $q_2$ in the same lemma also implies $q_2>1/2$, since $a-p_2^*>0$.
    Together, these bounds imply
    \[
    \begin{aligned}
    4q_2&>2>1+2(a-p_2^*),\\
    8(1-a)q_2&>4(1-a)
    >[1+2(a-p_2^*)](1-a)+(a-p_2^*).
    \end{aligned}
    \]
    Since $\alpha<1$, these inequalities make the partial derivatives of the numerator with respect to $1-a$ and $q_2$ positive.
    Setting $1-a=1/4$ and $\alpha=1$, and using the lower bound for $q_2$ in Lemma~\ref{lem:welfare_parameter_bounds}, the numerator is therefore greater than
    \[
    \frac{9(a-p_2^*)^4-236(a-p_2^*)^2-48(a-p_2^*)+32}{256}.
    \]
    Since $a-p_2^*<5/18$, this fraction is greater than
    \[
    \frac{32-236(5/18)^2-48(5/18)}{256}
    =\frac{37}{81\cdot256}>0.
    \]
    The numerator is positive, and thus $dW^*/ds<0$ under the maintained full-game equilibrium conditions.

\end{proof}

\subsection{Proof of Corollary \ref{cor: welfare_match_probability}}
\begin{proof}
    Differentiate the displayed expression for $CS^*$ with respect to $\alpha$. 
    Since $h_2=a+p_1^*-p_2^*$ and $a$ is defined by
    $\alpha\phi(a)=s$, the chain rule gives
    $$
        \frac{dCS^*}{d\alpha}
        =
        \frac{\partial CS^*}{\partial \alpha}
        +\frac{\partial CS^*}{\partial p_1^*}\frac{dp_1^*}{d\alpha}
        +\frac{\partial CS^*}{\partial p_2^*}\frac{dp_2^*}{d\alpha}
        +\frac{\partial CS^*}{\partial h_2}\frac{dh_2}{d\alpha}.
    $$

    As shown in the proof of Proposition~\ref{prop:limiting_welfare_analysis}, $\partial CS^*/\partial h_2=0$ at the equilibrium cutoff.
    To keep the notation compact, denote
    $$
    \begin{aligned}
        \mathcal{B}
        &\equiv
        p_1^*[\alpha\phi(p_2^*)-s]\\
        &\quad
        +\int_{p_1^*}^{h_2}
        \left[
        \varepsilon_1-p_1^*+\alpha\phi(p_2^*+\varepsilon_1-p_1^*)-s
        \right]\,d\varepsilon_1\\
        &\quad
        +\int_{h_2}^1(\varepsilon_1-p_1^*)\,d\varepsilon_1.
    \end{aligned}
    $$
    Holding $p_1^*$, $p_2^*$, and $h_2$ fixed, the direct partial derivative with respect to $\alpha$ is
    $$
    \begin{aligned}
        \frac{\partial CS^*}{\partial \alpha}
        &=
        \mathcal{B}
        -[\alpha\phi(p_2^*)-s]\\
        &\quad
        +(1-\alpha)\phi(p_2^*)\\
        &\quad
        +\alpha\left[
        p_1^*\phi(p_2^*)
        +\int_{p_1^*}^{h_2}
        \phi(p_2^*+\varepsilon_1-p_1^*)\,d\varepsilon_1
        \right].
    \end{aligned}
    $$
    The expression formed by $\mathcal{B}$ minus $[\alpha\phi(p_2^*)-s]$ is nonnegative.  Using $s=\alpha\phi(a)$, it can be written as
    $$
    \begin{aligned}
        &\mathcal{B}-[\alpha\phi(p_2^*)-s]\\
        &=
        \int_{p_1^*}^{h_2}
        \left[
        \varepsilon_1-p_1^*
        -\alpha\{\phi(p_2^*)-\phi(p_2^*+\varepsilon_1-p_1^*)\}
        \right]\,d\varepsilon_1\\
        &\quad
        +\int_{h_2}^1
        \left[
        \varepsilon_1-p_1^*
        -\alpha\{\phi(p_2^*)-\phi(a)\}
        \right]\,d\varepsilon_1
        \geq0.
    \end{aligned}
    $$
    To see the last inequality, use $\phi(z)=(1-z)^2/2$ and hence $-\phi'(z)=1-z\leq1$ on the relevant interval. 
    If $\varepsilon_1\in[p_1^*,h_2]$, then $p_2^*+\varepsilon_1-p_1^*\in[p_2^*,a]$, and therefore
    $$
    \begin{aligned}
        \phi(p_2^*)-\phi(p_2^*+\varepsilon_1-p_1^*)
        &=
        \int_{p_2^*}^{p_2^*+\varepsilon_1-p_1^*}(1-z)\,dz\\
        &\leq
        \int_{p_2^*}^{p_2^*+\varepsilon_1-p_1^*}1\,dz
        =
        \varepsilon_1-p_1^* .
    \end{aligned}
    $$
    Thus the first integrand is nonnegative because $0<\alpha<1$:
    $$
        \varepsilon_1-p_1^*
        -\alpha\{\phi(p_2^*)-\phi(p_2^*+\varepsilon_1-p_1^*)\}
        \geq
        (1-\alpha)(\varepsilon_1-p_1^*)\geq0 .
    $$
    Similarly,
    $$
        \phi(p_2^*)-\phi(a)
        =
        \int_{p_2^*}^{a}(1-z)\,dz
        \leq a-p_2^* .
    $$
    If $\varepsilon_1\in[h_2,1]$, then $\varepsilon_1-p_1^*\geq h_2-p_1^*=a-p_2^*$, so the second integrand is nonnegative:
    $$
    \begin{aligned}
        \varepsilon_1-p_1^*
        -\alpha\{\phi(p_2^*)-\phi(a)\}
        &\geq
        \varepsilon_1-p_1^*-\alpha(a-p_2^*)\\
        &\geq
        (1-\alpha)(a-p_2^*)>0 .
    \end{aligned}
    $$
    The remaining two terms in $\partial CS^*/\partial\alpha$ are strictly positive on the active branch.
    Indeed, $(1-\alpha)\phi(p_2^*)>0$ because $0<\alpha<1$ and $p_2^*<1$.
    Moreover, $h_2-p_1^*=a-p_2^*>0$, so the interval $[p_1^*,h_2]$ has positive length. 
    On this interval, $p_2^*+\varepsilon_1-p_1^*\leq a<1$, and hence $\phi(p_2^*+\varepsilon_1-p_1^*)>0$. 
    Therefore
    $$
        \alpha\left[
        p_1^*\phi(p_2^*)
        +\int_{p_1^*}^{h_2}
        \phi(p_2^*+\varepsilon_1-p_1^*)\,d\varepsilon_1
        \right]>0.
    $$
    Therefore $\partial CS^*/\partial\alpha>0$.

    The price partial derivatives are the same as above:
    $$
        \frac{\partial CS^*}{\partial p_1^*}=-D_1,
        \qquad
        \frac{\partial CS^*}{\partial p_2^*}=-D_2.
    $$
    Substituting these terms into the chain rule gives
    $$
        \frac{dCS^*}{d\alpha}
        =
        \frac{\partial CS^*}{\partial\alpha}
        -D_1\frac{dp_1^*}{d\alpha}
        -D_2\frac{dp_2^*}{d\alpha}.
    $$
    Proposition \ref{prop:limiting_comparative_statics} gives $dp_i^*/d\alpha<0$, so the two price terms are positive as well.
    Thus $dCS^*/d\alpha>0$.

    We next establish the effects on follow-on search, total returns, and welfare, holding $s$, $\eta$, and $c$ fixed.
    The expression for $dp_1^*/d\alpha$ in the proof of Proposition~\ref{prop:limiting_comparative_statics} can be written as
    \begin{equation}
    -\frac{dp_1^*}{d\alpha}
    =\frac{(1-p_2^*)^2}{4}
     -\frac{\alpha(1-p_2^*)}{2}\frac{dp_2^*}{d\alpha},
    \label{eq:match_price1_derivative}
    \end{equation}
    Using $da/d\alpha=(1-a)/(2\alpha)$ to differentiate $q_2=1-\alpha(1-a+p_2^*-p_1^*)$ gives
    \begin{equation}
    \frac{dq_2}{d\alpha}
    =-\frac{1-a}{2}-(p_2^*-p_1^*)
     +\alpha\frac{dp_1^*}{d\alpha}
     -\alpha\frac{dp_2^*}{d\alpha}.
    \label{eq:match_search_derivative}
    \end{equation}
    Differentiating firm~2's pricing equation \eqref{eq:welfare_eta_identity} gives
    \[
    \begin{aligned}
    -2\frac{dp_2^*}{d\alpha}
    ={}&\frac{(a-p_2^*)^2}{2q_2}
      +\frac{\alpha(a-p_2^*)}{q_2}
       \left(\frac{1-a}{2\alpha}-\frac{dp_2^*}{d\alpha}\right)\\
     &-\frac{\alpha(a-p_2^*)^2}{2q_2^2}\frac{dq_2}{d\alpha}.
    \end{aligned}
    \]
    First, substitute the expression for $dp_1^*/d\alpha$ in \eqref{eq:match_price1_derivative} into the search derivative \eqref{eq:match_search_derivative}:
    \begin{equation}
    \begin{aligned}
    \frac{dq_2}{d\alpha}
    ={}&-\left[\frac{1-a}{2}+p_2^*-p_1^*
                 +\frac{\alpha(1-p_2^*)^2}{4}\right]\\
       &-\alpha\left[1-\frac{\alpha(1-p_2^*)}{2}\right]
       \frac{dp_2^*}{d\alpha}.
    \end{aligned}
    \label{eq:match_search_reduced}
    \end{equation}
    Next, substitute this expression for $dq_2/d\alpha$ into the differentiated firm~2 pricing equation above.
    Multiplying both sides by $2\alpha$ and collecting the terms containing $dp_2^*/d\alpha$ gives
    \[
    \begin{aligned}
    -\alpha\frac{dp_2^*}{d\alpha}
    \biggl[{}&4-\frac{2\alpha(a-p_2^*)}{q_2}
      +\left(1-\frac{\alpha(1-p_2^*)}{2}\right)
       \frac{\alpha^2(a-p_2^*)^2}{q_2^2}\biggr]\\
    ={}&\frac{\alpha(a-p_2^*)(1-p_2^*)}{q_2}\\
    &+\frac{\alpha^2(a-p_2^*)^2}{q_2^2}
      \left[\frac{1-a}{2}+p_2^*-p_1^*
       +\frac{\alpha(1-p_2^*)^2}{4}\right].
    \end{aligned}
    \]

    We now use the bounds in Lemma~\ref{lem:welfare_parameter_bounds} to bound the price responses.
    Firm~1's pricing equation \eqref{eq:limiting_FOC1}, the equivalent form of firm~2's pricing equation \eqref{eq:welfare_eta_identity}, and $p_2^*>p_1^*$ imply
    \[
    \begin{aligned}
    0<p_2^*-p_1^*
    &=\frac{\alpha[(1-p_2^*)^2-(1-a)^2]}4
      -\frac{\alpha(a-p_2^*)^2}{4q_2}\\
    &<\frac{(19/36)^2-(1/4)^2}{4}=\frac{35}{648}.
    \end{aligned}
    \]
    The inequality uses $\alpha<1$, $1-p_2^*<19/36$, and $1-a>1/4$.
    Together with $1-a<1/2$, these bounds give
    \[
    \frac{1-a}{2}+p_2^*-p_1^*+\frac{\alpha(1-p_2^*)^2}{4}
    <\frac14+\frac{35}{648}+\frac{(19/36)^2}{4}
    =\frac{1937}{5184}<\frac38.
    \]
    We use these bounds on both sides of the equation for $-\alpha\,dp_2^*/d\alpha$ above.
    First, $0<\alpha<1$ and $0<1-p_2^*<1$ imply $1-\alpha(1-p_2^*)/2>0$.
    The coefficient on the left-hand side is therefore bounded below by
    \[
    \begin{aligned}
    &4-\frac{2\alpha(a-p_2^*)}{q_2}
      +\left(1-\frac{\alpha(1-p_2^*)}{2}\right)
       \frac{\alpha^2(a-p_2^*)^2}{q_2^2}\\
    &\quad>4-\frac{2\alpha(a-p_2^*)}{q_2}
      >4-2\left(\frac25\right)=\frac{16}{5}.
    \end{aligned}
    \]
    The last inequality uses $\alpha(a-p_2^*)/q_2<2/5$ from Lemma~\ref{lem:welfare_parameter_bounds}.
    For the right-hand side, the same lemma gives $1-p_2^*<19/36$.
    Using these two bounds and the upper bound $3/8$ established above, we obtain
    \[
    \begin{aligned}
    &\frac{\alpha(a-p_2^*)(1-p_2^*)}{q_2}
     +\frac{\alpha^2(a-p_2^*)^2}{q_2^2}
      \left[\frac{1-a}{2}+p_2^*-p_1^*
             +\frac{\alpha(1-p_2^*)^2}{4}\right]\\
    &\quad<\frac25\frac{19}{36}+\left(\frac25\right)^2\frac38
     =\frac{61}{225}.
    \end{aligned}
    \]
    Dividing the right-hand side by the positive coefficient on the left, we obtain
    \begin{equation}
    0<-\alpha\frac{dp_2^*}{d\alpha}
    <\frac{61/225}{16/5}=\frac{61}{720}<\frac7{80},
    \label{eq:match_price2_response_bound}
    \end{equation}
    where the first inequality follows from Proposition~\ref{prop:limiting_comparative_statics}.
    Using \eqref{eq:match_price1_derivative} and $\alpha<1$ then gives
    \begin{equation}
    0<-\alpha\frac{dp_1^*}{d\alpha}
    <\frac{(19/36)^2}{4}+\frac{19}{72}\frac7{80}
    =\frac{4807}{51840}<\frac3{32}.
    \label{eq:match_price1_response_bound}
    \end{equation}

    We first consider search and returns.
    Proposition~\ref{prop:limiting_price_comparison} gives $p_2^*>p_1^*$, and Proposition~\ref{prop:limiting_comparative_statics} gives $dp_1^*/d\alpha<0$.
    Thus \eqref{eq:match_search_derivative} and \eqref{eq:match_price2_response_bound} imply
    \[
    \frac{dq_2}{d\alpha}
    <-\frac{1-a}{2}-\alpha\frac{dp_2^*}{d\alpha}
    <-\frac18+\frac7{80}<0.
    \]
    Equation~\eqref{eq:match_search_reduced} also gives
    \[
    \frac{dq_2}{d\alpha}>-\frac38.
    \]
    The expression in its first square brackets is less than $3/8$, as shown above, and its last term is positive because $dp_2^*/d\alpha<0$.
    Since $S_2=\alpha q_2$, $q_2>1/2$ by Lemma~\ref{lem:welfare_parameter_bounds}, and $0<\alpha<1$,
    \[
    \frac{dS_2}{d\alpha}
    =q_2+\alpha\frac{dq_2}{d\alpha}
    >\frac12-\frac38>0.
    \]
    Firm~1's initial sales also increase because $S_1=\alpha$.
    Each firm's return rate is $1-(p_i^*-\eta)$ and rises as its price falls.
    Both firms therefore make more initial sales and receive returns on a larger share of those sales, so $dR^*/d\alpha>0$.

    We now turn to welfare.
    Evaluating the integrals in the expression for $\partial CS^*/\partial\alpha$ above, holding prices and $h_2$ fixed, gives
    \[
    \begin{aligned}
    \frac{\partial CS^*}{\partial\alpha}
    ={}&(1-p_2^*)^2\left[1-\frac{2\alpha(1-p_2^*)}{3}\right]\\
      &+(p_2^*-p_1^*)\left[1-p_2^*-\alpha(1-p_2^*)^2
                           +\frac{\alpha(1-a)^2}{2}\right]\\
      &+\frac{(p_2^*-p_1^*)^2}{2}+\frac{\alpha(1-a)^3}{6}\\
    >{}&(1-p_2^*)^2-\frac{2(1-p_2^*)^3}{3}+\frac{(1-a)^3}{6}.
    \end{aligned}
    \]
    The terms involving $p_2^*-p_1^*$ are positive because $0<1-p_2^*<1$ and $\alpha<1$.
    In the remaining terms, the coefficient of $\alpha$ is negative because $1-p_2^*>1-a>0$; replacing $\alpha$ by $1$ therefore gives a lower bound.

    The restrictions $\eta<\min\{1-a,2a-1\}\leq1/3$ and $p_2^*<(1+\eta)/2$ imply $0<p_2^*<2/3$.
    Also $\eta<3-4a$ gives $1-a>(1+\eta)/4>p_2^*/2$.
    It follows that
    \[
    \frac{\partial CS^*}{\partial\alpha}
    >(1-p_2^*)^2-\frac{2(1-p_2^*)^3}{3}
      +\frac{(p_2^*)^3}{48}>\frac5{54}.
    \]
    For the last inequality, the expression between the inequality signs decreases in $p_2^*\in(0,2/3)$: its derivative is
    $-2p_2^*(1-p_2^*)+(p_2^*)^2/16<0$, and its value at $p_2^*=2/3$ is $5/54$.

    The consumer-surplus derivative derived above is
    \[
    \frac{dCS^*}{d\alpha}
    =\frac{\partial CS^*}{\partial\alpha}
      -\alpha(p_1^*-\eta)\frac{dp_1^*}{d\alpha}
      -\alpha q_2(p_2^*-\eta)\frac{dp_2^*}{d\alpha}.
    \]
    Adding the profit derivatives gives
    \[
    \begin{aligned}
    \frac{dW^*}{d\alpha}
    ={}&\frac{\partial CS^*}{\partial\alpha}
      +[(p_1^*-\eta)^2-(c-\eta)]\\
    &+\frac{dS_2}{d\alpha}[(p_2^*-\eta)^2-(c-\eta)]\\
    &+\alpha(p_1^*-\eta)\frac{dp_1^*}{d\alpha}
     +\alpha q_2(p_2^*-\eta)\frac{dp_2^*}{d\alpha}.
    \end{aligned}
    \]
    The full-game conditions imply that both profit brackets are nonnegative.
    Using $dS_2/d\alpha>0$, $0<p_i^*-\eta<1/2$, $q_2<1$, and the price-response bounds \eqref{eq:match_price2_response_bound}--\eqref{eq:match_price1_response_bound}, we obtain
    \[
    \frac{dW^*}{d\alpha}
    >\frac5{54}-\frac12\left(\frac3{32}+\frac7{80}\right)
    =\frac{17}{8640}>0.
    \]

\end{proof}

\subsection{Proof of Corollary \ref{prop:positive_h_extension}}
\begin{proof}
    We first construct the local interior pricing candidate and then verify the remaining claims in the proposition.

    \textit{Active branch and demand.}
    After buying product~$1$, the relevant fallback payoff is
    $$
        v(x;h)=\max\{x,-h\}.
    $$
    Let $\mathcal V_{0,h}(p_2^*)$ and $\mathcal B_{1,h}(p_1;p_2^*)$ denote, respectively, the no-product search value and the ex ante value of buying a matched product~$1$ when the hassle cost is $h$.
    The prominent-product buy slack is
    \[
        \mathcal G_{1,h}(p_1,p_2^*)
        \equiv
        \mathcal B_{1,h}(p_1;p_2^*)-\mathcal V_{0,h}(p_2^*).
    \]
    Proposition~\ref{prop: general consumer rule} implies that matched consumers buy product~$1$ iff this slack is positive.
    The search cutoff after buying product~$1$ and the product~$2$ retention corner are
    \[
        H_1(p_1,p_2^*)\equiv A(h)+p_1-p_2^*,
        \qquad
        G_2(p_2,p_2^*)\equiv A(h)+p_2-p_2^*.
    \]
    The positive-$h$ active branch is the open set of actual prices $(p_1,p_2)$ and belief $p_2^*$ satisfying
    \begin{equation}
        \begin{aligned}
            &0<p_1-h<H_1(p_1,p_2^*)<1,\\[2mm]
            &0<p_2-h<G_2(p_2,p_2^*)<1,\\[2mm]
            &\mathcal G_{1,h}(p_1,p_2^*)>0,\\[2mm]
            &\mathcal V_{0,h}(p_2^*)>0.
        \end{aligned}
        \label{eq:positive_h_active_branch_conditions}
    \end{equation}
    The first two lines keep the return, search, and retention regions interior; the last two lines are the product~$1$ buy condition and the no-product search condition.

    In a firm~$2$ deviation, consumers search using the belief $p_2^*$ but buy after observing the actual price $p_2$.
    The purchase boundary is nevertheless slack for local deviations around the candidate.
    Indeed, Proposition~\ref{prop: general consumer rule} gives the purchase condition $p_2+v(x;h)<B(h)$.
    In the no-product state, $p_2^*<A(h)+h<B(h)+h$ implies $p_2-h<B(h)$ for $p_2$ near $p_2^*$.
    If the consumer holds product~$1$, search requires $x<A(h)-p_2^*$, so $p_2+x<A(h)+(p_2-p_2^*)<B(h)$ for $p_2$ close enough to $p_2^*$.
    Thus every locally searched matched product~$2$ is bought.

    Retained demand is computed by match state.
    In states $(1,0)$ and $(0,1)$, the only matched product is kept iff its taste value exceeds $p_i-h$, generating retained demand $\alpha(1-\alpha)(1-p_i+h)$ for firm~$i$.
    In state $(1,1)$, the consumer buys product~$1$ first.
    If $\varepsilon_1\geq H_1(p_1,p_2^*)$, she keeps product~$1$ and does not search.
    If $p_1-h\leq\varepsilon_1<H_1(p_1,p_2^*)$, she searches firm~$2$ while holding product~$1$ and keeps the product with the higher net surplus.
    Hence firm~$1$'s competing-branch demand is
    $$
    \int_{p_1-h}^{H_1(p_1,p_2^*)}\Pr(\varepsilon_2<\varepsilon_1-p_1+p_2)\,d\varepsilon_1
    =\frac{G_2(p_2,p_2^*)^2-(p_2-h)^2}{2}.
    $$
    Firm~$2$'s competing-branch demand can be computed by fixing $\varepsilon_2$.
    If $\varepsilon_2>G_2(p_2,p_2^*)$, product~$2$ beats product~$1$ for every consumer in the search region.
    Firm~$2$ therefore keeps all $H_1(p_1,p_2^*)$ consumers who search, yielding
    $H_1(p_1,p_2^*)[1-G_2(p_2,p_2^*)]$.

    Now fix $\varepsilon_2\in[p_2-h,G_2(p_2,p_2^*)]$.
    Consumers who already returned product~$1$ keep product~$2$.
    Consumers still holding product~$1$ keep product~$2$ iff
    $$
        \varepsilon_2-p_2>\varepsilon_1-p_1,
        \qquad\text{or equivalently}\qquad
        \varepsilon_1<\varepsilon_2+p_1-p_2 .
    $$
    For this fixed $\varepsilon_2$, the relevant variable is $\varepsilon_1$.
    All consumers with $\varepsilon_1<p_1-h$ have already returned product~$1$, so they keep product~$2$.
    Among consumers with $p_1-h\leq\varepsilon_1<H_1(p_1,p_2^*)$, product~$2$ is kept exactly when
    $\varepsilon_1<\varepsilon_2+p_1-p_2$.
    This cutoff is inside the interval $[p_1-h,H_1(p_1,p_2^*)]$:
    the lower bound follows from $\varepsilon_2\geq p_2-h$, and the upper bound follows from $\varepsilon_2\leq G_2(p_2,p_2^*)$.
    Therefore the set of $\varepsilon_1$ values leading firm~$2$ to keep the consumer is
    $[0,\varepsilon_2+p_1-p_2]$, which has measure $\varepsilon_2+p_1-p_2$.
    Therefore
    $$
    \int_{p_2-h}^{G_2(p_2,p_2^*)}(\varepsilon_2+p_1-p_2)\,d\varepsilon_2
    =\frac{H_1(p_1,p_2^*)^2-(p_1-h)^2}{2}.
    $$
    Combining the unmatched and both-matched states gives
    \[
    \begin{aligned}
        D_1(p_1,p_2;p_2^*)
        &=
        \alpha(1-\alpha)(1-p_1+h)
        +\alpha^2\Bigg[
        1-H_1(p_1,p_2^*)
        +\frac{G_2(p_2,p_2^*)^2-(p_2-h)^2}{2}
        \Bigg],\\[2mm]
        S_1&=\alpha,
    \end{aligned}
    \]
    and
    \[
    \begin{aligned}
        D_2(p_1,p_2;p_2^*)
        &=
        \alpha(1-\alpha)(1-p_2+h)
        +\alpha^2\Bigg[
        \bigl(1-G_2(p_2,p_2^*)\bigr)H_1(p_1,p_2^*)
        +\frac{H_1(p_1,p_2^*)^2-(p_1-h)^2}{2}
        \Bigg],\\[2mm]
        S_2(p_1,p_2^*)&=\alpha q_2(p_1,p_2^*),
        \qquad
        q_2(p_1,p_2^*)\equiv1-\alpha+\alpha H_1(p_1,p_2^*).
    \end{aligned}
    \]

    Initial sales are $S_1=\alpha$ and $S_2=\alpha q_2(p_1,p_2^*)$: all matched consumers buy product~$1$, while product~$2$ is bought only when it matches and is searched.

    Profits are
    \[
        \pi_1(p_1,p_2;p_2^*)=(p_1-\eta)D_1(p_1,p_2;p_2^*)-(c-\eta)\alpha,
    \]
    \[
        \pi_2(p_1,p_2;p_2^*)=(p_2-\eta)D_2(p_1,p_2;p_2^*)-(c-\eta)\alpha q_2(p_1,p_2^*).
    \]

    \textbf{Local first-order branch and convergence.}
    \textit{First-order conditions.}
    For firm~$1$, fix $p_2$ and $p_2^*$ and vary $p_1$.
    Since $\partial H_1/\partial p_1=1$ and $G_2$ does not depend on $p_1$,
    $\partial D_1/\partial p_1=-\alpha(1-\alpha)-\alpha^2=-\alpha$.
    Hence firm~$1$'s first-order condition is $D_1(p_1,p_2;p_2^*)-\alpha(p_1-\eta)=0$.
    Substituting the expression for $D_1$ and rearranging gives firm~$1$'s active best response
    \[
        p_1=b_{1,h}(p_2,p_2^*)\equiv
        \frac{1}{2}\left[
        1+(1-\alpha)h+\eta-\alpha(A(h)-p_2^*)
        +\frac{\alpha}{2}\left(G_2(p_2,p_2^*)^2-(p_2-h)^2\right)
        \right].
    \]

    For firm~$2$, hold $p_2^*$ fixed and vary $p_2$.
    Since $H_1(p_1,p_2^*)$ uses the belief $p_2^*$, it does not vary with $p_2$, while $\partial G_2/\partial p_2=1$.
    Hence
    \[
        \frac{\partial D_2(p_1,p_2;p_2^*)}{\partial p_2}
        =
        -\alpha(1-\alpha)
        +\alpha^2 H_1(p_1,p_2^*)\cdot(-1)
        =
        -S_2(p_1,p_2^*),
    \]
    and firm~$2$'s first-order condition is $D_2(p_1,p_2;p_2^*)/S_2(p_1,p_2^*)=p_2-\eta$.

    To obtain the candidate prices, impose the consistency condition $p_2=p_2^*$ and $p_1=p_1^*$, where $(p_1^*,p_2^*)$ denote the candidate prices. At these prices $G_2(p_2^*,p_2^*)=A(h)$ and $H_1(p_1^*,p_2^*)=A(h)+p_1^*-p_2^*$.
    Firm~$1$'s first-order condition reduces to its candidate best response as a function of $p_2^*$:
    \[
        p_1^*=b_{1,h}(p_2^*)\equiv
        \frac{1}{2}\left[
        1+(1-\alpha)h+\eta-\alpha(A(h)-p_2^*)
        +\frac{\alpha}{2}\left(A(h)^2-(p_2^*-h)^2\right)
        \right].
    \]
    Define the candidate search cutoff and initial-sales probability along this best response by
    \[
        h_{2,h}(p_2^*)\equiv A(h)+b_{1,h}(p_2^*)-p_2^*,
        \qquad
        q_{2,h}(p_2^*)\equiv1-\alpha+\alpha h_{2,h}(p_2^*).
    \]
    At the local fixed point, these same objects are denoted by
    $h_2^*(h)\equiv h_{2,h}(p_2^*(h))$ and
    $q_2^*(h)\equiv q_{2,h}(p_2^*(h))$.
    For firm~$2$, we simplify $D_2/S_2$ at the candidate prices. At $p_2=p_2^*$, the demand and initial sales are
    \[
        D_2=\alpha(1-\alpha)(1-p_2^*+h)+\alpha^2\Bigl[H_1(1-A(h))+\frac{H_1^2-(p_1^*-h)^2}{2}\Bigr],
        \qquad
        S_2=\alpha q_{2,h}(p_2^*).
    \]
    Using $H_1-(p_1^*-h)=A(h)-p_2^*+h$, the term in brackets can be rewritten as
    \[
        \begin{aligned}
        H_1(1-A(h))+\frac{H_1^2-(p_1^*-h)^2}{2}
        &=H_1(1-A(h))+\frac{(H_1-p_1^*+h)(H_1+p_1^*-h)}{2}\\[1mm]
        &=H_1(1-A(h))+\frac{(A(h)-p_2^*+h)(H_1+p_1^*-h)}{2}\\[1mm]
        &=H_1(1-p_2^*+h)-\frac{(A(h)-p_2^*+h)^2}{2}.
        \end{aligned}
    \]
    Substituting back,
    \[
        \begin{aligned}
        \frac{D_2}{S_2}
        &=\frac{\alpha(1-\alpha)(1-p_2^*+h)+\alpha^2\bigl[H_1(1-p_2^*+h)-(A(h)-p_2^*+h)^2/2\bigr]}{\alpha q_{2,h}(p_2^*)}\\[2mm]
        &=\frac{(1-\alpha+\alpha H_1)(1-p_2^*+h)-\alpha(A(h)-p_2^*+h)^2/2}{q_{2,h}(p_2^*)}\\[2mm]
        &=1-p_2^*+h-\frac{\alpha(A(h)-p_2^*+h)^2}{2q_{2,h}(p_2^*)}.
        \end{aligned}
    \]
    Substituting this and $p_1^*=b_{1,h}(p_2^*)$ into firm~$2$'s first-order condition yields the fixed-point equation
    \[
        p_2^*=B_{2,h}(p_2^*)
        \equiv
        \frac{1+h+\eta}{2}
        -\frac{\alpha(A(h)-p_2^*+h)^2}{4q_{2,h}(p_2^*)}.
    \]
    \textit{Interior domain.}
    The first-order conditions derived above are valid only when both $(p_1,p_2^*)$ and $(p_1,p_2)$ lie in the active branch. 
    For the benchmark case, Proposition~\ref{lem:limiting_armstrong_style} shows that the solution lies strictly in the interior. Because the positive-$h$ system depends continuously on $h$, the same holds for all sufficiently small $h$: if the belief $p_2^*$ is near its benchmark value, then $p_1^*=b_{1,h}(p_2^*)$ also lies in the interior. Consequently the composite map $B_{2,h}$ is well defined on a neighborhood of the benchmark price $p_2^*$ for small $h$.

    \textit{Local existence and uniqueness.}
    Differentiating $b_{1,h}$ gives
    \[
        b_{1,h}'(p_2^*)=\frac{\alpha(1-p_2^*+h)}{2}\in\Bigl(0,\frac{1}{2}\Bigr),
    \]
    where $b_{1,h}'>0$ uses $p_2^*<1+h$ and $b_{1,h}'<1/2$ uses $\alpha<1$ together with $p_2^*>h$.
    The active branch gives $A(h)-p_2^*+h>0$, and
    \[
        q_{2,h}(p_2^*)-\alpha(A(h)-p_2^*+h)
        =
        1-\alpha+\alpha(b_{1,h}(p_2^*)-h)>0,
    \]
    so
    \[
        0<
        \frac{\alpha(A(h)-p_2^*+h)}{q_{2,h}(p_2^*)}
        <1 .
    \]
    Also $h_{2,h}'(p_2^*)=b_{1,h}'-1$ and $q_{2,h}'(p_2^*)=\alpha(b_{1,h}'-1)$.
    Differentiating $B_{2,h}$ and using $q_{2,h}'/\alpha=b_{1,h}'-1$, we obtain
    \[
        \begin{aligned}
        B_{2,h}'(p_2^*)
        &=\frac{\alpha(A(h)-p_2^*+h)}{4q_{2,h}(p_2^*)}
        \left[2+\frac{(A(h)-p_2^*+h)q_{2,h}'(p_2^*)}{q_{2,h}(p_2^*)}\right]\\[2mm]
        &=\frac{\alpha(A(h)-p_2^*+h)}{4q_{2,h}(p_2^*)}
        \left[
        2-\frac{\alpha(A(h)-p_2^*+h)}{q_{2,h}(p_2^*)}
        \bigl(1-b_{1,h}'(p_2^*)\bigr)
        \right].
        \end{aligned}
    \]
    The fraction in front of the bracket lies in $(0,1/4)$, and the bracket lies in $(1,2)$ because $0<b_{1,h}'<1/2$.
    Therefore $0<B_{2,h}'(p_2^*)<1/2$ wherever the local active-branch inequalities hold.

    At $h=0$, the fixed-point equation coincides with the limiting reduced system of Proposition~\ref{lem:limiting_armstrong_style}, whose solution is $p_2^*$.
    Since
    \[
        1-B_{2,0}'(p_2^*)\neq0,
    \]
    the implicit-function theorem gives a neighborhood $U_2$ of $p_2^*$ and $\bar h>0$ such that, for every $h\in(0,\bar h)$, there is a unique fixed point $p_2^*(h)\in U_2$.
    Setting $p_1^*(h)=b_{1,h}(p_2^*(h))$ gives a unique local interior pricing candidate in a neighborhood $U$ of $(p_1^*,p_2^*)$, and $p_i^*(h)\to p_i^*$.
    Any interior active-branch solution to the first-order system in $U$ must satisfy these equations, so this candidate is the unique local interior solution.

    \textit{Active-branch conditions.}
    At $h=0$, the chain inequalities in \eqref{eq:positive_h_active_branch_conditions} reduce to
    \[
        0<p_1^*<H_1^0(p_1^*)<1,\qquad
        0<p_2^*<G_2^0(p_2^*)<1,\qquad
        \mathcal G_{1,0}(p_1^*,p_2^*)>0,\qquad
        \mathcal V_{0,0}(p_2^*)>0,
    \]
    where $H_1^0(p_1)=a+p_1-p_2^*$, $G_2^0(p_2)=a+p_2-p_2^*$, and $\mathcal G_{1,0}$ and $\mathcal V_{0,0}$ are the right limits of the corresponding positive-$h$ objects.
    In particular, $\mathcal V_{0,0}(p_2^*)=-s+\alpha(1-p_2^*)^2/2$.
    The chain inequalities are strict by Proposition~\ref{lem:limiting_armstrong_style} and Proposition~\ref{prop:limiting_price_comparison}: $0<p_i^*<a<1$, $p_2^*<a$, and $p_1^*<p_2^*$ imply
    $0<p_1^*<H_1^0(p_1^*)<1$ and $0<p_2^*<G_2^0(p_2^*)<1$.
    The product-$1$ buy slack is strict by Corollary~\ref{prop:limiting_consumer_rule}, because $p_1^*<1$.
    Finally, since $p_2^*<a$ and $\phi$ is strictly decreasing on the relevant interval,
    \[
        \phi(p_2^*)>\phi(a)=s/\alpha,
    \]
    hence $\mathcal V_{0,0}(p_2^*)=\alpha\phi(p_2^*)-s>0$.
    Because $p_i^*(h)\to p_i^*$, $A(h)\to a$, and the positive-$h$ value functions converge to their right limits, every term in these inequalities is continuous in $h$; strict inequalities at $h=0$ therefore remain strict for all sufficiently small $h$.
    Hence, after reducing $\bar h$, the candidate satisfies \eqref{eq:positive_h_active_branch_conditions}.

    Because the candidate satisfies the active-branch conditions, the purchase-boundary argument above applies: every searched and matched product~$2$ is bought, both at the candidate and for sufficiently small local deviations.

    The candidate therefore lies strictly in the interior of the active branch, including the product~$2$ purchase boundary.
    This proves existence, local uniqueness, and convergence of the local first-order branch.

    \textbf{Part (1): Price ranking and price comparative statics.}
    The convergence just proved gives $p_i^*(h)\to p_i^*$ for $i=1,2$.
    Hence $p_2^*(h)-p_1^*(h)\to p_2^*-p_1^*>0$, so after reducing $\bar h$ if necessary, $p_1^*(h)<p_2^*(h)$ for all $h\in(0,\bar h)$.
    For $s$, $\eta$, and $\alpha$, the signs are obtained by local sign-preservation under joint $C^1$ continuity and the implicit function theorem.
    The local candidate prices vary continuously differentiably with $h$ and with each parameter $s$, $\eta$, and $\alpha$ in a neighborhood of the benchmark point.
    This follows from the implicit function theorem and the smooth dependence of the fixed-point map $B_{2,h}$ on these variables through $A(h)$, $b_{1,h}$, and $q_{2,h}$.
    Because the benchmark derivatives with respect to $s$, $\eta$, and $\alpha$ have the strict signs stated in Proposition~\ref{prop:limiting_comparative_statics}, and the corresponding positive-$h$ derivatives are continuous in $h$, the same signs persist for all sufficiently small $h$.
    We do not claim global comparative statics across all feasible $h$; for $h$ away from zero, verification would require direct computation of the relevant partial derivatives.

    \textbf{Part (2): Candidate-level detrimental prominence.}
    Recall that $\rho=c-\eta$.
    At $h=0$, Proposition~\ref{lem:limiting_armstrong_style} and Proposition~\ref{prop:limiting_price_comparison} give
    \[
        \eta<p_1^*<p_2^*,
        \qquad
        0<h_2\equiv a+p_1^*-p_2^*<1.
    \]
    The second inequality follows because $p_2^*<a$ gives $h_2>p_1^*>0$, while $p_1^*<p_2^*$ gives $h_2<a<1$.
    By convergence of $p_i^*(h)$ and $A(h)\to a$, after reducing $\bar h$ if necessary,
    \[
        p_1^*(h)>\eta,\qquad p_1^*(h)<p_2^*(h),\qquad
        0<h_2^*(h)<1.
    \]
    Consequently
    \[
        0<q_2^*(h)=1-\alpha+\alpha h_2^*(h)<1.
    \]
    The candidate profit formulas follow from the first-order conditions. Using $D_1=\alpha(p_1^*(h)-\eta)$ and $D_2=S_2(p_2^*(h)-\eta)$, the profits are
    \[
        \pi_1^*(h)=\alpha\left[(p_1^*(h)-\eta)^2-\rho\right],
    \]
    and
    \[
        \pi_2^*(h)=\alpha q_2^*(h)
        \left[(p_2^*(h)-\eta)^2-\rho\right].
    \]
    Hence
    \[
    \begin{aligned}
        \pi_1^*(h)-\pi_2^*(h)
        =
        \alpha\left[
        (p_1^*(h)-\eta)^2
        -q_2^*(h)(p_2^*(h)-\eta)^2
        -(1-q_2^*(h))\rho
        \right].
    \end{aligned}
    \]
    Since $0<h_2^*(h)<1$, we have $1-q_2^*(h)=\alpha(1-h_2^*(h))>0$.
    Define
    \[
        \rho^\dagger(h)\equiv
        \frac{(p_1^*(h)-\eta)^2-q_2^*(h)(p_2^*(h)-\eta)^2}{1-q_2^*(h)},
        \qquad
        \overline\rho(h)\equiv(p_1^*(h)-\eta)^2 .
    \]
    Rearranging the profit gap yields
    \[
        \pi_1^*(h)-\pi_2^*(h)=\alpha(1-q_2^*(h))(\rho^\dagger(h)-\rho),
    \]
    so $\pi_1^*(h)<\pi_2^*(h)$ if and only if $\rho>\rho^\dagger(h)$.
    The price ranking $p_1^*(h)<p_2^*(h)$ and $p_1^*(h)>\eta$ imply $(p_1^*(h)-\eta)^2<(p_2^*(h)-\eta)^2$, hence
    \[
        \rho^\dagger(h)-\overline\rho(h)
        =
        \frac{q_2^*(h)\left[(p_1^*(h)-\eta)^2-(p_2^*(h)-\eta)^2\right]}{1-q_2^*(h)}<0.
    \]
    Thus, for sufficiently small $h$ at the interior pricing candidate, $\rho^\dagger(h)<\overline\rho(h)$.
    If $\rho\leq\overline\rho(h)$, then $\pi_1^*(h)\geq0$; since $(p_2^*(h)-\eta)^2>(p_1^*(h)-\eta)^2$, the same inequality gives $\pi_2^*(h)>0$.
    Conversely, if both candidate profits are nonnegative, then in particular $\pi_1^*(h)\geq0$, so $\rho\leq\overline\rho(h)$.
    Combining this nonnegative-profit bound with $\pi_1^*(h)<\pi_2^*(h)$ if and only if $\rho>\rho^\dagger(h)$ proves that nonnegative-profit detrimental prominence occurs exactly when $\rho\in(\rho^\dagger(h),\overline\rho(h)]$.
\end{proof}

\paragraph{Numerical illustration.}
For the baseline parameters $\alpha=1/2$, $a=0.7$, $s=0.0225$, $\eta=0.05$, $c=0.2274$, and $h=0.01$, the fixed-point equation gives $p_1^*(h)\simeq0.5145$, $p_2^*(h)\simeq0.5265$, and $q_2^*(h)\simeq0.8282$. All active-branch inequalities are satisfied: $c<p_1^*(h)<p_2^*(h)<A(h)+h$, $h_2^*(h)\in(0,1)$, $\mathcal G_{1,h}>0$, and the product-$2$ purchase boundary is slack. The per-unit return loss is $\rho=c-\eta\simeq0.1774$, and the thresholds are $\rho^\dagger(h)\simeq0.1614<\overline\rho(h)\simeq0.2158$, so $\rho\in(\rho^\dagger(h),\overline\rho(h)]$. The candidate profits are $\pi_1^*(h)\simeq0.0192$ and $\pi_2^*(h)\simeq0.0206$: both firms earn positive profits and the non-prominent firm earns more.

\subsection{Proof of Proposition \ref{prop:positive_h_hassle_comparative_statics}}
\begin{proof}
    All derivatives are evaluated along the local candidate, holding the primitive parameters fixed.
    Recall from the local-existence part of the proof of Proposition~\ref{prop:positive_h_extension} that
    \[
        p_2^*(h)=B_{2,h}(p_2^*(h)),
        \qquad
        0<
        \left.\frac{\partial B_{2,h}}{\partial p_2}\right|_{p_2=p_2^*(h)}
        <\frac{1}{2} ,
        \qquad
        0<b_{1,h}'(p_2^*(h))<\frac{1}{2} .
    \]
    Recall also from the active-branch and candidate-level parts of that proof that, after reducing $\bar h$ if necessary,
    \[
        p_1^*(h)-h>0,\qquad
        0<p_2^*(h)-h<A(h)<1,\qquad
        q_2^*(h)>0,\qquad
        \eta<p_1^*(h)<p_2^*(h).
    \]
    Finally, the local-existence proof gives
    \[
        0<
        z\equiv\frac{\alpha(A(h)-p_2^*(h)+h)}{q_{2,h}(p_2^*(h))}
        <1 .
    \]
    Throughout this subsection $A(h)>1/2$.
    Direct differentiation of the reduced map gives
    \[
        \frac{\partial B_{2,h}}{\partial h}
        =
        \frac{1}{2}+\frac{z(2A(h)-z)}{4(1-A(h))}
        +\frac{z^2[1+\alpha(p_2^*(h)-h)]}{8}>0.
    \]
    The strict inequality follows as follows.
    Since $0<z<1$ and $A(h)>1/2$, $2A(h)>1>z$, so $2A(h)-z>0$.
    Since $p_2^*(h)-h>0$ and $\alpha>0$, $1+\alpha[p_2^*(h)-h]>1>0$.
    Also $A(h)<1$, so $1-A(h)>0$.
    Hence the first term is positive, the second term is positive because $z>0$, $2A(h)-z>0$, and $1-A(h)>0$, and the third term is positive because $z^2>0$ and $1+\alpha[p_2^*(h)-h]>0$.

    Differentiating the fixed-point equation $p_2^*(h)=B_{2,h}(p_2^*(h))$ with respect to $h$ gives
    \[
        \left(
        1-\left.\frac{\partial B_{2,h}}{\partial p_2}\right|_{p_2=p_2^*(h)}
        \right)\frac{dp_2^*(h)}{dh}
        =
        \frac{\partial B_{2,h}}{\partial h}.
    \]
    The recalled bound on $\partial B_{2,h}/\partial p_2$ implies that the coefficient on $dp_2^*(h)/dh$ is larger than $1/2$.
    Since $\partial B_{2,h}/\partial h>0$, it follows that $dp_2^*(h)/dh>0$.
    The firm~$1$ price on the local branch satisfies $p_1^*(h)=b_{1,h}(p_2^*(h))$.
    Differentiating this identity gives
    \[
        \frac{dp_1^*(h)}{dh}
        =
        b_{1,h}'(p_2^*(h))\frac{dp_2^*(h)}{dh}
        +\frac{1+\alpha(p_2^*(h)-h)}{2}>0.
    \]
    The first term is positive because $b_{1,h}'(p_2^*(h))>0$ and $dp_2^*(h)/dh>0$.
    The second term is positive because $p_2^*(h)-h>0$ and $\alpha>0$.

    Recall that $h_2^*(h)=A(h)+p_1^*(h)-p_2^*(h)$ and $q_2^*(h)=1-\alpha+\alpha h_2^*(h)$.
    Since $A'(h)=-1/[1-A(h)]$, we have
    \[
        \frac{dh_2^*(h)}{dh}
        =
        \left[-\frac{1}{1-A(h)}+\frac{1+\alpha(p_2^*(h)-h)}{2}\right]
        +\left[b_{1,h}'(p_2^*(h))-1\right]\frac{dp_2^*(h)}{dh}.
    \]
    The first bracket is negative.
    Indeed, $0<p_2^*(h)-h<A(h)$ and $0<\alpha<1$ imply
    $\alpha[p_2^*(h)-h]<A(h)$, so
    \[
        \frac{1+\alpha(p_2^*(h)-h)}{2}
        <
        \frac{1+A(h)}{2}.
    \]
    Also, because $0<A(h)<1$,
    \[
        \frac{1+A(h)}{2}
        <
        \frac{1}{1-A(h)};
    \]
    this is equivalent to $1-A(h)^2<2$, which holds since $A(h)^2>0$.
    Therefore the first bracket is strictly negative.
    The second term is also strictly negative because $b_{1,h}'(p_2^*(h))<1/2<1$ and $dp_2^*(h)/dh>0$.
    Hence $dh_2^*(h)/dh<0$.
    Since $q_2^*(h)=1-\alpha+\alpha h_2^*(h)$ and $\alpha>0$, $dq_2^*(h)/dh<0$.

    For firm~$1$, initial sales are fixed at $S_1=\alpha$.
    Recall from the first-order condition that $D_1=\alpha[p_1^*(h)-\eta]$ at the local candidate.
    Hence, along the candidate branch,
    \[
        \pi_1^*(h)
        =
        \alpha[p_1^*(h)-\eta]^2-\alpha(c-\eta).
    \]
    Differentiating this identity with respect to $h$ gives
    \[
        \frac{d\pi_1^*(h)}{dh}
        =
        2\alpha(p_1^*(h)-\eta)\frac{dp_1^*(h)}{dh}>0.
    \]
    The sign follows from $\alpha>0$, $p_1^*(h)>\eta$, and $dp_1^*(h)/dh>0$.

    Total returns equal initial sales minus retained demand.
    Using the first-order conditions to substitute for retained demand gives
    $$
        R^*(h)=\alpha[1-(p_1^*(h)-\eta)]
        +\alpha q_2^*(h)[1-(p_2^*(h)-\eta)].
    $$
    Differentiating this expression yields
    \[
        \frac{dR^*(h)}{dh}
        =
        -\alpha\frac{dp_1^*(h)}{dh}
        +\alpha\frac{dq_2^*(h)}{dh}[1-(p_2^*(h)-\eta)]
        -\alpha q_2^*(h)\frac{dp_2^*(h)}{dh}.
    \]
    The first term is negative because $\alpha>0$ and $dp_1^*(h)/dh>0$.
    The second term is negative because $\alpha>0$, $dq_2^*(h)/dh<0$, and $1-[p_2^*(h)-\eta]>0$.
    The last inequality follows for sufficiently small $h$ by convergence from the benchmark inequalities $p_2^*<a<1$ and $\eta\geq0$.
    The third term is negative because $\alpha>0$, $q_2^*(h)>0$, and $dp_2^*(h)/dh>0$.
    Therefore $dR^*(h)/dh<0$.

    Finally, let $CS(h,p_1,p_2)$ denote the ex ante consumer value on the active branch.
    The envelope theorem gives
    \[
        \frac{dCS^*(h)}{dh}
        =
        \frac{\partial CS}{\partial h}
        +\frac{\partial CS}{\partial p_1}\frac{dp_1^*(h)}{dh}
        +\frac{\partial CS}{\partial p_2}\frac{dp_2^*(h)}{dh}.
    \]
    At fixed prices, the direct effect of $h$ is strictly negative because $h$ is paid only when a product is returned, and the active branch has a positive-measure return region.
    Hence $\partial CS/\partial h<0$.
    The price envelope gives $\partial CS/\partial p_i=-D_i<0$.
    The strict inequality follows from the first-order conditions:
    \[
        D_1=\alpha[p_1^*(h)-\eta]>0,
        \qquad
        D_2=\alpha q_2^*(h)[p_2^*(h)-\eta]>0,
    \]
    where the signs use $\alpha>0$, $q_2^*(h)>0$, and $\eta<p_1^*(h)<p_2^*(h)$.
    Since $dp_i^*(h)/dh>0$, the chain-rule expression above implies $dCS^*(h)/dh<0$.
    The preceding arguments establish all signed comparative statics in the proposition.
\end{proof}

\subsection{Proof of Corollary \ref{prop:partial_refund_continuity}}
\begin{proof}
    We first follow the backward-induction logic of Proposition~\ref{prop: general consumer rule} to characterize the consumer's optimal search and purchase decisions for $\tau<1$.
    First fix the decision after product~2 has been searched.
    If the consumer has no product at hand and product~2 is unmatched, buying product~2 cannot give a strict gain over not buying it.
    If product~2 is matched, the expected payoff from buying it at actual price $p_2$ is
    $$
        V^{\text{buy}}_\tau(p_2)
        \equiv
        E_{\varepsilon_2}\bigl[\max\{\varepsilon_2-p_2,\;-(1-\tau)p_2\}\bigr]
        =
        \phi(\tau p_2)-(1-\tau)p_2 .
    $$
    Here $\phi(z)=(1-z)^2/2$.
    Thus a matched product~2 is bought from the no-product state iff $V^{\text{buy}}_\tau(p_2)>0$, and conditional on purchase it is kept iff $\varepsilon_2\geq\tau p_2$.
    When $\tau=1$, $V^{\text{buy}}_1(p_2)=\phi(p_2)$.

    Now suppose the consumer is holding product~1, so $\varepsilon_1\geq\tau p_1$.
    Let $x=\varepsilon_1-p_1$.
    If product~2 is unmatched, buying it cannot improve on keeping product~1.
    If product~2 is matched and bought at actual price $p_2$, the relevant final payoffs are
    $$
        x-(1-\tau)p_2,
        \qquad
        \varepsilon_2-p_2-(1-\tau)p_1 .
    $$
    Returning both is dominated by keeping product~1 and returning product~2 because $x\geq-(1-\tau)p_1$.
    Therefore the expected incremental value of buying a matched product~2, relative to stopping with product~1, is
    $$
        \phi\bigl(\varepsilon_1+\tau(p_2-p_1)\bigr)-(1-\tau)p_2 .
    $$
    Product~2 is bought iff this expression is positive.
    Conditional on buying product~2, the consumer keeps product~2 rather than product~1 iff
    $$
        \varepsilon_2>\varepsilon_1+\tau(p_2-p_1).
    $$

    We next consider search for product~2.
    From the no-product state, search is optimal iff
    $$
        \mathcal V_0^\tau(p_2^*)\equiv -s+\alpha V^{\text{buy}}_\tau(p_2^*)>0 .
    $$
    Because $s>0$, this condition implies $V^{\text{buy}}_\tau(p_2^*)>0$.
    Hence a matched product~2 is bought at the believed price, and by continuity it is also bought for actual prices $p_2$ close to $p_2^*$.
    If the consumer bought product~1 and observed $\varepsilon_1<\tau p_1$, she returns product~1 and faces the same search decision, up to the sunk loss $(1-\tau)p_1$.

    If the consumer still holds product~1, search is optimal iff
    $$
        \alpha\max\left\{
        \phi\bigl(\varepsilon_1+\tau(p_2^*-p_1)\bigr)-(1-\tau)p_2^*,0
        \right\}>s.
    $$
    Since $s>0$, this is equivalent to
    $$
        \alpha\left[
        \phi\bigl(\varepsilon_1+\tau(p_2^*-p_1)\bigr)-(1-\tau)p_2^*
        \right]>s.
    $$
    Define
    $$
        \tilde a_\tau(p_2^*)\equiv\phi^{-1}\bigl(\tfrac{s}{\alpha}+(1-\tau)p_2^*\bigr),
        \qquad
        H_1^\tau(p_1,p_2^*)\equiv \tilde a_\tau(p_2^*)-\tau(p_2^*-p_1).
    $$
    Then the holding-product search condition is $\varepsilon_1<H_1^\tau(p_1,p_2^*)$.
    At the believed price, the corresponding purchase value is at least $s/\alpha>0$ throughout this search region.
    Indeed, $\varepsilon_1\leq H_1^\tau(p_1,p_2^*)$ implies
    $\varepsilon_1+\tau(p_2^*-p_1)\leq\tilde a_\tau(p_2^*)$, so
    $$
        \phi\bigl(\varepsilon_1+\tau(p_2^*-p_1)\bigr)-(1-\tau)p_2^*
        \geq
        \phi\bigl(\tilde a_\tau(p_2^*)\bigr)-(1-\tau)p_2^*
        =
        \frac{s}{\alpha}.
    $$
    Hence, after restricting attention to a sufficiently small local price neighborhood, the matched product~2 is also bought at actual prices $p_2$ close to $p_2^*$.

    This gives the active-branch cutoffs.
    In the competing branch, define
    $$
        G_2^\tau(p_1,p_2;p_2^*)
        \equiv H_1^\tau(p_1,p_2^*)+\tau(p_2-p_1)
        =\tilde a_\tau(p_2^*)+\tau(p_2-p_2^*).
    $$
    The cutoff $G_2^\tau$ is the value of $\varepsilon_2$ at which product~2 is just tied with product~1 for the highest-searching consumer $\varepsilon_1=H_1^\tau$.

    Finally consider the initial product-$1$ purchase decision.
    If product~1 is unmatched, it is not bought.
    If product~1 is matched, buying product~1 gives the ex ante continuation value
    $$
    \begin{aligned}
        \mathcal B_1^\tau(p_1;p_2^*)
        &\equiv
        \int_0^{\tau p_1}\left[-(1-\tau)p_1+\mathcal V_0^\tau(p_2^*)\right]\,d\varepsilon_1\\
        &\quad+
        \int_{\tau p_1}^{H_1^\tau(p_1,p_2^*)}
        \left[
            \varepsilon_1-p_1-s
            +\alpha\left\{\phi\bigl(\varepsilon_1+\tau(p_2^*-p_1)\bigr)-(1-\tau)p_2^*\right\}
        \right]\,d\varepsilon_1\\
        &\quad+
        \int_{H_1^\tau(p_1,p_2^*)}^{1}(\varepsilon_1-p_1)\,d\varepsilon_1 .
    \end{aligned}
    $$
    On the active branch, skipping product~1 gives the no-product search value $\mathcal V_0^\tau(p_2^*)$.
    Define the product-$1$ buy slack by
    $$
        \mathcal G_{1,\tau}(p_1,p_2^*)\equiv
        \mathcal B_1^\tau(p_1;p_2^*)-\mathcal V_0^\tau(p_2^*).
    $$
    The partial-refund active branch is the open set of actual prices $(p_1,p_2)$ and belief $p_2^*$ satisfying
    \begin{equation}
        \begin{aligned}
            &0<\tau p_1<H_1^\tau(p_1,p_2^*)<1,\\[2mm]
            &0<\tau p_2<G_2^\tau(p_1,p_2;p_2^*)<1,\\[2mm]
            &\mathcal G_{1,\tau}(p_1,p_2^*)>0,\\[2mm]
            &\mathcal V_0^\tau(p_2^*)>0.
        \end{aligned}
        \label{eq:partial_refund_active_branch_conditions}
    \end{equation}
    As in the positive-$h$ extension, this is a local active branch: actual deviations are restricted to a neighborhood in which the purchase slack just described holds.

    Next compute the demand objects needed for the firms' first-order conditions.
    Under the active-branch assumptions, initial sales are
    $$
        S_1=\alpha,\qquad
        S_2^\tau(p_1,p_2^*)=\alpha(1-\alpha+\alpha H_1^\tau(p_1,p_2^*)).
    $$
    For kept-demand purposes, the competing branch contains three relevant regions.
    If $\varepsilon_1>H_1^\tau$, firm~1 keeps the consumer.
    If $\varepsilon_1<\tau p_1$ and $\varepsilon_2>\tau p_2$, the consumer returns product~1 and keeps product~2.
    If $\tau p_1\le\varepsilon_1<H_1^\tau$, the consumer holds product~1, buys product~2, and keeps the product with the higher net surplus.
    Firm~1's retained demand in the last subregion is
    $$
        k_1^\tau(p_1,p_2;p_2^*) = \int_{\tau p_1}^{H_1^\tau(p_1,p_2^*)}\bigl(\varepsilon_1+\tau(p_2-p_1)\bigr)\,d\varepsilon_1
        = \frac{(H_1^\tau)^2-(\tau p_1)^2}{2} + \tau(p_2-p_1)(H_1^\tau-\tau p_1),
    $$
    and firm~2 receives the complement $(H_1^\tau-\tau p_1)-k_1^\tau$.
    Summing over match states gives
    $$
        \begin{aligned}
            D_1^\tau(p_1,p_2;p_2^*) &= \alpha(1-\alpha)(1-\tau p_1) + \alpha^2\bigl[1-H_1^\tau + k_1^\tau\bigr],\\[2mm]
            D_2^\tau(p_1,p_2;p_2^*) &= \alpha(1-\alpha)(1-\tau p_2) + \alpha^2\bigl[\tau p_1(1-\tau p_2) + (H_1^\tau-\tau p_1) - k_1^\tau\bigr].
        \end{aligned}
    $$
    The residual area $\tau^2 p_1p_2$ corresponds to $\varepsilon_1<\tau p_1$ and $\varepsilon_2<\tau p_2$, where both products are returned.
    It therefore enters returns rather than retained demand.
    At $\tau=1$, these formulas reduce exactly to the baseline demand expressions.

    We next derive the corrected first-order system.
    If a sold unit is returned, the firm refunds only $\tau p_i$, retains $(1-\tau)p_i$, and receives salvage value $\eta$.
    Let $R_i^\tau\equiv S_i^\tau-D_i^\tau$ denote returns.
    Firm $i$'s profit is
    $$
    \begin{aligned}
        \Pi_i^\tau
        &=(p_i-c)D_i^\tau+\bigl[(1-\tau)p_i+\eta-c\bigr]R_i^\tau\\
        &=(\tau p_i-\eta)D_i^\tau+\bigl[(1-\tau)p_i-\rho\bigr]S_i^\tau,
        \qquad \rho\equiv c-\eta .
    \end{aligned}
    $$
    This expression reduces to the baseline full-refund profit when $\tau=1$.

    Local own-price derivatives are taken holding consumers' belief $p_2^*$ fixed.
    The relevant own-price demand derivatives are
    $$
        \frac{\partial D_i^\tau}{\partial p_i}=-\tau S_i^\tau,
        \qquad i=1,2,
        \qquad S_1^\tau=\alpha .
    $$
    For firm~1, this follows from $\partial H_1^\tau/\partial p_1=\tau$ and
    $$
    \begin{aligned}
        \frac{\partial k_1^\tau}{\partial p_1}
        &=
        H_1^\tau\tau-\tau^2p_1
        -\tau(H_1^\tau-\tau p_1)
        +\tau(p_2-p_1)\left(\tau-\tau\right)\\
        &=0,
    \end{aligned}
    $$
    so
    $$
        \frac{\partial D_1^\tau}{\partial p_1}
        =
        -\alpha(1-\alpha)\tau-\alpha^2\tau
        =
        -\alpha\tau
        =
        -\tau S_1^\tau .
    $$
    For firm~2, $H_1^\tau(p_1,p_2^*)$ does not depend on the actual deviating price $p_2$, and
    $$
        \frac{\partial k_1^\tau}{\partial p_2}
        =
        \tau(H_1^\tau-\tau p_1).
    $$
    Hence
    $$
    \begin{aligned}
        \frac{\partial D_2^\tau}{\partial p_2}
        &=
        -\alpha(1-\alpha)\tau
        +\alpha^2\left[-\tau^2p_1-\tau(H_1^\tau-\tau p_1)\right]\\
        &=
        -\alpha\tau\left[(1-\alpha)+\alpha H_1^\tau\right]
        =
        -\tau S_2^\tau .
    \end{aligned}
    $$
    Also, $S_i^\tau$ is locally independent of firm $i$'s own deviating price: $S_1^\tau=\alpha$, and $S_2^\tau$ is determined by the search probability under the fixed belief $p_2^*$.
    Therefore
    $$
    \begin{aligned}
        \frac{\partial \Pi_i^\tau}{\partial p_i}
        &=
        \tau D_i^\tau
        +(\tau p_i-\eta)\frac{\partial D_i^\tau}{\partial p_i}
        +(1-\tau)S_i^\tau\\
        &=
        \tau D_i^\tau
        -\tau S_i^\tau(\tau p_i-\eta)
        +(1-\tau)S_i^\tau .
    \end{aligned}
    $$
    Setting this derivative equal to zero gives the corrected first-order conditions
    $$
        D_i^\tau(p_1,p_2;p_2^*)
        =
        S_i^\tau\left(\tau p_i-\eta-\delta_\tau\right),
        \qquad
        i=1,2,
        \qquad
        \delta_\tau\equiv\frac{1-\tau}{\tau},
    $$
    where $S_1^\tau=\alpha$ and $S_2^\tau=S_2^\tau(p_1,p_2^*)$.
    On the active branch, $\partial^2D_i^\tau/\partial p_i^2=0$, so
    $$
        \frac{\partial^2 \Pi_i^\tau}{\partial p_i^2}
        =
        -2\tau^2 S_i^\tau<0,
    $$
    where $S_1^\tau=\alpha$.
    Hence the first-order conditions are sufficient for local own-price optimality within the active branch.
    Substituting the first-order conditions into profits gives
    $$
        \Pi_i^\tau
        =
        S_i^\tau\left[
            (\tau p_i-\eta)^2+\delta_\tau\eta-\rho
        \right],
        \qquad
        S_1^\tau=\alpha.
    $$
    At $\tau=1$, $\delta_\tau=0$ and these expressions reduce to the baseline formulas $\pi_i=S_i[(p_i-\eta)^2-\rho]$.

    It remains to apply the implicit function theorem.
    Impose candidate consistency by setting the belief equal to the actual price, $p_2^*=p_2$.
    Write
    $$
        \widehat D_i^\tau(p_1,p_2)\equiv D_i^\tau(p_1,p_2;p_2),
        \qquad
        \widehat S_2^\tau(p_1,p_2)\equiv S_2^\tau(p_1,p_2).
    $$
    The consistent candidate equations are $F_1=F_2=0$, where
    $$
        F(p_1,p_2;\tau)=(F_1,F_2),
    $$
    $$
        F_1
        \equiv
        \widehat D_1^\tau(p_1,p_2)-\alpha\left(\tau p_1-\eta-\delta_\tau\right),
    $$
    and
    $$
        F_2
        \equiv
        \widehat D_2^\tau(p_1,p_2)-\widehat S_2^\tau(p_1,p_2)\left(\tau p_2-\eta-\delta_\tau\right).
    $$
    At $\tau=1$, $F(p_1^*,p_2^*;1)=(0,0)$.
    The map $F$ is $C^1$ near $(p_1^*,p_2^*,1)$: the active-branch demand expressions are smooth there, $\delta_\tau=(1-\tau)/\tau$ is smooth near $\tau=1$, and the implicit cutoff remains in the interior of the domain of $\phi^{-1}$ for $\tau$ close to $1$.

    We first solve the firm~1 equation.
    At $\tau=1$, the equations coincide with the benchmark first-order system, and
    $$
        \frac{\partial F_1}{\partial p_1}(p_1^*,p_2^*;1)=-2\alpha<0.
    $$
    Hence the implicit function theorem applied to $F_1=0$ gives a unique local $C^1$ function $g(p_2,\tau)$, defined near $(p_2^*,1)$, such that
    $$
        F_1(g(p_2,\tau),p_2;\tau)=0,
        \qquad
        g(p_2^*,1)=p_1^* .
    $$
    Thus, near the benchmark candidate, solving the two-equation system is equivalent to solving the single reduced equation
    $$
        \widetilde R(p_2,\tau)\equiv F_2(g(p_2,\tau),p_2;\tau)=0 .
    $$

    It remains to verify that this reduced equation is regular at $(p_2^*,1)$.
    At $\tau=1$, the reduced equation coincides with the benchmark reduced firm~2 equation.
    The proof of Proposition~\ref{lem:limiting_armstrong_style} writes that benchmark equation as $p_2=\mathcal{T}(p_2)$ and establishes $0<\mathcal{T}'(p_2^*)<1/2$.
    Along the curve $p_1=g(p_2,1)$, firm~2's equation can be written as
    $$
        \widetilde R(p_2,1)
        =
        2\widehat S_2^1(g(p_2,1),p_2)\,[\mathcal{T}(p_2)-p_2].
    $$
    Since $\mathcal{T}(p_2^*)=p_2^*$, differentiating at $p_2^*$ gives
    $$
        \frac{\partial \widetilde R}{\partial p_2}(p_2^*,1)
        =
        2\widehat S_2^1(g(p_2^*,1),p_2^*)\,[\mathcal{T}'(p_2^*)-1].
    $$
    The first factor is strictly positive because firm~2 has positive initial sales on the active branch.
    The second factor is strictly negative because $\mathcal{T}'(p_2^*)<1/2<1$.
    Hence
    $$
        \frac{\partial \widetilde R}{\partial p_2}(p_2^*,1)\neq0.
    $$
    A second application of the implicit function theorem gives a unique $C^1$ function $p_2^\tau$ solving $\widetilde R(p_2^\tau,\tau)=0$ for $\tau$ close to $1$.
    Setting
    $$
        p_1^\tau=g(p_2^\tau,\tau)
    $$
    gives the unique local $C^1$ solution to $F=0$ in a neighborhood of $(p_1^*,p_2^*)$.
    Strict active-branch inequalities at $\tau=1$ remain strict for $\tau$ sufficiently close to $1$, so this local solution remains inside the active branch.
    This proves the existence, local uniqueness, and convergence claims.
    The price ranking follows from continuity and $p_2^*-p_1^*>0$.
    Finally, the corrected profit formulas are continuous in $(\tau,p_1,p_2)$ near $\tau=1$, so $\Pi_i^\tau\to\pi_i^*$.
    If $c-\eta\in(\rho^\dagger(\eta),(p_1^*-\eta)^2)$, Proposition~\ref{prop:limiting_detrimental_prominence} gives $\pi_1^*<\pi_2^*$ and $\pi_i^*>0$.
    The same strict inequalities hold for all $\tau$ sufficiently close to $1$.
\end{proof}

\subsection{Proof of Proposition \ref{prop:oligopoly_random_search}} \label{appendix:oligopoly_proof}
 \begin{proof}
    For finite $n$, all non-prominent firms are expected to charge the same price $p_N$.
    Because returns are costless for consumers, after the prominent product has been inspected the consumer compares any future non-prominent draw with her current best available option.

    As in \cite{armstrong2009prominence}, the finite-horizon search rule is stationary.
    If this current option is expressed as an equivalent non-prominent taste value $y$, one more non-prominent search generates incremental expected surplus
    \[
        \alpha\int_y^1(\varepsilon-y)\,d\varepsilon
        =
        \alpha\phi(y).
    \]
    Since $a$ is defined by $\alpha\phi(a)=s$, search is strictly optimal iff $y<a$.
    This cutoff does not depend on the number of non-prominent firms remaining.
    The argument is by backward induction.
    With only one non-prominent firm left, the comparison is exactly between the current option and the expected one-period gain above.
    Now suppose the cutoff rule is valid when fewer firms remain.
    If $y\geq a$, sampling one more non-prominent firm cannot lower the current best option.
    After the draw, the best equivalent taste value is either still $y$ or is replaced by a higher draw, and therefore remains at least $a$.
    By the induction hypothesis, the consumer would stop after that draw.
    Hence the net value of searching once is
    \[
        -s+\alpha\int_y^1(\varepsilon-y)\,d\varepsilon
        =
        -s+\alpha\phi(y)
        \leq
        -s+\alpha\phi(a)
        =
        0,
    \]
    where the inequality uses that $\phi$ is decreasing and $y\geq a$.
    Under the strict-action convention, the consumer does not search.
    If $y<a$, then even if the consumer were forced to stop after the next draw, the net value of searching once would be
    \[
        -s+\alpha\phi(y)
        >
        -s+\alpha\phi(a)
        =
        0.
    \]
    The option to continue searching afterward can only increase this value, so search is strictly optimal.

    Hence a consumer leaves the prominent firm iff
    \[
        v_0(\varepsilon_1-p_1)<a-p_N,
    \]
    and a searched non-prominent product is satisfactory iff it is matched and $\varepsilon\geq a$.
    On the interior active branch, $p_N<a$ and
    \[
        H_N\equiv a+p_1-p_N\in(0,1).
    \]
    The probability that the consumer reaches the first non-prominent firm is
    \[
        1-\alpha+\alpha H_N.
    \]
    After that, each previously inspected non-prominent firm fails to stop search with probability
    \[
        1-\alpha(1-a),
    \]
    because it is either unmatched or matched with $\varepsilon<a$.
    Therefore the probability that the consumer reaches the $\ell$th non-prominent position is
    \[
        \Pr(\mathcal R_\ell)
        =
        (1-\alpha+\alpha H_N)[1-\alpha(1-a)]^{\ell-1},
        \qquad \ell=1,\ldots,n-1 .
    \]
    Since the non-prominent order is uniform, the representative non-prominent firm's inspection probability is
    \[
        q_N(n)
        =
        \frac{1-\alpha+\alpha H_N}{n-1}
        \sum_{\ell=0}^{n-2}[1-\alpha(1-a)]^\ell .
    \]
    The active-branch inequalities $0<H_N<1$, $0<\alpha<1$, and $a<1$ imply $0<q_N(n)<1$.

    For an active price $p_1<1$, every matched consumer initially buys the prominent product, so $S_1^n=\alpha$.
    A representative non-prominent firm is initially purchased exactly when it is inspected and matched; hence $S_N^n=\alpha q_N(n)$.
    Let $D_1^n$ and $D_N^n$ denote the retained demands induced by the stationary cutoff.
    These retained demands include both products kept immediately at the cutoff and products recalled after all searched products fail to reach the cutoff.

    We next simplify how retained demand changes with a firm's own local deviation price.
    For the prominent firm charging $p$, the immediate retained demand is
    \[
        \alpha(1-a-p+p_N),
    \]
    because a matched product is kept immediately iff $\varepsilon-p\geq a-p_N$.
    Its recall component is
    \[
        R_1^n(p_N)
        =
        \alpha\int_0^{a-p_N}
        [1-\alpha+\alpha(u+p_N)]^{n-1}\,du .
    \]
    To see this, put $u=\varepsilon-p$ for the prominent product's realized net surplus.
    A recall buyer must have $0\leq u<a-p_N$.
    For each non-prominent firm, the product does not beat the prominent product iff it is unmatched or, if matched, its taste value is below $u+p_N$; this has probability $1-\alpha+\alpha(u+p_N)$.
    Hence
    \[
        \frac{\partial D_1^n(p,p_N)}{\partial p}=-\alpha .
    \]

    The same argument applies to a representative non-prominent firm.
    Its inspection probability is fixed at $q_N(n)$ under a local price deviation, because search decisions use the symmetric price belief $p_N$.
    Its total immediate retained demand is $\alpha q_N(n)(1-a-p+p_N)$ and therefore changes by $-\alpha q_N(n)$ per unit increase in $p$.
    Its recall component is
    \[
        R_N^n(p_1,p_N)
        =
        \alpha\int_0^{a-p_N}
        [1-\alpha+\alpha(u+p_N)]^{n-2}
        [1-\alpha+\alpha(u+p_1)]\,du .
    \]
    Here $u=\varepsilon-p$ is the representative non-prominent firm's realized net surplus.
    The product is recalled and retained at this firm only when $0\leq u<a-p_N$, all other non-prominent products fail to beat it, and the prominent product also fails to beat it.
    These events have probabilities $[1-\alpha+\alpha(u+p_N)]^{n-2}$ and $1-\alpha+\alpha(u+p_1)$, respectively.
    This event also implies that the representative firm is reached, so no additional inspection-probability factor multiplies the recall component.
    Again, the integral is independent of the firm's own price $p$.
    Therefore
    \[
        \left.
        \frac{\partial D_N^n(p;p_1,p_N)}{\partial p}
        \right|_{p=p_N}
        =
        -\alpha q_N(n).
    \]

    The local deviation profit functions have the same accounting structure as in the duopoly.
    For the prominent firm, initial sales are locally fixed at $\alpha$, so its local payoff from price $p$ is
    \[
        \Pi_1^n(p;p_N)
        =
        (p-\eta)D_1^n(p,p_N)-(c-\eta)\alpha .
    \]
    For a non-prominent firm, a local deviation in its actual price does not change the event that it is reached, because search decisions are made using the symmetric price belief $p_N$.
    Thus its initial sales are locally fixed at $\alpha q_N(n)$.
    Writing $D_N^n(p;p_1,p_N)$ for the deviating firm's retained demand when its actual price is $p$ and consumers' search beliefs remain at $p_N$, its local payoff is
    \[
        \Pi_N^n(p;p_1,p_N)
        =
        (p-\eta)D_N^n(p;p_1,p_N)-(c-\eta)\alpha q_N(n).
    \]

    Differentiating these payoffs gives the local first-order conditions.
    For the prominent firm, an interior local optimum satisfies
    \[
        D_1^n(p_1,p_N)
        -\alpha(p_1-\eta)
        =
        0 .
    \]
    Equivalently,
    \[
        D_1^n(p_1,p_N)=\alpha(p_1-\eta).
    \]
    For a non-prominent firm, an interior symmetric optimum satisfies
    \[
        D_N^n(p_N;p_1,p_N)
        -\alpha q_N(n)(p_N-\eta)
        =
        0 .
    \]
    Equivalently,
    \[
        D_N^n(p_N;p_1,p_N)=\alpha q_N(n)(p_N-\eta).
    \]

    These first-order conditions imply the price ranking $p_1<p_N$.
    Suppose, to the contrary, that $p_1\geq p_N$, and put $\mu=p_1-p_N\geq0$.
    Then $H_N=a+\mu$, and the active branch gives $H_N<1$ and $p_N<a$.
    Define
    \[
        I_{\nu}
        \equiv
        \int_0^{a-p_N}
        [1-\alpha+\alpha(u+p_N)]^{\nu}\,du>0 .
    \]
    The demand formulas above give
    \[
        \frac{D_1^n}{\alpha}
        =
        1-a-\mu+I_{n-1},
    \]
    and, because $p_1=p_N+\mu$,
    \[
        \frac{D_N^n}{\alpha q_N(n)}
        =
        1-a+\frac{I_{n-1}+\alpha\mu I_{n-2}}{q_N(n)} .
    \]
    Hence the first-order conditions imply
    \[
    \begin{aligned}
        p_N-p_1
        &=
        \frac{D_N^n}{\alpha q_N(n)}
        -
        \frac{D_1^n}{\alpha}\\
        &=
        \mu+\left(\frac{1}{q_N(n)}-1\right)I_{n-1}
        +\frac{\alpha\mu I_{n-2}}{q_N(n)}
        >0,
    \end{aligned}
    \]
    where the strict inequality uses $0<q_N(n)<1$ and $I_{n-1}>0$.
    This contradicts $p_1\geq p_N$.
    Therefore $p_1<p_N$.
    Since $D_1^n>0$, the firm~1 first-order condition also gives $p_1>\eta$, and thus
    \[
        0<p_1-\eta<p_N-\eta .
    \]

    At the symmetric candidate, $\pi_1^n\equiv\Pi_1^n(p_1;p_N)$ and $\pi_N^n\equiv\Pi_N^n(p_N;p_1,p_N)$.
    Substituting the first-order conditions into the payoff functions gives
    \[
        \pi_1^n=\alpha\left[(p_1-\eta)^2-(c-\eta)\right],
        \qquad
        \pi_N^n=\alpha q_N(n)\left[(p_N-\eta)^2-(c-\eta)\right].
    \]
    Their difference is
    \[
        \pi_1^n-\pi_N^n
        =
        \alpha\left\{
        (p_1-\eta)^2-q_N(n)(p_N-\eta)^2
        -[1-q_N(n)](c-\eta)
        \right\}.
    \]
    Define
    \[
        \rho_n^\dagger
        \equiv
        \frac{(p_1-\eta)^2-q_N(n)(p_N-\eta)^2}{1-q_N(n)},
        \qquad
        \overline\rho_n
        \equiv
        \min\left\{
        (p_1-\eta)^2,\,
        (p_N-\eta)^2
        \right\}.
    \]
    Since $0<q_N(n)<1$, rearranging this expression gives
    \[
        \pi_1^n<\pi_N^n
        \quad\Longleftrightarrow\quad
        c-\eta>\rho_n^\dagger .
    \]
    The two nonnegative-profit inequalities are
    \[
        \pi_1^n\geq0
        \quad\Longleftrightarrow\quad
        c-\eta\leq(p_1-\eta)^2,
    \]
    and
    \[
        \pi_N^n\geq0
        \quad\Longleftrightarrow\quad
        c-\eta\leq(p_N-\eta)^2 .
    \]
    Combining the profit-gap threshold with these two nonnegative-profit bounds gives detrimental prominence with nonnegative profits exactly when
    \[
        c-\eta\in(\rho_n^\dagger,\overline\rho_n].
    \]
    The price ranking proved above gives $\overline\rho_n=(p_1-\eta)^2$ and
    \[
        \rho_n^\dagger<(p_1-\eta)^2
        \quad\Longleftrightarrow\quad
        (p_1-\eta)^2<(p_N-\eta)^2,
    \]
    so this interval is nonempty.
    This proves the finite-$n$ claim.

    Now consider the infinite-seller limit.
    Once the consumer leaves the prominent firm, the remaining search problem is stationary.
    With infinitely many possible non-prominent draws, there is no terminal seller; the cutoff is therefore obtained from the stationary reservation condition rather than from backward induction.
    As in \cite{armstrong2009prominence}, if all non-prominent firms charge $p_\infty$, the reservation taste value is characterized by
    \[
        s
        =
        \alpha\int_a^1(\varepsilon-a)\,d\varepsilon
        =
        \alpha\phi(a).
    \]
    Hence the stationary reservation surplus is $a-p_\infty$, and a searched non-prominent product is accepted iff
    \[
        \varepsilon-p_\infty\geq a-p_\infty,
        \qquad\text{equivalently}\qquad
        \varepsilon\geq a .
    \]

    Consider first a non-prominent seller that is reached by a consumer who has not yet found an acceptable product.
    A local deviation in its actual price $p$ does not affect the probability of being inspected, because search decisions use the equilibrium belief $p_\infty$.
    Conditional on inspection, the seller makes an initial sale with probability $\alpha$ on the active branch.
    The product is finally kept iff
    \[
        \varepsilon-p\geq a-p_\infty,
        \qquad\text{equivalently}\qquad
        \varepsilon\geq p+a-p_\infty .
    \]
    Thus, after suppressing the inspection probability, the seller's local payoff is
    \[
        \widetilde\Pi_N(p;p_\infty)
        =
        \alpha\left\{
        (p-\eta)\bigl[1-p-a+p_\infty\bigr]-(c-\eta)
        \right\}.
    \]
    This payoff is strictly concave in $p$, since
    \[
        \frac{\partial\widetilde\Pi_N(p;p_\infty)}{\partial p}
        =
        \alpha\bigl[1-a+p_\infty+\eta-2p\bigr],
        \qquad
        \frac{\partial^2\widetilde\Pi_N(p;p_\infty)}{\partial p^2}
        =
        -2\alpha<0 .
    \]
    In a symmetric stationary non-prominent equilibrium, $p=p_\infty$, so the first-order condition gives
    \[
        p_\infty^*=1-a+\eta .
    \]
    The maintained condition $\eta<2a-1$ implies $p_\infty^*<a$, so the stationary active branch is nonempty.

    The prominent firm faces the same reservation outside option, but it is inspected by every consumer.
    If firm~1 charges $p$, its retained demand is
    \[
        \alpha\bigl[1-p-a+p_\infty^*\bigr],
    \]
    and its payoff is
    \[
        \Pi_1(p)
        =
        \alpha\left\{
        (p-\eta)\bigl[1-p-a+p_\infty^*\bigr]-(c-\eta)
        \right\}.
    \]
    The payoff is strictly concave, and the first-order condition is
    \[
        1-a+p_\infty^*+\eta-2p=0 .
    \]
    Therefore
    \[
        p_1^*
        =
        \frac{1-a+p_\infty^*+\eta}{2}
        =
        1-a+\eta
        =
        p_\infty^* .
    \]
    Firm~1's stationary profit is consequently
    \[
        \pi_1^\infty
        =
        \alpha\left[(1-a)^2-(c-\eta)\right].
    \]

    It remains only to compare this profit with that of an individual non-prominent seller.
    Let $q_N$ be the probability that a given non-prominent seller is inspected in the stationary infinite-seller environment.
    Such an inspection can occur only if the consumer does not keep firm~1's product.
    At the stationary prices, firm~1's retained demand is $D_1^\infty=\alpha(1-a)$.
    Hence
    \[
        q_N\leq1-D_1^\infty=1-\alpha(1-a)<1.
    \]
    The individual non-prominent firm's stationary profit is
    \[
        \pi_N^\infty
        =
        \alpha q_N\left[(1-a)^2-(c-\eta)\right].
    \]
    Therefore the profit difference between the prominent firm and an individual non-prominent firm is
    \[
    \begin{aligned}
    \pi_1^\infty-\pi_N^\infty&=
        \alpha\left[(1-a)^2-(c-\eta)\right]
        -
        \alpha q_N\left[(1-a)^2-(c-\eta)\right]\\
        &=
        \alpha(1-q_N)\left[(1-a)^2-(c-\eta)\right].
    \end{aligned}
    \]
    If the prominent firm has nonnegative profit, then $(1-a)^2-(c-\eta)\geq0$.
    Since $1-q_N>0$, the profit difference is nonnegative, and it is strictly positive whenever firm~1's profit is positive.
    Hence an individual non-prominent firm cannot earn strictly more than the prominent firm when the latter has nonnegative profit, so detrimental prominence cannot arise in the infinite-seller limit.
\end{proof}

\section{Auxiliary Lemmas}

\subsection{Nonbinding Salvage-Value Threshold}
\begin{lemma}{Nonbinding Salvage-Value Threshold}
    \label{lem:nonbinding_cgamma}

    At the pricing equilibrium,
    \[
        c_\Gamma(s,\alpha,\eta)<\eta .
    \]
\end{lemma}

\begin{proof}
    We first convert $c_\Gamma(s,\alpha,\eta)<\eta$ into a sign condition, and then decompose the resulting expression into terms with clear signs.
    Let
    \[
        \Delta\equiv
        \frac{dp_1^*}{d\eta}-\frac{dp_2^*}{d\eta},
        \qquad
        \mu\equiv p_2^*-p_1^* .
    \]
    Proposition \ref{prop:limiting_comparative_statics} gives
    $0<dp_2^*/d\eta<dp_1^*/d\eta<1$, so $\Delta>0$.
    Proposition \ref{prop:limiting_price_comparison} gives $\mu>0$ and
    $0<p_1^*-\eta<p_2^*-\eta$.
    The reduced firm-$2$ first-order condition \eqref{eq:foc1_2} gives
    \[
        p_2^*-\eta
        =
        \frac{1-\eta}{2}
        -\frac{\alpha(a-p_2^*)^2}{4q_2}
        <\frac{1}{2} .
    \]

    The proof of Proposition~\ref{prop:limiting_comparative_statics} gives
    \[
        \frac{dq_2}{d\eta}=\alpha\Delta .
    \]
    Therefore, by the definition of $c_\Gamma(s,\alpha,\eta)$,
    $c_\Gamma(s,\alpha,\eta)<\eta$ is equivalent to
    \[
        (1-q_2)
        +2(p_1^*-\eta)\left(\frac{dp_1^*}{d\eta}-1\right)
        -2q_2(p_2^*-\eta)\left(\frac{dp_2^*}{d\eta}-1\right)
        -\frac{dq_2}{d\eta}(p_2^*-\eta)^2
        >0 .
    \]

    We verify this sign condition as follows.
    Using
    $p_1^*-\eta=(p_2^*-\eta)-\mu$,
    $1-\frac{dp_1^*}{d\eta}=1-\frac{dp_2^*}{d\eta}-\Delta$,
    and $dq_2/d\eta=\alpha\Delta$, the left-hand side equals
    \[
    \begin{aligned}
        &(1-q_2)\left[
        1-2(p_2^*-\eta)\left(1-\frac{dp_2^*}{d\eta}\right)
        \right]
        +2\mu\left(1-\frac{dp_2^*}{d\eta}\right)\\
        &\qquad
        +\Delta\left[
        2(p_1^*-\eta)-\alpha(p_2^*-\eta)^2
        \right].
    \end{aligned}
    \]
    The displayed decomposition isolates three terms.
    Since $0<h_2<1$, $1-q_2=\alpha(1-h_2)>0$.
    The first term is strictly positive because $p_2^*-\eta<1/2$ and
    $0<1-dp_2^*/d\eta<1$.
    The second term is strictly positive because $\mu>0$ and
    $1-dp_2^*/d\eta>0$.
    Thus it remains only to prove the bracketed inequality
    \[
        2(p_1^*-\eta)-\alpha(p_2^*-\eta)^2>0 .
    \]

    To prove this last inequality, we bound the price gap $\mu$.
    Set $r\equiv1-a$ and $d\equiv a-p_2^*$.
    Then $r>0$, $d>0$, and $d<1$.
    The firm-$2$ reduced first-order equation \eqref{eq:foc1_2} is
    \[
        p_2^*=\frac{1+\eta}{2}-\frac{\alpha d^2}{4q_2}.
    \]
    Firm $1$'s first-order condition \eqref{eq:limiting_FOC1} gives
    \[
        p_1^*=\frac{1+\eta}{2}-\frac{\alpha d(2r+d)}{4}.
    \]
    The price-gap expression in the proof of Corollary~\ref{cor: welfare_match_probability} can be written as
    \[
        \mu
        =
        \frac{\alpha d(2r+d)}{4}
        -\frac{\alpha d^2}{4q_2}.
    \]
    Since $d(2r+d)>0$ and $\alpha<1$,
    \[
        \mu
        <
        \frac{d(2r+d)}{4}
        -\frac{\alpha d^2}{4q_2}.
    \]
    The same firm-$2$ equation, together with $p_2^*=a-d=1-r-d$, also implies
    \[
        p_2^*-\eta=r+d-\frac{\alpha d^2}{2q_2}.
    \]
    Therefore
    \[
        \frac{p_2^*-\eta}{2}
        -\left[
        \frac{d(2r+d)}{4}
        -\frac{\alpha d^2}{4q_2}
        \right]
        =
        \frac{2r(1-d)+d(2-d)}{4}>0 .
    \]
    Hence $\mu<(p_2^*-\eta)/2$.
    Therefore $p_1^*-\eta=(p_2^*-\eta)-\mu>(p_2^*-\eta)/2$.
    Since $0<p_2^*-\eta<1/2$ and $\alpha<1$,
    \[
        2(p_1^*-\eta)>p_2^*-\eta>\alpha(p_2^*-\eta)^2 .
    \]
    This proves the remaining bracketed inequality.
    Therefore all three terms in the decomposition above are strictly positive, so the target sign condition holds.
    Hence $c_\Gamma(s,\alpha,\eta)<\eta$.
\end{proof}

\subsection{Bounds for the Welfare Comparative Statics}
\begin{lemma}{Auxiliary Equilibrium Bounds}\label{lem:welfare_parameter_bounds}

Under the maintained conditions of Proposition~\ref{prop:limiting_global_deviation},
$$
 \frac12<a<\frac34,\qquad 0<a-p_2^*<\frac5{18},\qquad
 q_2>\frac12+\frac{3(a-p_2^*)}{4}+\frac{3(a-p_2^*)^2}{16},
$$
$$
 0<\frac{\alpha (a-p_2^*)}{q_2}<\frac25,\qquad
 0<1-p_2^*<\frac{19}{36}.
$$
\end{lemma}
\begin{proof}
The restrictions $a>1/2$ and $\eta<3-4a$, together with $\eta\geq0$, imply $1/2<a<3/4$.
The price bounds in Proposition~\ref{lem:limiting_armstrong_style} give $1-a<p_2^*<a$, so $0<a-p_2^*<2a-1<1/2$.
Moreover, $p_1^*>0$ implies
$$
 q_2=1-\alpha+\alpha(p_1^*+a-p_2^*)>\alpha (a-p_2^*).
$$

The reduced firm~2 first-order condition derived in the proof of Proposition~\ref{lem:limiting_armstrong_style} can be rearranged as
\begin{equation}
 \eta=2p_2^*-1+\frac{\alpha(a-p_2^*)^2}{2q_2}.
 \label{eq:welfare_eta_identity}
\end{equation}
Since $\eta\geq0$ and $\alpha (a-p_2^*)/q_2<1$, this equation implies
\begin{equation}
 2-a-p_2^*=1-\eta-(a-p_2^*)+\frac{\alpha (a-p_2^*)^2}{2q_2}<1.
 \label{eq:welfare_intermediate_bound}
\end{equation}

Next, substituting firm~1's pricing equation \eqref{eq:limiting_FOC1} into $q_2=1-\alpha+\alpha(a+p_1^*-p_2^*)$ gives
$$
 q_2=1-\frac\alpha2+\frac{\alpha\eta}{2}
       +\alpha (a-p_2^*)-\frac{\alpha^2(a-p_2^*)(2-a-p_2^*)}4.
$$
After dropping the nonnegative term $\alpha\eta/2$, the right-hand side decreases in $\alpha$, because $a-p_2^*<1/2$.
Using $\alpha<1$ and \eqref{eq:welfare_intermediate_bound}, we obtain
\begin{equation}
 q_2>\frac12+(a-p_2^*)-\frac{(a-p_2^*)(2-a-p_2^*)}4>\frac12+\frac{3(a-p_2^*)}{4}.
 \label{eq:welfare_coarse_q_bound}
\end{equation}

We now use these bounds to obtain a tighter upper bound on $a-p_2^*$.
First, rearranging \eqref{eq:welfare_eta_identity} gives
\[
 (a-p_2^*)-\frac14
 =\frac14-(1-a)-\frac{\eta}{2}
  +\frac{\alpha(a-p_2^*)^2}{4q_2}.
\]
Since $1-a>1/4$ and $\eta\geq0$, the first two terms on the right sum to a negative number and $-\eta/2\leq0$. Thus
\[
 (a-p_2^*)-\frac14<\frac{\alpha(a-p_2^*)^2}{4q_2}
 \leq\frac{(a-p_2^*)^2}{4q_2}
 <\frac{(a-p_2^*)^2}{2+3(a-p_2^*)}.
\]
The second inequality uses $\alpha\leq1$; the third follows from the lower bound on $q_2$ in \eqref{eq:welfare_coarse_q_bound}.
Rearranging gives $8(a-p_2^*)^2+5(a-p_2^*)-2<0$.
The left-hand side increases for $a-p_2^*\geq0$ and equals $1/162>0$ at $a-p_2^*=5/18$; hence $a-p_2^*<5/18$.
Combining this bound with \eqref{eq:welfare_coarse_q_bound} yields
\[
 \frac{\alpha (a-p_2^*)}{q_2}<\frac{(a-p_2^*)}{1/2+3(a-p_2^*)/4}<\frac25.
\]

Using \eqref{eq:welfare_intermediate_bound} and $\alpha(a-p_2^*)/q_2<2/5$, we obtain
\[
 \begin{aligned}
 2-a-p_2^*
 &=1-\eta-(a-p_2^*)
   +\frac{a-p_2^*}{2}\frac{\alpha(a-p_2^*)}{q_2}\\
 &<1-\eta-(a-p_2^*)+\frac{a-p_2^*}{5}
 \leq1-\frac{4(a-p_2^*)}{5},
 \end{aligned}
\]
where the final inequality uses $\eta\geq0$.
Substituting this upper bound for $2-a-p_2^*$ into the first inequality in \eqref{eq:welfare_coarse_q_bound} gives
\[
 \begin{aligned}
 q_2
 &>\frac12+(a-p_2^*)
    -\frac{(a-p_2^*)\left[1-4(a-p_2^*)/5\right]}{4}\\
 &=\frac12+\frac{3(a-p_2^*)}{4}+\frac{(a-p_2^*)^2}{5}\\
 &>\frac12+\frac{3(a-p_2^*)}{4}+\frac{3(a-p_2^*)^2}{16}.
 \end{aligned}
\]
The last inequality uses $1/5>3/16$ and $a-p_2^*>0$.
Finally, rearranging \eqref{eq:welfare_eta_identity} and using $\eta\geq0$ gives
\[
 1-p_2^*
 =\frac12-\frac{\eta}{2}+\frac{\alpha(a-p_2^*)^2}{4q_2}
 \leq\frac12+\frac{\alpha(a-p_2^*)^2}{4q_2}.
\]
Because $a-p_2^*>0$ and $\alpha(a-p_2^*)/q_2<2/5$, the last term satisfies
\[
 \frac{\alpha(a-p_2^*)^2}{4q_2}
 <\frac{a-p_2^*}{4}\cdot\frac25
 =\frac{a-p_2^*}{10}.
\]
Moreover, $1-a>1/4$ implies
\[
 a-p_2^*=(1-p_2^*)-(1-a)<1-p_2^*-\frac14.
\]
Combining these inequalities and rearranging, we obtain
\[
 1-p_2^*<\frac12+\frac{1-p_2^*-1/4}{10}
 \quad\Longrightarrow\quad
 \frac9{10}(1-p_2^*)<\frac{19}{40}
 \quad\Longrightarrow\quad
 1-p_2^*<\frac{19}{36},
\]
completing the proof.
\end{proof}

\section{Sufficient Conditions for Full-Game Equilibrium}\label{appendix:sufficient_conditions_full_equilibrium}
\subsection{Global Deviations}\label{appendix:global_deviation_condition}
Throughout this appendix section, $(p_1^*,p_2^*)$, $h_2$, and $q_2$ refer to
the limiting pricing equilibrium and the associated search cutoff and exposure
defined in the main text.  The following result gives simple sufficient
restrictions under which this price pair is immune to deviations outside the
active branch.

\begin{proposition}{Sufficient Conditions for Equilibrium}
    \label{prop:limiting_global_deviation}

    Suppose \eqref{eq:armstrong_style_search_cost}--\eqref{eq:armstrong_style_salvage_conditions} hold and, in addition,
    \begin{equation}
        \eta<3-4a,
        \label{eq:armstrong_style_extra_primitive}
    \end{equation}
    and
    \begin{equation}
        0<c-\eta\leq (p_1^*-\eta)^2 .
        \label{eq:armstrong_style_equilibrium_return_loss_bound}
    \end{equation}
    Then the price pair characterized in Proposition \ref{lem:limiting_armstrong_style} is a full-game Perfect Bayesian equilibrium.

    Condition~\eqref{eq:armstrong_style_equilibrium_return_loss_bound} can be implied by the following condition on the primitives:
    \begin{equation}
        0<c-\eta\leq\frac{[1-\eta-\alpha(a-1/2)]^2}{4}.
        \label{eq:primitive_return_loss_bound}
    \end{equation}
\end{proposition}

\begin{proof}
    We first show that \eqref{eq:primitive_return_loss_bound} implies \eqref{eq:armstrong_style_equilibrium_return_loss_bound}.
    The right-hand side of firm~1's pricing equation \eqref{eq:limiting_FOC1} increases in $p_2^*$.
    Since $p_2^*>1-a$ by Proposition~\ref{lem:limiting_armstrong_style},
    \[
        p_1^*>\frac{1+\eta-\alpha(a-1/2)}{2},
        \qquad
        p_1^*-\eta>\frac{1-\eta-\alpha(a-1/2)}{2}>0.
    \]
    The last inequality follows from $\eta<1-a$, $a>1/2$, and $\alpha<1$.
    Squaring the positive terms shows that the upper bound in \eqref{eq:primitive_return_loss_bound} is smaller than $(p_1^*-\eta)^2$.
    Thus \eqref{eq:primitive_return_loss_bound} implies \eqref{eq:armstrong_style_equilibrium_return_loss_bound}.

    Recall that $\rho=c-\eta$.
    Proposition \ref{lem:limiting_armstrong_style} and Proposition \ref{prop:limiting_price_comparison} give
    $1-a<p_1^*<p_2^*<(1+\eta)/2<a$ and, because $\eta<1-a$, $\eta<p_1^*$.
    Hence
    $$
        0<p_1^*-\eta<p_2^*-\eta<1.
    $$
    The maintained return-loss bound \eqref{eq:armstrong_style_equilibrium_return_loss_bound} gives
    $$
        \pi_1^*=\alpha\left[(p_1^*-\eta)^2-\rho\right]\geq0,
        \qquad
        \pi_2^*=\alpha q_2\left[(p_2^*-\eta)^2-\rho\right]>0 .
    $$

    Since $\eta\geq0$ and $\eta<p_1^*<1$, we have
    $$
        c\leq\eta+(p_1^*-\eta)^2<p_1^*.
    $$
    The second inequality follows from
    $0<p_1^*-\eta<1$ and therefore
    $(p_1^*-\eta)^2<p_1^*-\eta$.
    Together with Proposition \ref{lem:limiting_armstrong_style} and Proposition
    \ref{prop:limiting_price_comparison}, this gives
    $$
        c<p_1^*<p_2^*<a.
    $$
    Thus the equilibrium prices are interior relative to the active-branch
    pricing interval considered below.

    First consider low-price and inactive deviations. 
    If $D$ and $S$ denote retained demand and initial sales, deviation profit can be written as
    $$
        (p-c)D+(\eta-c)(S-D)\leq0,
    $$
    because $p<c$, $c>\eta$, and $S\geq D$.  By Corollary
    \ref{prop:limiting_consumer_rule},
    when $p\geq1$, even a matched consumer does not buy the product under the strict-action convention. 
    A deviation to such a price therefore gives zero sales and hence zero profit.
    Since equilibrium profits are nonnegative, neither low-price nor inactive deviations are profitable.

    For the remaining deviation analysis, set $d\equiv a-p_2^*>0$.

    Consider deviations with $p\in[c,1-d]$.
    For either firm, $p\leq1-d$ implies $a+p-p_2^*\leq1$ and $a+p-p_2^*=d+p>p$, so these deviations remain in the active branch.
    As shown in the derivation of \eqref{eq:limiting_FOC1}--\eqref{eq:limiting_FOC2}, the active-branch payoff functions are strictly concave in own price.
    Since $c<p_1^*<p_2^*<1-d$, the pricing equilibrium is the interior maximizer on this interval.  
    Hence no deviation in $[c,1-d]$ is profitable.

    It remains to consider high-price deviations $p\in[1-d,1)$.
    Write $q=1-p\in(0,d]$.
    If firm $1$ deviates to $p=1-q$, every consumer who has bought product $1$ searches firm $2$, because $a+p-p_2^*=1+d-q\geq1$.
    If firm $2$ is matched, she also buys it: product $1$'s fallback surplus is at most $q$, so $p_2^*+v_0(\varepsilon_1-p)\leq p_2^*+q\leq p_2^*+d=a<1$.
    Conditional on product $1$ being matched, firm $1$'s retained demand is
    $$
        (1-\alpha)q
        +\alpha\int_0^q(p_2^*+u)\,du
        =
        \left(1-\alpha+\alpha p_2^*\right)q+\frac{\alpha q^2}{2}.
    $$
    The first term is the state in which firm $2$ is unmatched; in the second term $u=\varepsilon_1-(1-q)$ is product $1$'s positive net surplus.
    Given $u\in[0,q]$, firm $1$ obtains the consumer against a matched firm $2$ when $\varepsilon_2-p_2^*\leq u$, an event of probability $p_2^*+u$.
    Firm $1$'s initial sales remain $\alpha$.  Hence its high-price deviation profit is
    $$
        \widehat\pi_1(q)
        =
        \alpha\left[
        (1-\eta-q)\left(\left(1-\alpha+\alpha p_2^*\right)q+\frac{\alpha q^2}{2}\right)-\rho
        \right]
    $$
    for $q\in(0,d]$.

    If firm $2$ deviates to $p=1-q$, Corollary
    \ref{prop:limiting_consumer_rule} implies that consumers still decide
    whether to search firm $2$ using the equilibrium belief $p_2^*$.
    A consumer searches iff
    $$
        v_0(\varepsilon_1-p_1^*)<a-p_2^*=d .
    $$
    After search, however, she observes the deviating price $1-q$.
    By Corollary \ref{prop:limiting_consumer_rule}, she buys product $2$ iff
    $$
        1-q+v_0(\varepsilon_1-p_1^*)<1,
        \quad\text{or equivalently}\quad
        v_0(\varepsilon_1-p_1^*)<q .
    $$
    Since $q\leq d$, every consumer who would buy at the deviating price is among those who search.
    Conditional on product $2$ being matched, firm $2$'s retained demand is
    $$
        (1-\alpha)q
        +\alpha\left[p_1^*q+\int_0^q(q-u)\,du\right]
        =
        \left(1-\alpha+\alpha p_1^*\right)q+\frac{\alpha q^2}{2}.
    $$
    The first term is the no-product state arising when firm $1$ is unmatched.
    In the bracket, consumers with $\varepsilon_1<p_1^*$ have fallback surplus zero; their mass is $p_1^*$, and product $2$ is kept with probability $q$.
    For a positive fallback surplus $u\in[0,q]$, product $2$ is kept when $\varepsilon_2-(1-q)\geq u$, an event of probability $q-u$.
    The same purchase condition gives firm $2$'s initial sales:
    $$
        \alpha\left[(1-\alpha)+\alpha(p_1^*+q)\right].
    $$
    Thus
    \[
        \widehat\pi_2(q)
        =
        \alpha\left[
        (1-\eta-q)\left(\left(1-\alpha+\alpha p_1^*\right)q+\frac{\alpha q^2}{2}\right)
        -\rho\left(1-\alpha+\alpha p_1^*+\alpha q\right)
        \right].
    \]
    Differentiating gives
    \[
    \begin{aligned}
        \frac{1}{\alpha}\widehat\pi_1'(q)
        &=
        \left(1-\alpha+\alpha p_2^*\right)(1-\eta-2q)
        +\alpha q(1-\eta)-\frac{3\alpha q^2}{2},\\
        \frac{1}{\alpha}\widehat\pi_2'(q)
        &=
        \left(1-\alpha+\alpha p_1^*\right)(1-\eta-2q)
        +\alpha q(1-\eta)-\frac{3\alpha q^2}{2}
        -\alpha\rho .
    \end{aligned}
    \]
    Each derivative is a concave quadratic in $q$.  A concave function lies
    above the chord joining any two points, so it is enough to show that each
    derivative is positive at $q=0$ and at $q=d$.

    Firm $1$'s derivative at $q=0$ is
    \[
        \left(1-\alpha+\alpha p_2^*\right)(1-\eta)>0.
    \]
    At $q=d$ it equals
    \[
        \left(1-\alpha+\alpha a\right)(1-\eta-2d)+\frac{\alpha(d)^2}{2},
    \]
    because $1-\alpha+\alpha a=(1-\alpha+\alpha p_2^*)+\alpha d$.
    Since $p_2^*>1-a$, we have $d=a-p_2^*<2a-1$.  Therefore
    \[
        1-\eta-2d>3-4a-\eta>0
    \]
    by \eqref{eq:armstrong_style_extra_primitive}.  The endpoint derivative is
    positive.

    Firm $2$'s derivative is decreasing in $\rho$.  Hence it suffices to check
    the larger value $\rho=(p_1^*-\eta)^2$.  At $q=0$,
    \[
    \begin{aligned}
        \frac{1}{\alpha}\widehat\pi_2'(0)
        &=
        \left(1-\alpha+\alpha p_1^*\right)(1-\eta)
        -\alpha(p_1^*-\eta)^2\\
        &>
        (1-\eta)(1-\alpha+\alpha\eta)>0.
    \end{aligned}
    \]
    The strict inequality follows from $0<p_1^*-\eta<1-\eta$.
    At $q=d$ and $\rho=(p_1^*-\eta)^2$, the endpoint derivative is
    \[
        q_2(1-\eta-2d)+\frac{\alpha(d)^2}{2}-\alpha(p_1^*-\eta)^2 .
    \]
    The firm-$2$ first-order condition implies
    \[
        q_2(2a-1-\eta-2d)=-\frac{\alpha(d)^2}{2},
    \]
    so this endpoint expression equals
    \[
        2(1-a)q_2-\alpha(p_1^*-\eta)^2 .
    \]
    Since $p_1^*>1-a$ and $p_1^*<p_2^*<(1+\eta)/2$,
    \[
        q_2>1-\alpha+\frac{\alpha(1-\eta)}{2},
        \qquad
        p_1^*-\eta<\frac{1-\eta}{2}.
    \]

    The required endpoint inequality follows directly from
    \eqref{eq:armstrong_style_extra_primitive}:
    since $2(1-a)>(1+\eta)/2\geq(1-\eta)/2$ and $0<\alpha<1$,
    \[
    \begin{aligned}
    2(1-a)\left(1-\alpha+\frac{\alpha(1-\eta)}{2}\right)
    &>\frac{1-\eta}{2}\left(1-\alpha+\frac{\alpha(1-\eta)}{2}\right)\\
    &>\frac{\alpha(1-\eta)^2}{4}.
    \end{aligned}
    \]
    Thus firm~ $2$'s endpoint derivative is positive at $\rho=(p_1^*-\eta)^2$.

     Since every admissible $\rho$ is no larger than $(p_1^*-\eta)^2$, the same
    endpoint signs hold for all admissible return losses.

    By concavity, $\widehat\pi_i'(q)>0$ for all $q\in[0,d]$ and $i=1,2$.
    Thus high-price deviation profits are increasing in $q$, or equivalently
    decreasing in the deviating price $p=1-q$.
    Their maxima on the high-price interval occur at $q=d$, namely
    $p=1-d$.
    At this boundary, the high-price formulas above coincide with the
    active-branch payoff formulas: the relevant purchase cutoff is exactly at
    the top of the match-value support.
    Since the active-branch payoff is strictly concave and maximized at the
    interior price $p_i^*<1-d$, that boundary payoff is no larger than the
    equilibrium payoff.
    Hence no high-price deviation is profitable, and no firm has a profitable global deviation.
\end{proof}

\end{document}